\documentclass[onecolumn,pra, superscriptaddress,longbibliography,
12pt,tightenlines,aps]{revtex4-2}

\pdfoutput=1

\usepackage{graphicx}
\usepackage{bm}
\usepackage{amsmath}
\usepackage{environ}
\usepackage{appendix}
\usepackage{physics}
\usepackage{xfrac}
\usepackage[shortlabels]{enumitem}
\usepackage[nopar]{lipsum}
\usepackage{optidef}
\usepackage{stmaryrd}
\usepackage{xcolor}
\usepackage{bm}

\usepackage{amsthm}
\usepackage{amssymb}
\usepackage{amsfonts}
\usepackage{slashed}
\usepackage{bbm}
\usepackage{latexsym,epsfig,bbm}
\usepackage[makeroom]{cancel}
\usepackage{hyperref}
\hypersetup{colorlinks = true, linkcolor = [rgb]{0.19411,0.51882,0.667058}, urlcolor = [rgb]{0.125490,0.29542,0.1647058}, citecolor = [rgb]{0.75882,0.37411,0.14117}}

\usepackage{color}
\definecolor{blue}{rgb}{0,0.2,1}

\definecolor{red}{rgb}{0.9,0,0}

\newtheorem{theorem}{Theorem}
\newtheorem{lemma}[theorem]{Lemma}
\newtheorem{definition}[theorem]{Definition}

\newtheorem{corollary}[theorem]{Corollary}
\newtheorem{remark}{Remark}

\usepackage{mathtools}

\newcommand{\poly} {\operatorname{poly}}

\newcommand{\Sn}{\mathfrak{S}_n}
\NewDocumentCommand{\End}{o}{\mathrm{End}^{\Sn}}
\newcommand{\cH}{\mathcal{H}}
\newcommand{\HS}{\mathcal{H}}
\newcommand{\eps}{\varepsilon}

\newcommand{\N}{\mathbb{N}}

\allowdisplaybreaks
\newcommand{\comment}[1]{}

\begin{document}

 \title{\texorpdfstring{An invitation to the sample complexity of\\quantum hypothesis testing}{An invitation to the sample complexity of quantum hypothesis testing}}

\author{Hao-Chung~Cheng}
\affiliation{Department of Electrical Engineering,
National Taiwan University, Taipei 106, Taiwan (R.O.C.)}
\affiliation{Hon Hai (Foxconn) Quantum Computing Center, New Taipei City 236, Taiwan (R.O.C.)}
\affiliation{Physics Division, National Center for Theoretical Sciences, Taipei 106, Taiwan (R.O.C.)}

\author{Nilanjana~Datta}
\affiliation{Department of Applied Mathematics and Theoretical Physics,
University of Cambridge, Wilberforce Road, Cambridge, CB3 0WA, United Kingdom}

\author{Nana~Liu}
\email{nana.liu@quantumlah.org}
\affiliation{Institute of Natural Sciences, School of Mathematical Sciences, MOE-LSC,
Shanghai Jiao Tong University, Shanghai, 200240, P. R. China}
\affiliation{Shanghai Artificial Intelligence Laboratory, Shanghai, China}
\affiliation{University of Michigan-Shanghai Jiao Tong University Joint Institute, Shanghai 200240, China}

\author{Theshani~Nuradha}
\affiliation{School of Electrical and Computer Engineering, Cornell University, Ithaca, New York 14850, USA}

\author{Robert~Salzmann}
\affiliation{Department of Applied Mathematics and Theoretical Physics,
University of Cambridge, Wilberforce Road, Cambridge, CB3 0WA, United Kingdom}
\affiliation{Universit\'e de  Lyon, Inria, ENS Lyon, UCBL, LIP, F-69342, Lyon Cedex 07, France}

\author{Mark~M.~Wilde}
\affiliation{School of Electrical and Computer Engineering, Cornell University, Ithaca, New York 14850, USA}

\date{\today}

\begin{abstract} 
Quantum hypothesis testing has been traditionally studied from the information-theoretic perspective, wherein one is interested in the optimal decay rate of error probabilities as a function of the number of samples of an unknown state. In this paper, we study the sample complexity of quantum hypothesis testing, wherein the goal is to determine the minimum number of samples needed to reach a desired error probability. By making use of the wealth of knowledge that already exists in the literature on quantum hypothesis testing, we characterize the sample complexity of binary quantum hypothesis testing in the symmetric and asymmetric settings, and we provide bounds on the sample complexity of multiple quantum hypothesis testing. In more detail, we prove that the sample complexity of symmetric binary quantum hypothesis testing depends logarithmically on the inverse error probability and inversely on the negative logarithm of the fidelity. As a counterpart of the quantum Stein's lemma, we also find that the sample complexity of asymmetric binary quantum hypothesis testing depends logarithmically on the inverse type~II error probability and inversely on the quantum relative entropy, provided that the type~II error probability is sufficiently small. We then provide lower and upper bounds on the sample complexity of multiple quantum hypothesis testing, with it remaining an intriguing open question to improve these bounds.
The final part of our paper outlines and reviews how sample complexity of quantum hypothesis testing is relevant to a broad swathe of research areas and can enhance understanding of many fundamental concepts, including quantum algorithms for simulation and search, quantum learning and classification, and foundations of quantum mechanics. As such, we view our paper as an invitation to researchers coming from different communities to study and contribute to the problem of sample complexity of quantum hypothesis testing, and we outline a number of open directions for future research.
\end{abstract}

\maketitle

\setcounter{tocdepth}{2} 
\tableofcontents

\section{Introduction}

\label{sec:intro}

Distinguishing between various possibilities is fundamental to the scientific method and the process of discovery and categorization. As such, mathematical methods underlying distinguishability have been studied for a long time~\cite{NP33} and are now highly developed and fall under the general framework of hypothesis testing~\cite{casella2021statistical}. There are various ways of formulating the hypothesis testing problem, and due to the core role of distinguishability in a variety of fields, it has also been influential in many domains beyond mathematical statistics.

Here we are interested in quantum hypothesis testing, which has its origins in early work in the field of quantum information theory~\cite{Hel67,Hel69,Hol73}. Quantum hypothesis testing is an important foundational topic in the sense that it merges quantum mechanics and mathematical statistics and is applicable to basic quantum physics experiments. Moreover, in addition to applications in communication, information processing, and computation~\cite{bae2015quantum}, it has even found use in understanding foundational aspects of quantum mechanics~\cite{pusey2012reality}. In the most basic version of the problem, a quantum system is prepared in one of two possible states (density operators), denoted by $\rho$ and $\sigma$,  and it is the goal of the distinguisher, who does not know \textit{a priori} which state was prepared, to determine the identity of the unknown state. The distinguisher ideally wishes to minimize the probability of making an incorrect decision.

The task of quantum hypothesis testing becomes more interesting when multiple copies of the unknown state are provided to the distinguisher. In the case that $n\in \mathbb{N}$ copies (or samples) are provided, the states are then denoted by $\rho^{\otimes n}$ and $\sigma^{\otimes n}$ (i.e., tensor-power states). Intuitively,  extra samples are helpful for decreasing the error probability when trying to determine the unknown state. In fact, it is well known in the information theory literature that the error probability decreases exponentially fast in the number of samples provided that $\rho \neq \sigma$~\cite{HP91,ON00,ACM+07,NS09,Nag06,Hay07}.

Quantum hypothesis testing can be generalized to the case in which there are multiple hypotheses~\cite{YKL75}. In this case, a value $x$ is selected from a finite alphabet (say, $\mathcal{X}$) with a prior probability distribution $\{p(x)\}_{x \in \mathcal{X}}$. and a state $\rho_x^{\otimes n}$ is prepared, where $n\in \mathbb{N}$. It is then the goal of the distinguisher to determine the value of $x$ with as small an error probability as possible. It is also known in this case that the error probability generally decays exponentially fast with $n$~\cite{Li16}. Alternatively, one can turn this problem around and fix a constant upper bound on the desired error probability and demand that the distinguisher determines the value of $x$ with as few samples as possible (minimize $n$) while meeting the error probability constraint.

On the one hand, it is a common goal in the information theory literature to determine the optimal exponential rate of decay of the error probability, known as an error exponent. This is typically studied in the large $n$ limit, but more recently researchers have sought out more refined statements~\cite{MTMH13,Li14,DPR16}. Furthermore, this line of study has direct links with communication theory and allows for making statements about error exponents when communicating messages over a quantum channel~\cite{Hay07,TV15,DTW16,TBR16,CHT17,CTT2017,Hao17,WTB17}.

On the other hand, it is more common in the  algorithms and machine learning literature to consider a notion called sample complexity~\cite{arunachalam2018optimal,Canonne22}; that is, if we fix the error probability to be no larger than a constant $\varepsilon \in [0,1]$, what is the smallest value of $n$ (minimum number of samples) needed to ensure that the distinguisher can figure out the unknown state with an error probability no larger than $\varepsilon$? This latter notion is useful because it indicates how long one needs to wait in order to make a decision with a desired performance, and it is more compatible with the runtime of a probabilistic or quantum algorithm that solves a problem of interest. 

In this paper, we provide a systematic study of the sample complexity of quantum hypothesis testing, when there are either two or more hypotheses to distinguish. More broadly, we consider our paper to be an invitation to the wider community to study more general questions regarding the sample complexity of quantum hypothesis testing, which touches the interface between statistics, quantum information, quantum algorithms, and learning theory.
We divide our paper into three main parts:
\begin{enumerate}
    \item 
The first part is mainly pedagogical, which we have aimed to make broadly accessible.
We provide some background on quantum information-theoretic quantities of interest used throughout our paper, and we also review symmetric binary, asymmetric binary, and multiple quantum hypothesis testing. Here we also include a new finding that there is an efficient, polynomial-time algorithm for computing the optimal error probabilities in these three different settings.

\item The second part contains our main technical results. Here we precisely define various instances of the problem of sample complexity for quantum hypothesis testing, and we prove bounds for these various sample complexities. We also discuss distinctions between sample complexity in the classical and quantum cases. 

\item The last part includes a selection of applications, which serve only as a springboard for open questions that the wider community can begin to study in more depth.

\end{enumerate}

\section{Background}

In this section, we establish some notation and recall various quantities of interest used throughout the rest of our paper.

Let $\mathbb{N}\coloneqq \{1,2,\ldots\}$.
Throughout our paper, we let $\rho$ and $\sigma$ denote quantum states, which are positive semi-definite operators acting on a separable Hilbert space and with trace equal to one. Note that separable Hilbert spaces are those that are infinite-dimensional and are spanned by a countable orthonormal basis~\cite[Proposition 1.12]{HeinosaariZ11}.
Let  $I$  denote the identity operator.
For every bounded operator $A$ and $p \geq 1$, we define the Schatten $p$-norm as
\begin{align} \label{eq:Schatten}
	\left\|A\right\|_p \coloneqq  \left( \Tr\!\left[ \left( A^\dagger A \right)^{p/2} \right] \right)^{1/p}.
\end{align}
Due to the variational characterization of the trace norm, 
\begin{align} \label{eq:variational_trace-norm}
	\left\|A\right\|_1 = \max_{U} \mathrm{Re}\!\left[ \Tr\!\left[ AU\right] \right],
\end{align}
where the optimization is over every unitary $U$.

\subsection{Quantum information-theoretic quantities}

\begin{definition}[Fidelities] \label{definition:fidelities}
	Let $\rho$ and $\sigma$ be quantum states.
	\begin{enumerate}[1)]
		\item The quantum fidelity is defined as~\cite{Uhl76}
		\begin{align} \label{eq:fidelity}
			F(\rho,\sigma) \coloneqq  \left\|\sqrt{\rho} \sqrt{\sigma} \right\|_1^2.
		\end{align}
		Note that $F(\rho,\sigma) \in [0,1]$, it is equal to one if and only if $\rho=\sigma$, and it is equal to zero if and only if $\rho $ is orthogonal to $\sigma$ (i.e., $\rho\sigma=0$).
		
		\item The Holevo fidelity  is defined as~\cite{Holevo1972fid}
		\begin{align} \label{eq:Holevo_fidelity}
			F_\mathrm{H}(\rho,\sigma) \coloneqq  \left(\Tr\!\left[\sqrt{\rho} \sqrt{\sigma} \right]\right)^2.
		\end{align}
	Note that $F_\mathrm{H}(\rho,\sigma) \in [0,1]$, it is equal to one if and only if $\rho=\sigma$, and it is equal to zero if and only if $\rho $ is orthogonal to $\sigma$ (i.e., $\rho\sigma=0$).
	\end{enumerate}
\end{definition}

A divergence $\bm{D}(\rho\|\sigma)$ is a function of two quantum states $\rho$ and $\sigma$, and we say that it obeys the data-processing inequality if the following holds for all states $\rho$ and $\sigma$ and every channel $\mathcal{N}$:
\begin{equation}
    \bm{D}(\rho\|\sigma) \geq \bm{D}(\mathcal{N}(\rho)\|\mathcal{N}(\sigma)).
\end{equation}

\begin{definition}[Distances and divergences] \label{definition:distances}
	Let $\rho$ and $\sigma$ be quantum states.
	\begin{enumerate}[1)]
		\item The normalized trace distance is defined as 
		\begin{align}
			 \frac12 \left\|\rho-\sigma \right\|_1.
		\end{align}

		\item The Bures distance is defined as~\cite{Hel67b, Bur69}
		\begin{align}
			d_\mathrm{B}(\rho,\sigma) 
			&\coloneqq  \min_{U} \left\| \sqrt{\rho} - U \sqrt{\sigma} \right\|_2
			\\
			&= \min_{U} \sqrt{ 2\left( 1 - \mathrm{Re}\!\left[ \Tr\!\left[ \sqrt{\rho}\sqrt{\sigma} U \right] \right] \right) }
			\\
			&\overset{\eqref{eq:variational_trace-norm}}{=}\sqrt{2\left(1 - \sqrt{F}(\rho,\sigma)\right)},
		\end{align}
		where the optimization is over every unitary $U$.
		
		\item The quantum Hellinger distance is defined as~\cite{jencova2004,LZ04}
		\begin{align}
			d_\mathrm{H}(\rho,\sigma) \coloneqq  \left\|\sqrt{\rho}-\sqrt{\sigma}\right\|_2
			=\sqrt{2\left(1 - \sqrt{F_\mathrm{H}}(\rho,\sigma)\right)}.
		\end{align}

        \item The Petz--R\'enyi divergence of order $\alpha \in (0,1) \cup (1,\infty)$ is defined as~\cite{Pet86}
        \begin{align}
            D_\alpha(\rho\|\sigma) & \coloneqq \frac{1}{\alpha-1} \ln Q_\alpha(\rho\|\sigma) \label{eq:petz-renyi-div} \\
            {\hbox{where}} \quad
            Q_\alpha(A\|B) & \coloneqq \lim_{\varepsilon \to 0^+} \Tr[A^\alpha (B+\varepsilon I)^{1-\alpha}], \qquad \forall A,B\geq 0.
        \end{align}
        It obeys the data-processing inequality for all $\alpha \in (0,1)\cup(1,2]$~\cite{Pet86}. Note also that $D_{1/2}(\rho\|\sigma) = -\ln F_{\mathrm{H}}(\rho,\sigma)$.

        \item The sandwiched R\'enyi divergence of order $\alpha \in (0,1) \cup (1,\infty)$ is defined as~\cite{MDS+13,wilde_strong_2014}
        \begin{align}
            \widetilde{D}_\alpha(\rho\|\sigma) & \coloneqq \frac{1}{\alpha-1} \ln Q_\alpha(\rho\|\sigma) \label{eq:sandwiched-renyi-div} \\
            {\hbox{where}} \quad \widetilde{Q}_\alpha(A\|B) & \coloneqq \lim_{\varepsilon\to 0^+} \Tr[\left(A^{\frac{1}{2}} (B+\varepsilon I)^{\frac{1-\alpha}{\alpha}} A^{\frac{1}{2}} \right)^\alpha], \qquad \forall A,B\geq 0.
        \end{align}
        It obeys the data-processing inequality for all $\alpha \in [1/2,1)\cup(1,\infty)$~\cite{FL13,Wilde_2018}. Note also that $\widetilde{D}_{1/2}(\rho\|\sigma) = -\ln F(\rho,\sigma)$.

        \item The quantum relative entropy is defined as~\cite{Ume62}
        \begin{equation}
            D(\rho\|\sigma) \coloneqq \lim_{\varepsilon \to 0^+} \Tr[\rho(\ln \rho - \ln (\sigma + \varepsilon I))],
        \end{equation}
        and it obeys the data-processing inequality~\cite{lindblad1975completely}.
  
		\item The Chernoff divergence is defined as~\cite{ACM+07}
		\begin{align}
			C(\rho\|\sigma) &\coloneqq  -\ln Q_{\min}(\rho\|\sigma), \\
			{\hbox{where}} \quad Q_{\min}(A\|B )&\coloneqq  \min_{s \in [0,1] }Q_s(A\|B), \quad \forall A,B\geq 0. 
		\end{align}
	\end{enumerate}
\end{definition}

If the quantum states $\rho$ and $\sigma$ commute, i.e., $\rho \sigma = \sigma \rho$, then they share a common eigenbasis.
Let $\{P(x)\}_{x}$ and $\{Q(x)\}_{x}$ be the sets of eigenvalues of $\rho$ and $\sigma$, respectively.
This situation is known as  the \emph{classical scenario} because the above-defined divergences reduce to their corresponding classical divergences between probability mass functions $P$ and $Q$.
For instance, in this case, the normalized trace distance is equal to the classical total variation distance, i.e., $\frac{1}{2}\left\|\rho -\sigma\right\|_1 = \frac12 \sum_x |P(x)-Q(x)|$. Furthermore,
both the quantum fidelity and Holevo fidelity reduce to the squared Bhattacharyya coefficient, i.e.,
$F(\rho,\sigma) = F_{\mathrm{H}}(\rho,\sigma) = \big( \sum_x \sqrt{P(x)Q(x)} \big)^2$.
Likewise, the Bures distance and quantum Hellinger distance correspond to the classical Hellinger distance (up to a factor $\sqrt{2}$).
Both the Petz--R\'enyi  and sandwiched R\'enyi divergences reduce to the classical R\'enyi divergence~\cite{Ren62}, and the quantum relative entropy reduces to the Kullback--Leibler divergence~\cite{KLdiv51}.

In the general case in which $\rho$ and $\sigma$ do not commute, $F(\rho,\sigma) \neq F_{\mathrm{H}}(\rho,\sigma)$ and 
$D_{\alpha}(\rho\Vert \sigma) \neq \widetilde{D}_{\alpha}(\rho\Vert \sigma)$ for all $\alpha \in (0,1)\cup(1,\infty)$ 
if they are finite
\cite[Theorem 2.1]{Hia94}.
This explains why there are more divergences in the quantum scenario than there are in the classical scenario.
We refer readers to Lemma~\ref{lemma:Fact} in Appendix~\ref{sec:proofs-divergences} for the relations between those quantities.

\subsection{Quantum hypothesis testing}

\label{sec:qht-review}

Let us first review binary classical hypothesis testing, which is a special case of quantum hypothesis testing.
Suppose a classical system is modeled by a random variable $Y$, which is distributed according to a distribution $P$ under the null hypothesis and  according to a distribution $Q$ under the alternative hypothesis.
The goal of a distinguisher is to guess which hypothesis is correct.
To do so, the distinguisher may take $n$ independent and identical samples of $Y$ and apply a (possibly randomized) test, which is mathematically described by a conditional distribution $P_{0|Y^n}^{(n)}$ and $P_{1|Y^n}^{(n)} = 1 - P_{0|Y^n}^{(n)}$, where $0$ and $1$ correspond to guessing the null hypothesis $P$ and alternative hypothesis $Q$, respectively.
There are two kinds of errors that can occur, and performance metrics can be constructed based on the probabilities of these errors.  
The probability of guessing $Q$ under null hypothesis $P$ is then given by 
the expectation $\mathbb{E}_{Y^n\sim P^{\otimes n}}\!\left[P_{1|Y^n}^{(n)}\right]$,
and the probability of guessing $P$ under the alternative hypothesis $Q$ is given by 
$\mathbb{E}_{Y^n\sim Q^{\otimes n}}\!\left[P_{0|Y^n}^{(n)}\right]$.

Now let us move on to the scenario of quantum hypothesis testing, in which the underlying system is  quantum mechanical.
Instead of being modeled by a probability distribution, a quantum system is prepared in a quantum state, which is modeled by a density operator.
In the case of two hypotheses (i.e., two states $\rho$ and $\sigma$), the system is either prepared in the state $\rho^{\otimes n}$ or $\sigma^{\otimes n}$, where $n \in \mathbb{N}$.
As in the classical scenario, the goal of a distinguisher is to guess which state was prepared, doing so by means of a quantum measurement with two outcomes. Mathematically, such a measurement is described by two operators $\Lambda^{(n)}_\rho$ and $\Lambda^{(n)}_\sigma$ satisfying $\Lambda^{(n)}_\rho,\Lambda^{(n)}_\sigma \geq 0$ and $\Lambda^{(n)}_\rho+\Lambda^{(n)}_\sigma = I^{\otimes n}$, where the outcome~$\Lambda^{(n)}_\rho$ is associated with guessing $\rho$ and $\Lambda^{(n)}_\sigma$ is associated with guessing $\sigma$. 
The probability of guessing $\sigma$ when the prepared state is $\rho$ is equal to $\Tr[\Lambda^{(n)}_\sigma \rho^{\otimes n}]$, and the probability of guessing $\rho$ when the prepared state is $\sigma$ is equal to $\Tr[\Lambda^{(n)}_\rho \sigma^{\otimes n}]$.
If $\rho$ and $\sigma$ commute, then without loss of generality the distinguisher can choose a measurement $(\Lambda_{\rho}^{(n)}, \Lambda_{\sigma}^{(n)})$ sharing a common eigenbasis with $\rho$ and $\sigma$~\cite[Remark 3]{Bus12}.
In such a case, quantum hypothesis testing reduces to classical hypothesis testing, comparing the distribution resulting from the eigenvalues of $\rho$ against that resulting from the eigenvalues of $\sigma$.

There are two scenarios of interest here, called the symmetric and asymmetric settings, reviewed in more detail in Sections~\ref{sec:sym-bin-QHT} and \ref{sec:asym-bin-QHT}, respectively. We then move on to reviewing multiple hypothesis testing in Section~\ref{sec:multiple-QHT}.

\subsubsection{Symmetric binary quantum hypothesis testing}

\label{sec:sym-bin-QHT}

In the setting of symmetric binary quantum hypothesis testing, we suppose that there is a prior probability $p \in (0,1)$ associated with preparing the state $\rho$, and there is a prior probability $q \equiv 1-p$ associated with preparing the state $\sigma$. Indeed, the unknown state is prepared by flipping a coin, with the probability of heads being $p$ and the probability of tails being $q$. If the outcome of the coin flip is heads, then $n$ quantum systems are prepared in the state $\rho^{\otimes n}$, and if the outcome is tails, then the $n$ quantum systems are prepared in the state $\sigma^{\otimes n}$.  Thus, the expected error probability in this experiment is as follows:
\begin{align}
    p_{e,\Lambda^{(n)}_\rho}(p, \rho,q,\sigma, n) & \coloneqq   p \Tr[\Lambda^{(n)}_\sigma \rho^{\otimes n}] + q \Tr[\Lambda^{(n)}_\rho \sigma^{\otimes n}] \\
    & = p \Tr[(I^{\otimes n} - \Lambda^{(n)}_\rho) \rho^{\otimes n}] + q \Tr[\Lambda^{(n)}_\rho \sigma^{\otimes n}],
    \label{eq:err-prob-fixed-meas}
\end{align}
where the second equality follows because $\Lambda^{(n)}_\sigma = I^{\otimes n} - \Lambda^{(n)}_\rho$ and furthermore indicates that the error probability can be written as a function of only the first measurement operator~$ \Lambda^{(n)}_\rho$.

Given $p$ and descriptions of the states $\rho$ and $\sigma$, the distinguisher can minimize the error-probability expression in~\eqref{eq:err-prob-fixed-meas} over all possible measurements.
The Helstrom--Holevo theorem~\cite{Hel69, Hol73} states that the optimal error probability $p_e(p, \rho,q,\sigma,n)$ of hypothesis testing is as follows:
\begin{align}
	\label{eq:Helstrom1}
	p_e(p, \rho,q,\sigma,n) &\coloneqq  \inf_{\Lambda^{(n)}} p_{e,\Lambda^{(n)}}(p, \rho,q,\sigma, n) \\
 & = \inf_{\Lambda^{(n)}} \left\{ p \Tr\!\left[ (I^{\otimes n}-\Lambda^{(n)})\rho^{\otimes n} \right] + q \Tr[\Lambda^{(n)} \sigma^{\otimes n}] : 0\leq \Lambda^{(n)} \leq I^{\otimes n}  \right\}
	\\
	\label{eq:Helstrom2}	
	&= \frac12 \left( 1 - \left\| p \rho^{\otimes n} - q \sigma^{\otimes n} \right\|_1 \right).
\end{align}

\subsubsection{Asymmetric binary quantum hypothesis testing}

\label{sec:asym-bin-QHT}

In the setting of asymmetric binary quantum hypothesis testing, there are no prior probabilities associated with preparing an unknown state---we simply assume that a state is prepared deterministically, but the identity of the prepared state is unknown to the distinguisher. The goal of the distinguisher is to minimize the probability of the second kind of error subject to a constraint on the probability of the first kind of error. Indeed, given a fixed $\varepsilon \in [0,1]$, the scenario reduces to the following optimization problem:
\begin{equation} 
\beta_{\varepsilon}(\rho^{\otimes n}\Vert\sigma^{\otimes n})\coloneqq\inf
_{\Lambda^{(n)}}\left\{
\begin{array}
[c]{c}
\operatorname{Tr}[\Lambda^{(n)}\sigma^{\otimes n}]:\operatorname{Tr}
[(I^{\otimes n}-\Lambda^{(n)})\rho^{\otimes n}]\leq\varepsilon,\\
0\leq\Lambda^{(n)}\leq I^{\otimes n}
\end{array}
\right\}  .
\label{eq:beta-err-asymm}
\end{equation}
This formulation of the hypothesis testing problem is relevant, for example, to detecting anomalies or disease signatures, where the type~I error probability is called the ``false positive rate,'' and the type~II error probability is called the ``missed detection rate'' or ``false negative rate.'' Indeed, in such a scenario, one is willing to tolerate a fixed rate of false positives but then wishes to minimize the rate of missed detections.

\subsubsection{Multiple quantum hypothesis testing}

\label{sec:multiple-QHT}

We now review multiple quantum hypothesis testing, also known as 
$M$-ary quantum hypothesis testing because the goal is to select one among $M$ possible hypotheses. 
Let $\mathcal{S} \coloneqq  \left\{ (p_m,  \rho_m)\right\}_{m=1}^M $ be an ensemble of $M$ states with prior probabilities taking values in the set $\left\{ p_m \right\}_{m=1}^M$.
Without loss of generality, let us assume that $p_m>0$ for all $m\in \{1,2,\ldots, M\}$.
The minimum error probability of $M$-ary hypothesis testing, given $n$ copies of the unknown state, is as follows:
\begin{align}
	\label{eq:error_M-ary}
	p_e(\mathcal{S},n) 
	&\coloneqq  \inf_{ (\Lambda^{(n)}_1, \ldots, \Lambda^{(n)}_M) } \sum_{m=1}^M p_m \Tr\!\left[(I^{\otimes n}-\Lambda^{(n)}_m)\rho^{\otimes n}_m  \right],
\end{align}
where the minimization is over every positive operator-valued measure (POVM) (i.e., a tuple $(\Lambda^{(n)}_1,\ldots, \Lambda^{(n)}_M)$ satisfying $\Lambda^{(n)}_i \geq 0$ for all $i\in \{1,\ldots,M\}$ and $\sum_{m=1}^M \Lambda^{(n)}_m = I^{\otimes n}$).

\subsubsection{Efficient algorithm for computing optimal error probabilities in quantum hypothesis testing}

In Sections~\ref{sec:sym-bin-QHT}, \ref{sec:asym-bin-QHT}, and \ref{sec:multiple-QHT}, we outlined several quantum hypothesis testing tasks, including symmetric binary, asymmetric binary, and multiple hypotheses. Along with these tasks are the associated optimal error probabilities in~\eqref{eq:Helstrom1},~\eqref{eq:beta-err-asymm}, and~\eqref{eq:error_M-ary}, which involve an optimization over all possible measurements acting on $n$ quantum systems.

Let us suppose that each quantum system has a fixed dimension $d$. In this case, all of these optimization problems can be cast as semi-definite programs (SDPs)~\cite{BV04}, meaning that standard semi-definite programming solvers can be used to calculate the optimal values. However, the matrices involved in these optimization problems are of dimension $d^n \times d^n$, and thus a naive approach to solving these SDPs requires time exponential in the number $n$ of systems. As such, it might seem as if calculating these optimal values is an intractable task.

The naive approach mentioned above neglects the fact that the optimization problems in~\eqref{eq:Helstrom1},~\eqref{eq:beta-err-asymm}, and~\eqref{eq:error_M-ary} possess a large amount of symmetry, due to the fact that the inputs to these problems are tensor-power states and are thus invariant under permutations of the systems. By exploiting this permutation symmetry, we show in Appendix~\ref{app:efficient-algo-sdp} that the SDPs needed to compute the optimal values can be reduced to SDPs of size polynomial in $n$, where the polynomial degree is a function of $d$. Thus, based on standard results on the efficiency of SDP solvers~\cite{PW2000,Arora2005,Arora2012,Lee2015}, there is thus a polynomial-time algorithm for computing the optimal error probabilities in quantum hypothesis testing. This observation closely follows other recent advances in quantum information, having to do with computing bounds on channel capacities~\cite{FawziIEEE22} and with computing the optimal error probability in asymmetric channel discrimination~\cite{BDS+24}. 

\section{Statement of the problem}

While reviewing the various hypothesis testing problems in Section~\ref{sec:qht-review}, we see that the goal is to minimize the error probability for a fixed choice of $n\in\mathbb{N}$. As stated in the introduction (Section~\ref{sec:intro}), the optimization task for sample complexity flips this reasoning around: indeed, the goal is to determine the minimum value of $n\in\mathbb{N}$ (i.e., minimum number of samples) needed to meet a fixed error probability constraint. As discussed in the introduction of our paper, this formulation of the hypothesis testing problem is more consistent with the notion of the runtime of a probabilistic or quantum algorithm, for which one typically finds that the runtime depends on the desired error. Indeed, if there are procedures for preparing the various states involved in a given hypothesis testing problem, with fixed runtimes, then the sample complexity corresponds to the total amount of time one must wait to prepare $n$ samples to achieve a desired error probability. It should certainly be mentioned that sample complexity ignores the runtime of a quantum measurement procedure that actually discriminates the states, so that this notion can be understood as straddling the boundary between information theory and complexity theory.
Let us finally mention here that the notion of sample complexity has been extensively studied in the classical setting, along with applications to property testing and studies on information-constrained settings including privacy and communication limitations \cite{thesis_SC_lower,Cann_PT,Goldreich_2017,Inference_SC1,inference_SC2,com_priv_accur23,Structu_Optimal_Tests19,Private_H_Selec,Locally_private_HS20,Contraction_local_new24}.

More formally, we state the  definitions of the sample complexity of symmetric binary, asymmetric binary, and multiple quantum hypothesis testing in the following Definitions~\ref{def:binary_symmetric}, \ref{def:binary_asymmetric}, and~\ref{def:M-ary}, respectively. In each case, the simplest way to state the definitions is to employ the various error-probability metrics in~\eqref{eq:Helstrom1},~\eqref{eq:beta-err-asymm}, and~\eqref{eq:error_M-ary} and define the sample complexity to be the minimum value of $n\in\mathbb{N}$ needed to get the error probability metric below a threshold~$\varepsilon \in [0,1]$.

\begin{figure}
       \centering
       \includegraphics[width=0.55\linewidth]{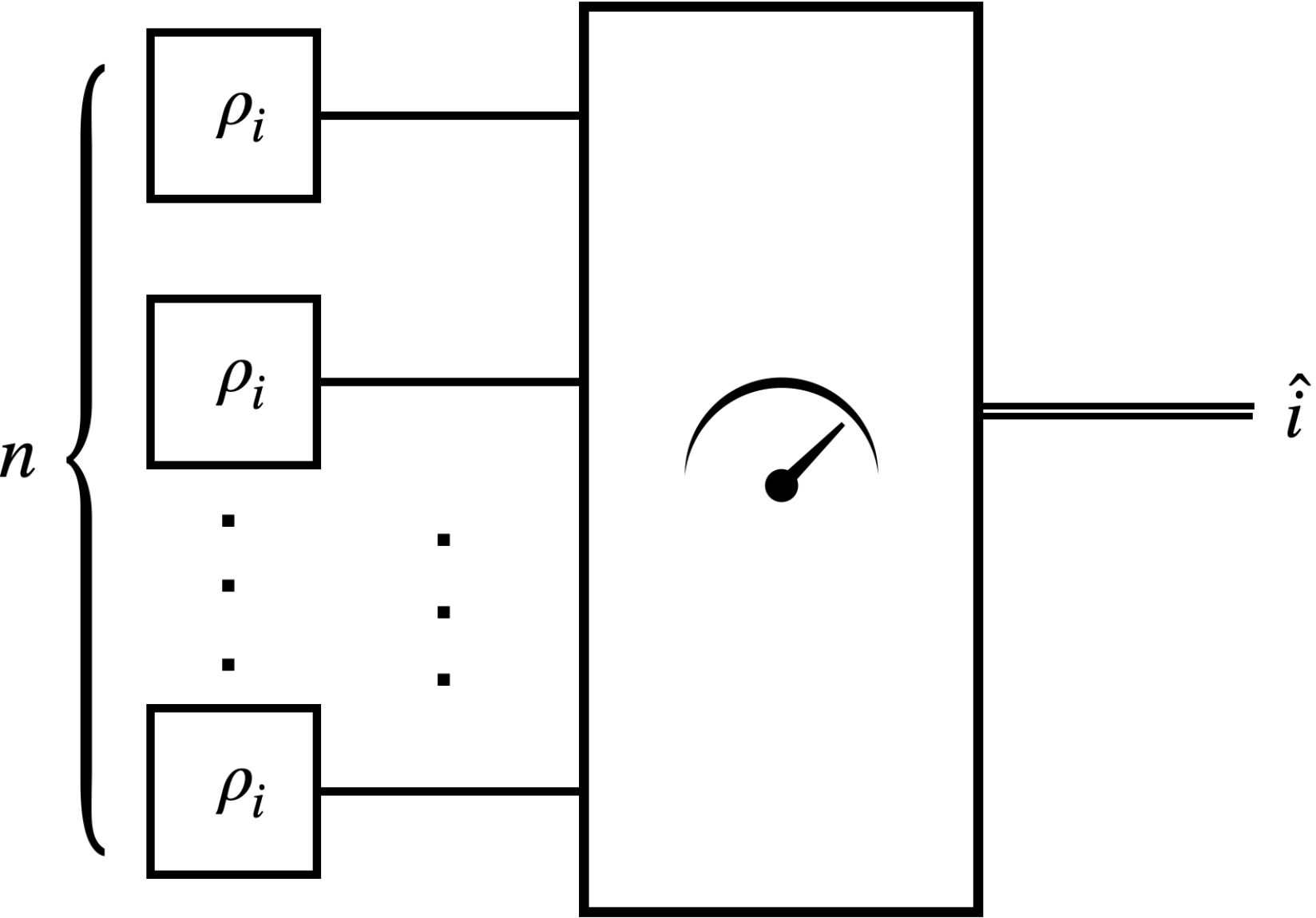}
\caption{Sample complexity of symmetric hypothesis testing: Let $\left\{ (p_i,  \rho_i)\right\}_{i=1}^M $ be an ensemble of $M$ states with prior probabilities taking values in the set $\left\{ p_i \right\}_{i=1}^M$.
Fix $\varepsilon \in [0,1]$. The sample complexity is the smallest $n \in \mathbb{N}$ such that $n$ copies of the unknown state are sufficient to decide its identity, in such a way that the error probability does not exceed the threshold~$\varepsilon
       $ (i.e.; $ \sum_{i=1}^M p_i \operatorname{Pr}\{\hat{i}  \neq i\} \leq \varepsilon$). When $M=2$, the setting is  symmetric binary hypothesis testing.
       }
\label{fig:intro_sample_complexity}
\end{figure} 

\begin{definition}[Sample complexity of symmetric binary hypothesis testing] \label{def:binary_symmetric}
	Let $p\in(0,1)$, $q = 1-p$, and $\varepsilon \in [0,1]$, and let $\rho$ and $\sigma$ be states.
	The sample complexity $n^*(p,\rho,q,\sigma, \varepsilon)$ of symmetric binary quantum hypothesis testing is defined as follows:
	\begin{align}
		n^*(p,\rho,q,\sigma, \varepsilon)
\label{eq:def:sample_complexity_symmetric1}
		&\coloneqq \inf\left\{ n \in \mathbb{N} : p_e(p, \rho,q,\sigma,n) \leq \varepsilon \right \} .
	\end{align}
\end{definition}

\begin{definition}[Sample complexity of asymmetric binary hypothesis testing] 
\label{def:binary_asymmetric}
Let $\varepsilon,\delta\in\left[  0,1\right]  $, and let $\rho$ and $\sigma$
be states. The sample complexity $n^{\ast}(\rho,\sigma,\varepsilon,\delta)$ of asymmetric binary quantum hypothesis testing is
defined as follows:
\begin{equation}
n^{\ast}(\rho,\sigma,\varepsilon,\delta)\coloneqq \inf\left\{  n\in\mathbb{N}
:\beta_{\varepsilon}(\rho^{\otimes n}\Vert\sigma^{\otimes n})\leq
\delta\right\}  \label{eq:asymm-beta-rewrite}  .
\end{equation}

\end{definition}

\begin{definition}[Sample complexity of $M$-ary hypothesis testing] \label{def:M-ary}
	Let $\varepsilon\in[0,1]$, and let $\mathcal{S} \coloneqq  \left\{ (p_m,  \rho_m)\right\}_{m=1}^M $ be an ensemble of $M$ states.
	The sample complexity $n^*(\mathcal{S},\varepsilon)$ of $M$-ary quantum hypothesis testing is defined as follows:
	\begin{align} \label{eq:def:sc_M-ary}
		n^*(\mathcal{S},\varepsilon) 
		& \coloneqq  \inf\left\{ n \in \mathbb{N} : p_e(\mathcal{S},n) \leq \varepsilon\right\}.
	\end{align}
\end{definition}

Figure~\ref{fig:intro_sample_complexity} depicts the setting involved in the sample complexity of symmetric hypothesis testing,
and 
Figure~\ref{fig:sample_complexity} illustrates the figures of merit underlying sample complexity in both the asymmetric and symmetric binary settings.

\begin{remark}[Equivalent expressions for sample complexities]
The sample complexity $n^*(p,\rho,q,\sigma, \varepsilon)$ of symmetric binary quantum hypothesis testing has the following equivalent expressions:
	\begin{align}
		n^*(p,\rho,q,\sigma, \varepsilon)
& =     \inf_{\Lambda^{(n)}} 
		\begin{Bmatrix} n\in \mathbb{N} : p \Tr\!\left[\left(I^{\otimes n} - \Lambda^{(n)}\right) \rho^{\otimes n} \right] + q \Tr\!\left[ \Lambda^{(n)} \sigma^{\otimes n} \right] \leq \varepsilon ,  \\  0\leq \Lambda^{(n)}\leq I^{\otimes n} 
		\end{Bmatrix}
		\\
		\label{eq:def:sample_complexity_symmetric2}
		&= \inf \left\{ n\in\mathbb{N}: \frac12 \left( 1 - \left\| p \rho^{\otimes n} - q \sigma^{\otimes n} \right\|_1 \right) \leq \varepsilon  \right\}
		\\
		\label{eq:def:sample_complexity_symmetric3}
		&=\inf \left\{ n\in\mathbb{N}: 1-2\varepsilon \leq  \left\| p \rho^{\otimes n} - q \sigma^{\otimes n} \right\|_1   \right\},
	\end{align}
	where the equality~\eqref{eq:def:sample_complexity_symmetric2} follows from the Helstrom--Holevo theorem in~\eqref{eq:Helstrom1}--\eqref{eq:Helstrom2}.
 By recalling  the quantity $\beta_{\varepsilon}(\rho^{\otimes n}\Vert\sigma^{\otimes n})$ defined in~\eqref{eq:beta-err-asymm}, 
we can rewrite the sample complexity $n^{\ast}(\rho,\sigma,\varepsilon
,\delta)$ of asymmetric binary quantum hypothesis testing in the following two ways:
\begin{align}
n^{\ast}(\rho,\sigma,\varepsilon,\delta)  & =\inf_{\Lambda^{(n)}}\left\{
\begin{array}
[c]{c}
n\in\mathbb{N}:\operatorname{Tr}[(I^{\otimes n}-\Lambda^{(n)})\rho^{\otimes
n}]\leq\varepsilon,\\
\operatorname{Tr}[\Lambda^{(n)}\sigma^{\otimes n}]\leq\delta,\ 0\leq
\Lambda^{(n)}\leq I^{\otimes n}
\end{array}
\right\}
\label{eq:asymm-beta-rewrite-both-errs}
\\
& =\inf\left\{  n\in\mathbb{N}:\beta_{\delta}(\sigma^{\otimes n}\Vert
\rho^{\otimes n})\leq\varepsilon\right\}  .\label{eq:asymm-beta-rewrite-2}
\end{align}
See Appendix~\ref{app:proof-of-equiv-exps-asymm-err} for an explicit proof.
The expression in~\eqref{eq:asymm-beta-rewrite-both-errs} indicates that the sample complexity for asymmetric binary quantum hypothesis testing can be thought of as 
the minimum number of samples required to get the type~I error probability below $\varepsilon$ and the type~II error probability below $\delta$.
Finally, the sample complexity $n^*(\mathcal{S},\varepsilon)$ of $M$-ary quantum hypothesis testing can be rewritten as follows:
	\begin{align}
		n^*(\mathcal{S},\varepsilon) 
  & = \inf_{ \left(\Lambda^{(n)}_1, \ldots, \Lambda^{(n)}_M \right) } 
		\left\{ n\in \mathbb{N}: \sum_{m=1}^M p_m \Tr\!\left[\left(I^{\otimes n}-\Lambda_m^{(n)}\right) \rho_m^{\otimes n}  \right] \leq \varepsilon
		\right\},
	\end{align}
	where $\Lambda^{(n)}_1,\ldots, \Lambda^{(n)}_M \geq 0$ and $\sum_{m=1}^M \Lambda^{(n)}_m = I^{\otimes n}$.
    
\end{remark}

Before proceeding to the development of our main results in the next section, let us first identify some conditions under which the sample complexity of symmetric binary quantum hypothesis testing is  trivial, i.e., such that it is equal to either one or infinity.

\begin{remark}[Trivial cases]
\label{rem:trivial-conditions}
Let $p$, $q$, $\varepsilon$, $\rho$, and $\sigma$ be as stated in Definition~\ref{def:binary_symmetric}.
    If 
    $\rho \perp \sigma$ (i.e., $\rho\sigma = 0$),
    $\varepsilon \in [1/2,1]$, or $\exists\, s\in[0,1]$ such that $\varepsilon \geq p^s q^{1-s}$, then the following equality holds
	\begin{align}
		\label{eq:binary_symmetric1}
		n^*(p,\rho,q,\sigma, \varepsilon) = 1.
	\end{align}
	If $\rho = \sigma$ and $\min\{p,q\} > \varepsilon \in [0,1/2)$, then
	\begin{align}
		\label{eq:binary_symmetric2}
		n^*(p,\rho,q,\sigma, \varepsilon) = +\infty.
	\end{align}
\end{remark}

\begin{proof}
See Appendix~\ref{app:proof-remark-triv-cond}.
\end{proof}

\begin{figure}
       \centering
       \includegraphics[width=0.8\linewidth]{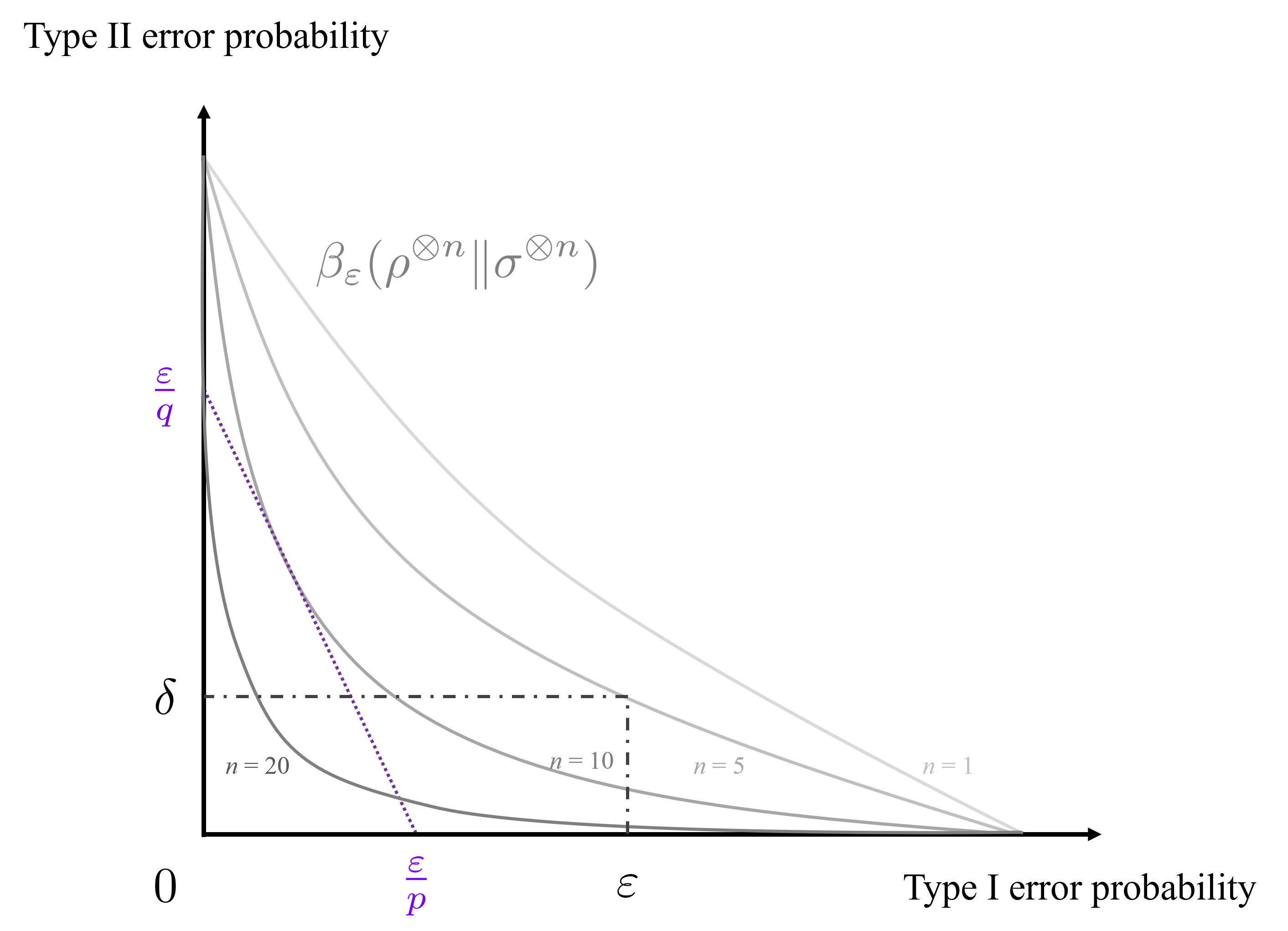}
       \caption{
       Illustration of sample complexity in the asymmetric and symmetric binary settings.
       Each convex curve plots the minimum type~II error probability when the type~I error probability is no larger than $\varepsilon$,
       i.e., $\beta_{\varepsilon}(\rho^{\otimes n}\Vert \sigma^{\otimes n})$ defined in \eqref{eq:beta-err-asymm}, for some sample $n$, showing the trade-off between the minimum type~I error probability and the minimum type~II error probability (unless $\rho \perp \sigma$).
       As $n$ increases, one can find a collective measurement $\{\Lambda_{\rho}^{(n)}, \Lambda_{\sigma}^{(n)} \}$ to make both error probabilities smaller; hence, the curve moves toward the origin.
       The smallest sample size such that the type~I error probability  is no larger than $\varepsilon$ and the type~II error probability is no larger than $\delta$ (as shown by the dash-dotted line) is $n=5$, showing the sample complexity $n^*(\rho,\sigma, \varepsilon ,\delta) = 5$ of asymmetric binary hypothesis testing in \eqref{eq:asymm-beta-rewrite}.
       The minimum probability of symmetric hypothesis testing, \eqref{eq:Helstrom1}, constrained to be $\varepsilon$, is given by the dotted linear line, and the curve with minimum sample size that touches the linear line is $n=10$.
       In this case, the sample complexity of symmetric hypothesis testing is $n^*(p,\rho,q,\sigma,\varepsilon) = 10$.
       }
       \label{fig:sample_complexity}
\end{figure}

\section{Sample complexity results}

\label{sec:samp-comp-results}

Having defined various sample complexities of interest in Definitions~\ref{def:binary_symmetric}, \ref{def:binary_asymmetric}, and \ref{def:M-ary}, it is clear that calculating the precise values of $n^*(p,\rho,q,\sigma, \varepsilon)$, $
n^{\ast}(\rho,\sigma,\varepsilon,\delta)$, and $
n^*(\mathcal{S},\varepsilon)$ is not an easy computational problem. As such, we are then motivated to find lower and upper bounds on these sample complexities, which are easier to compute and ideally match in an asymptotic sense. We are able to meet this challenge for the symmetric and asymmetric binary settings, mostly by building on the vast knowledge that already exists regarding quantum hypothesis testing. For {$M$-ary} hypothesis testing, we are only able to give lower and upper bounds that differ asymptotically by a factor of $\ln M$. However, we note that this finding is consistent with the best known result in the classical case~\cite[Fact~2.4]{pensia2023communicationconstrained}. Before proceeding with stating our results, let us note here that all of our results hold for states acting on a separable (infinite-dimensional) Hilbert space, unless otherwise noted.

\subsection{Symmetric binary quantum hypothesis testing}

\subsubsection{Two pure states}

Let us begin by considering the sample complexity of symmetric binary quantum hypothesis testing when distinguishing two pure states (i.e., rank-$1$ projection operators), which is much simpler than the more general case of two arbitrary mixed states. It is also interesting from a fundamental perspective because there is no classical analog of pure states~\cite{hardy2001quantum,Hardy2002}, as all classical pure states correspond to degenerate (deterministic) probability distributions that are either perfectly distinguishable or perfectly indistinguishable. In this case, we find an exact result, which furthermore serves as a motivation for the kind of expression we wish to obtain in the general mixed-state case. This finding was essentially already reported in~\cite[Eq.~(39)]{kimmel2017hamiltonian}, with the main difference below being a generalization to arbitrary priors. In any case, Theorem~\ref{theorem:binary_symmetric_pure} is a direct consequence of some simple algebra and the following equality~\cite[Proposition~21]{mishra2023optimal}, which holds for (unnormalized) vectors $|\varphi\rangle$ and $|\zeta\rangle$:
\begin{equation}
\left\Vert |\varphi\rangle\!\langle\varphi|-|\zeta\rangle\!\langle\zeta
|\right\Vert _{1}^{2}=\left(  \langle\varphi|\varphi\rangle+\langle\zeta
|\zeta\rangle\right)  ^{2}-4\left\vert \langle\zeta|\varphi\rangle\right\vert
^{2}.
\label{eq:vector-trace-norm-fid-identity}
\end{equation}

\begin{theorem}[Sample complexity: symmetric binary case with two pure states] \label{theorem:binary_symmetric_pure}
Let $p$, $q$, $\rho$, $\sigma$, and $\varepsilon$ be as stated in Definition~\ref{def:binary_symmetric}, and furthermore let $\rho = \psi \equiv |\psi\rangle\!\langle \psi|$ and $\sigma = \varphi \equiv |\varphi\rangle\!\langle \varphi|$ be pure states, such that the conditions in Remark~\ref{rem:trivial-conditions} do not hold. Then
\begin{equation}
n^{\ast}(p,\psi,q,\phi,\varepsilon)=\left\lceil \frac{\ln\!\left(  \frac
{pq}{\varepsilon\left(  1-\varepsilon\right)  }\right)  }{-\ln F(\psi,\phi
)}\right\rceil .
\end{equation}
\end{theorem}

\begin{proof}
See Appendix~\ref{app:proof-binary_symmetric_pure}.
\end{proof}

The exact result above indicates the kind of expression that we should strive for in the more general mixed-state case: logarithmic dependence on the inverse error probability and inverse dependence on the divergence $-\ln F(\psi,\phi
)$. For pure states, by inspecting~\eqref{eq:petz-renyi-div} and~\eqref{eq:sandwiched-renyi-div}, we see that this latter divergence is equivalent to the Petz--R\'enyi relative entropy of order 1/2, as well as the sandwiched R\'enyi relative entropy of order 1/2. Furthermore, it is equivalent, up to a constant, to the same divergences for every order $\alpha \in (0,1)$. 

\begin{figure}
       \centering
       \includegraphics[width=0.8\linewidth]{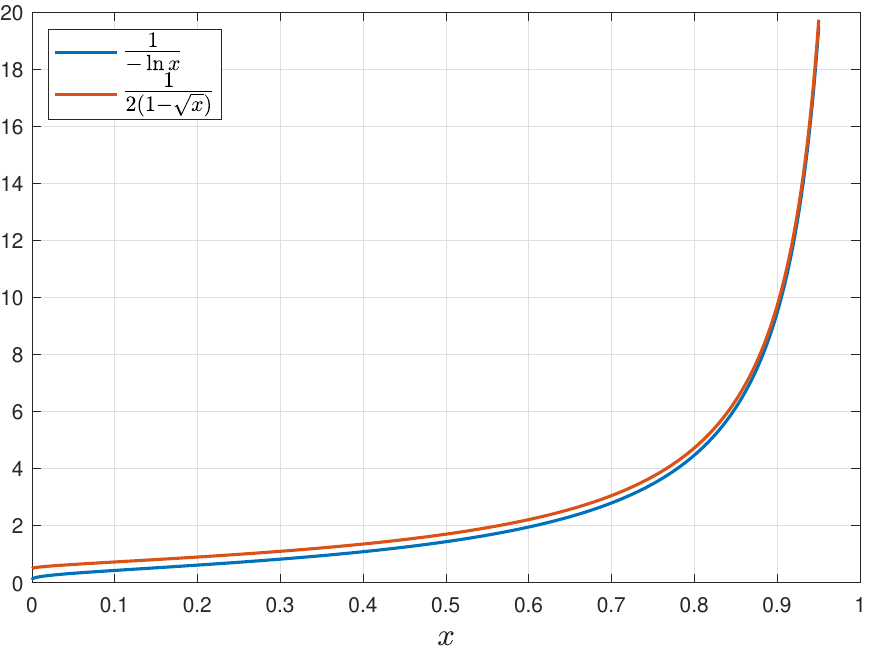}
       \caption{Comparison between the functions $\frac{1}{-\ln x}$ and $\frac{1}{2(1-\sqrt{x})}$ for $x \in [0,1]$, demonstrating that there is little difference in characterizing sample complexity of binary symmetric hypothesis testing by $[-\ln F(\rho,\sigma)]^{-1}$ instead of $\left[ d_{\mathrm{B}}(\rho,\sigma) \right]^{-2}$. The largest gap between these functions is equal to $\frac{1}{2}$ and occurs at $x=0$.}
       \label{fig:compare-F-Hellinger}
   \end{figure}

The characterization in Theorem~\ref{theorem:binary_symmetric_pure} depends inversely on the negative logarithm of the fidelity, rather than the inverse of the squared Hellinger or Bures distance, the latter being more common in formulations of the sample complexity of classical hypothesis testing (see~\cite[Theorem~4.7]{bar2002complexity} and~\cite{canonne17note}). However, we should note that there is little difference between the functions $\frac{1}{-\ln x}$ and $\frac{1}{2(1-\sqrt{x})}$ when $x \in [0,1]$. Indeed, the largest gap between these functions is $1/2$, occurring at $x=0$. As such, our characterization in terms of $[-\ln F(\rho,\sigma)]^{-1}$ instead of $\left[ d_{\mathrm{B}}(\rho,\sigma) \right]^{-2}$ makes little to no difference in terms of asymptotic sample complexity. Note, however, that this factor of $1/2$ can make a notable difference in applications in the finite or small sample regime. The two functions are plotted in Figure~\ref{fig:compare-F-Hellinger} in order to make this point visually clear.

\subsubsection{Two general states} \label{sec:two_general_states}

Let us now move on to the general mixed-state case. Theorem~\ref{theorem:binary_symmetric} below provides lower and upper bounds on the sample complexity for this case. The main tool for establishing both lower bounds is the generalized Fuchs-van-de-Graaf inequality recalled in~\eqref{eq:FG99}, and the main tool for establishing the upper bound is the Audenaert inequality recalled in~\eqref{eq:Chernoff}.
Let us note that the upper bound is achieved by the Helstrom--Holevo test, sometimes also called the quantum Neyman--Pearson test; i.e.,~$ \Lambda^{(n)}$ is a projection onto the positive part of 
$p \rho^{\otimes n} -q \sigma^{\otimes n}$.

\begin{theorem}[Sample complexity: symmetric binary case with two general states] \label{theorem:binary_symmetric}
	Let $p$, $q$, $\varepsilon$, $\rho$, and $\sigma$ be as stated in Definition~\ref{def:binary_symmetric} such that the conditions in Remark~\ref{rem:trivial-conditions} do not hold.
	Then the following bounds hold
	\begin{align}
		\label{eq:binary_symmetric3}
		\max\left\{ \frac{\ln(pq/\varepsilon) }{ -\ln F(\rho,\sigma) } ,\frac{1-\frac{\varepsilon(1-\varepsilon)}{pq}}{ \left[d_{\mathrm{B}}(\rho,\sigma)\right]^2  } \right\} \leq n^*(p,\rho,q,\sigma, \varepsilon)
		\leq \left \lceil \inf_{s\in\left[  0,1\right]  }\frac{\ln\!\left(  \frac
{p^{s}q^{1-s}}{\varepsilon}\right)  }{-\ln\operatorname{Tr}[\rho^{s}
\sigma^{1-s}]}\right\rceil.
	\end{align}
\end{theorem}

\begin{proof}
    See Appendix~\ref{app:binary_symmetric}. 
\end{proof}

The statement given in Theorem~\ref{theorem:binary_symmetric} is sufficiently strong to lead to the following characterization of sample complexity of symmetric binary quantum hypothesis testing in the general case, given by Corollary~\ref{corollary:binary-informal-2} below. This finding essentially matches the exact result for the pure-state case in Theorem~\ref{theorem:binary_symmetric_pure}, up to constant factors and for small error probability~$\varepsilon$. Corollary~\ref{corollary:binary-informal-2} below follows by using the first lower bound in~\eqref{eq:binary_symmetric3} and by picking $s=1/2$ in the upper bound in~\eqref{eq:binary_symmetric3}, along with relations between $F$ and $F_{\mathrm{H}}$ recalled in~\eqref{eq:relation_fidelity}.

\begin{corollary} \label{corollary:binary-informal-2}
	Let $p,q,\varepsilon, \rho$, and $\sigma$ be as stated in Definition~\ref{def:binary_symmetric}, such that the conditions in Remark~\ref{rem:trivial-conditions} do not hold.
	Then the following inequalities hold:
	\begin{align}
		\label{eq:binary-informal-other1}
		\frac{\ln\!\left(\frac{pq}{\varepsilon}\right) }{ -\ln F_\mathrm{H}(\rho,\sigma) }   & \leq n^*(p,\rho,q,\sigma, \varepsilon)
		\leq  \left\lceil \frac{ \ln \!\left( \frac{\sqrt{p q} }{ \varepsilon } \right) }{-\frac{1}{2}\ln  F_{\mathrm{H}}(\rho,\sigma)} \right \rceil  .
		\\
		\label{eq:binary-informal-other2}
		\frac{\ln\!\left(\frac{pq}{\varepsilon}\right) }{ -\ln F(\rho,\sigma) }   & \leq n^*(p,\rho,q,\sigma, \varepsilon)
		\leq  \left \lceil \frac{ \ln \!\left( \frac{\sqrt{p q} }{ \varepsilon } \right) }{-\frac{1}{2}\ln  F(\rho,\sigma)} \right\rceil  .
	\end{align}
	Thus, after fixing the priors $p$ and $q$ to be constants, we have characterized the sample complexity as follows:
	\begin{align}
		n^*(p,\rho,q,\sigma,\varepsilon) 
		= \Theta\!\left( \frac{\ln\!\left(\frac{1}{\varepsilon}\right) }{ -\ln F_\mathrm{H}(\rho,\sigma) }  \right)
		= \Theta\!\left( \frac{\ln\!\left(\frac{1}{\varepsilon}\right) }{ -\ln F(\rho,\sigma) }  \right).
  \label{eq:asymp-samp-comp-sym-bin-QHT}
	\end{align}
\end{corollary}

\begin{proof}
The first inequality in~\eqref{eq:binary-informal-other1} follows from the first inequality in~\eqref{eq:binary_symmetric3} and the inequalities in~\eqref{eq:relation_fidelity}. The second inequality in~\eqref{eq:binary-informal-other1} follows from the second inequality in~\eqref{eq:binary_symmetric3} by picking $s=1/2$. The inequalities in~\eqref{eq:binary-informal-other2} follow from similar reasoning and using the inequalities in~\eqref{eq:relation_fidelity}.
\end{proof}

Corollary~\ref{corollary:binary-informal-2} demonstrates that the asymptotic sample complexity in the classical case of commuting $\rho$ and $\sigma$ is  uniquely characterized by the negative logarithm of the classical fidelity, because $F_\mathrm{H}(\rho,\sigma) = F(\rho,\sigma)$ in such a case. The characterization in Corollary~\ref{corollary:binary-informal-2}  is a strengthening of the existing characterizations of sample complexity in the classical case (see~\cite[Theorem~4.7]{bar2002complexity} and~\cite{canonne17note}), as it has asymptotically matching lower and upper bounds and includes dependence on the error probability, and the underlying distributions (commuting states). 

However, Corollary~\ref{corollary:binary-informal-2} also demonstrates that the asymptotic sample complexity in the quantum case does not have a unique characterization. Indeed, the quantities $-\ln F(\rho,\sigma)$ and $-\ln F_\mathrm{H}(\rho,\sigma)$ are related by multiplicative constants, as recalled in~\eqref{eq:relation_Bures_Chernoff}, and these constants get discarded when using $O$, $\Omega$, and $\Theta$ notation. These constants can actually have dramatic ramifications for the quantum technology required to implement a given measurement strategy. On the one hand, the upper bound on sample complexity in~\eqref{eq:binary_symmetric3} assumes the ability to perform a collective measurement on all $n$ copies of the unknown state. This is essentially equivalent to performing a general unitary on all $n$ systems followed by a product measurement, and so will likely need a full-scale, fault-tolerant quantum computer for its implementation. On the other hand, the upper bound on sample complexity in~\eqref{eq:binary-informal-other2} can be achieved by product measurements and classical post-processing.
This is in contrast to the upper bound in~\eqref{eq:binary_symmetric3}.  Indeed, one can repeatedly apply a measurement known as the Fuchs--Caves measurement~\cite{Fuchs1995} on each individual copy of the unknown state and then process the measurement outcomes classically to achieve the upper bound in~\eqref{eq:binary-informal-other2} (see Appendix~\ref{app:fuchs-caves} for details). Asymptotically, there is no difference between the sample complexities of these strategies, even though there are drastic differences between the technologies needed to realize them. Thus, if one is interested in merely achieving the asymptotic sample complexity, then the latter strategy using the Fuchs--Caves measurement and classical post-processing is preferred.

Continuing with the observation that the asymptotic sample complexity in the quantum case is not unique, let us note that one can even equivalently characterize it by means of the following family of $z$-fidelities, which are based on the $\alpha$-$z$ divergences~\cite{AD15} by setting $\alpha = 1/2$ and obey the data-processing inequality for all $z\geq 1/2$~\cite[Theorem~1.2]{Z20}:
\begin{equation}
F_{z}(\rho,\sigma)\coloneqq \left(  \operatorname{Tr}\!\left[  \left(  \sigma^{1/4z}
\rho^{1/2z}\sigma^{1/4z}\right)  ^{z}\right]  \right)  ^{2} = \left\Vert \rho^{1/4z}\sigma^{1/4z}\right\Vert _{2z}^{4z}.
\end{equation}
When $z=1/2$, we get that $F_{z=1/2}(\rho,\sigma) = F(\rho,\sigma)$, and when $z=1$, we get that $F_{z=1}(\rho,\sigma) = F_{\mathrm{H}}(\rho,\sigma)$.
Furthermore,~\cite[Proposition~6]{LT15} implies that the $z$-fidelities are monotone decreasing for all $z\geq 1/2$, so that
\begin{equation}
    -\ln F(\rho,\sigma) \leq -\ln F_z(\rho,\sigma) \leq -\ln F_{z'}(\rho,\sigma) \leq -\ln F_{\mathrm{H}}(\rho,\sigma) \leq -2\ln F(\rho,\sigma),
    \label{eq:z-fidelities-ineqs}
\end{equation}
for all $z$ and $z'$ satisfying $1/2 \leq z \leq z' \leq 1$.
The inequalities in~\eqref{eq:z-fidelities-ineqs} combined with~\eqref{eq:asymp-samp-comp-sym-bin-QHT} then imply that all of these $z$-fidelities equivalently characterize the asymptotic sample complexity of symmetric binary quantum hypothesis testing.

\subsection{Asymmetric binary quantum hypothesis testing}

Theorem~\ref{thm:binary_asymmetric} below provides lower and upper bounds on the sample complexity of asymmetric quantum hypothesis testing, as introduced in Definition~\ref{def:binary_asymmetric}.
Unlike the symmetric setting considered in Section~\ref{sec:two_general_states}, the lower bound is expressed in terms of the sandwiched R\'enyi divergence~$\widetilde{D}_{\alpha}$, while the upper bound is expressed in terms of the Petz--R\'enyi divergence~$D_{\alpha}$.
The main tool for establishing the lower bounds is the strong converse bound recalled in Lemma~\ref{lemma:strong_converse}, and the main tool for establishing the upper bound is the quantum Hoeffding bound recalled in Lemma~\ref{lemma:Hoeffding}.
Let us note that the upper bound can be achieved by the Helstrom--Holevo test, i.e.,~projection 
onto the positive part of $ \rho^{\otimes n} - \lambda \sigma^{\otimes n}$, or the \emph{pretty-good measurement}~\cite{Bel75, HW94}
$ \Lambda^{(n)} = ( \rho^{\otimes n} + \lambda \sigma^{\otimes n} )^{-1/2} \rho^{\otimes n} ( \rho^{\otimes n} + \lambda \sigma^{\otimes n} )^{-1/2}
$ for some properly chosen parameter $\lambda > 0$.

\begin{theorem}
\label{thm:binary_asymmetric} Fix $\varepsilon,\delta \in (0,1)$, and let $\rho$ and $\sigma$ be states. 
Suppose there exists
$\gamma>1$ such that $\widetilde{D}_{\gamma}(\rho\Vert\sigma)<+\infty
$ and $\widetilde{D}_{\gamma}(\sigma\Vert\rho)<+\infty
$.
Then the following  bounds hold for the sample complexity $n^{\ast}(\rho,\sigma,\varepsilon,\delta)$ of asymmetric binary quantum hypothesis testing:
\begin{multline}
\max\left\{  \sup_{ \alpha \in (1,\gamma] } \left(  \frac{\ln\!\left(  \frac{\left(  1-\varepsilon\right)
^{\alpha^{\prime}}}{\delta}\right)  }{\widetilde{D}_{\alpha}(\rho\Vert\sigma
)}\right)  ,\ \sup_{ \alpha \in (1,\gamma] } \left(  \frac{\ln\!\left(  \frac{\left(
1-\delta\right)  ^{\alpha^{\prime}}}{\varepsilon}\right)  }{\widetilde
{D}_{\alpha}(\sigma\Vert\rho)}\right)  \right\}  \leq n^{\ast}(\rho,\sigma,\varepsilon,\delta) \\ \leq 
\min\left\{  \left\lceil
\inf_{\alpha\in\left(  0,1\right)  }\left(  \frac{\ln\!\left(  \frac
{\varepsilon^{\alpha^{\prime}}}{\delta}\right)  }{D_{\alpha}(\rho\Vert\sigma
)}\right)  \right\rceil ,\left\lceil \inf_{\alpha\in\left(  0,1\right)
}\left(  \frac{\ln\!\left(  \frac{\delta^{\alpha^{\prime}}}{\varepsilon
}\right)  }{D_{\alpha}(\sigma\Vert\rho)}\right)  \right\rceil \right\}  .
\label{eq:binary_asymmetric_samp_comp}
\end{multline}
where $\alpha^{\prime}\coloneqq\frac{\alpha}{\alpha-1}$.
\end{theorem}

\begin{proof}
See Appendix~\ref{app:proof_samp_comp_binary_asymmetric}.
\end{proof}

\begin{remark}
For the finite-dimensional case, let us note that 
\begin{equation}
    \widetilde{D}_{\alpha}(\rho\Vert\sigma
)\leq D_{\max}(\rho\Vert\sigma
) \coloneqq \ln \inf\left\{ \lambda > 0 : \rho \leq \lambda \sigma \right\}
= \ln \lambda_{\max}\!\left( \sigma^{-1/2} \rho \sigma ^{-1/2} \right) < +\infty
\end{equation}
for all $\alpha \in [0,\infty]$
as long as $\operatorname{supp}(\rho) \subseteq \operatorname{supp}(\sigma)$~\cite[Theorem 7]{MDS+13},~\cite[Lemma 3.12]{MO17} (see~\cite{Dat09} for $D_{\max}$).
\end{remark}

The bounds given in Theorem~\ref{thm:binary_asymmetric}  lead to the following asymptotic characterization of the sample complexity of asymmetric binary quantum hypothesis testing, given by Corollary~\ref{corollary:binary-asymmetric-informal} below.
Interestingly, 
the asymptotic sample complexity given in~\eqref{eq:binary-asymmetric-informal1} below establishes, up to a multiplicative constant, the following asymptotic relationship between the number $n$ of samples and the type~II error probability $\delta$:
\begin{align}
\label{eq:Stein1}
   n \simeq \frac{\ln\!\left(
\frac{1}{\delta}\right)  }{D(\rho\Vert\sigma)} \qquad \Leftrightarrow \qquad \delta \simeq \mathrm{e}^{- n D(\rho \Vert \sigma)},
\end{align}
whenever $\varepsilon \in (0,1)$ is fixed to be a constant and provided that $\delta$ is sufficiently small.
This  characterization is consistent with that from the  quantum Stein's lemma~\cite{HP91,ON00}, which states that the largest decaying rate of the type~II error probability is governed by the quantum relative entropy whenever the type~I error probability is at most a fixed $\varepsilon \in (0,1)$; i.e.,
\begin{align}
\label{eq:Stein2}
\lim_{n\to \infty} - \frac{1}{n} \ln \beta_{\varepsilon}( \rho^{\otimes n} \Vert \sigma^{\otimes n} )
= D(\rho \Vert \sigma).
\end{align}

\begin{corollary}
[Bounds on sample complexity of asymmetric binary quantum hypothesis
testing]
\label{corollary:binary-asymmetric-informal}
Fix $\varepsilon,\delta\in\left(  0,1\right)  $, and let $\rho$ and $\sigma$
be as stated in Definition~\ref{def:binary_asymmetric}. For sufficiently small $\delta$, the sample
complexity of asymmetric binary quantum hypothesis testing satisfies the
following bounds:
\begin{equation}
\frac{2}{5}\left(  \frac{\ln\!\left(  \frac{1}{\delta}\right)  }{D(\rho
\Vert\sigma)}\right)  \leq n^{\ast}(\rho,\sigma,\varepsilon,\delta
)\leq\left\lceil 4\left(  \frac{\ln\!\left(  \frac{1}{\delta}\right)  }
{D(\rho\Vert\sigma)}\right)  \right\rceil ,
\label{eq:binary-asymmetric-informal1}
\end{equation}
provided that there exists $\gamma\in(0,1]$ such that $D_{1+\gamma}(\rho
\Vert\sigma)<+\infty$. 
\end{corollary}

\begin{proof}
See Appendix~\ref{app:proof-binary-asymmetric-informal}. Also, see~\eqref{eq:small-delta-1} and~\eqref{eq:small-delta-2} for the precise meaning of \textquotedblleft
sufficiently small~$\delta$.\textquotedblright
\end{proof}

\subsection{Multiple quantum hypothesis testing} 

\label{sec:M-ary}

Now we generalize the binary mixed-state case discussed in Section~\ref{sec:two_general_states} to an arbitrary number of states, $M$.
The lower bound essentially follows from the lower bound for the two-state case because discriminating $M$ states simultaneously is not easier than discriminating any pair of states.
The upper bound is achieved by the \emph{pretty-good measurement}~\cite{Bel75, HW94}, which can be implemented via a quantum algorithm~\cite{GLM+22}.
The main tool for showing the upper bound is by testing each state against any of the other states.
Since there are $M(M-1)$ such pairs, the upper bound ends up with a logarithmic dependence on~$M$. This idea was introduced in~\cite[Theorem~4]{BK02}, and a similar idea has been used in~\cite{harrow2012many} and~\cite[Section 4]{AM14}. Let us note that the upper bound in~\eqref{eq:sc_M-ary} below is a refinement of that presented in~\cite{harrow2012many}.

\begin{theorem}[Sample complexity of $M$-ary quantum hypothesis testing] \label{theorem:sc_M-ary}
	Let $n^*(\mathcal{S},\varepsilon)$ be as stated in Definition~\ref{def:M-ary}.
	Then,
	\begin{align} \label{eq:sc_M-ary}
		\max_{m\neq \bar{m}}  \frac{ \ln\!\left( \frac{p_m p_{\bar m}}{ (p_m + p_{\bar m})\varepsilon } \right) }{ -\ln F(\rho_m, \rho_{\bar m }) }
		\leq 
		n^*(\mathcal{S},\varepsilon) \leq 
		\left\lceil  \max_{m\neq \bar{m}}  \frac{2\ln\!\left( \frac{ M(M-1) \sqrt{p_m} \sqrt{p_{\bar{m}}}  }{ 2\varepsilon } \right) }{ - \ln F\!\left(\rho_{m},\rho_{\bar{m}} \right) } \right\rceil .
	\end{align}
\end{theorem}

\begin{proof}
    See Appendix~\ref{app:proof-sc_M-ary}.
\end{proof}

\begin{remark} \label{rem:Ke_Li}
Ref.~\cite[Eq.~(36)]{Li16} established a multiple quantum Chernoff bound for $M$-ary hypothesis testing with error exponent $\min_{m\neq \bar{m}}  C(\rho_m\Vert \rho_{\bar{m}})$, which holds for finite-dimensional states. 
This result also implies 
\begin{align}
\label{eq:LiUpperBound}
    n^*(\mathcal{S},\varepsilon)
    \leq  O\!\left( \max_{m \neq \bar{m}} \frac{\ln M}{ -\ln F(\rho_m,\rho_{\bar{m}})} \right),
\end{align}
similarly to the upper bound in~\eqref{eq:sc_M-ary}.
\end{remark}
\begin{proof}
See Appendix~\ref{app:proof:rem:Ke_Li}.
\end{proof}

\begin{remark}[Pure-state case]
For a collection of pure states, the following upper bound can be derived by using~\cite[Eq.~(9)]{AM14}:
\begin{equation}
    n^*(\mathcal{S},\varepsilon) 
		\leq  \left\lceil\max_{m\neq \tilde{m}}  \frac{\ln\!\left( \frac{ M(M-1) \left(p_m^2 + p_{\tilde{m}}^{2} \right)  }{2 p_m p_{\tilde{m}} \varepsilon } \right) }{ - \ln F\!\left(\rho_{m},\rho_{\tilde{m}} \right) } \right\rceil.
\end{equation}
\end{remark}

\begin{remark}[Classical case]
As far as we know, the tightest upper bound on the sample complexity of $M$-ary hypothesis testing in the classical scenario is as follows~\cite[Theorem 15]{SV18}:
\begin{align} \label{eq:SV18}
    n^*(\mathcal{S}, \varepsilon)
    \leq \left\lceil  \max_{m\neq \bar{m}}  \inf_{s\in[0,1]} \frac{\ln\!\left( \frac{ M p_m^{s} {p}_{\bar{m}}^{1-s}  }{ \varepsilon } \right) }{ -\ln \Tr\!\left[ \rho_{m}^s \rho_{\bar{m}}^{1-s} \right] } \right\rceil,
\end{align}
which has an improved logarithmic dependence on $M$.
\end{remark}

By employing the same analysis as given in the proof of Corollary~\ref{corollary:binary-informal-2}, we obtain the following results:

\begin{corollary}[Asymptotic sample complexity] \label{cor:mary}
	Let $\mathcal{S}$ be as stated in Definition~\ref{def:M-ary}.
	If we regard each element of $\{p_m\}_{m=1}^M$ and $\varepsilon$ as constants, and suppose that 
	\begin{align}
		\frac{p_m p_{\bar m}}{ p_m + p_{\bar m} } \geq \varepsilon,
	\end{align}
	for $m, \bar{m} \in \{1,2,\ldots, M\}$, and $\bar{m} \ne m$, then we have characterized the sample complexity of $M$-ary quantum hypothesis testing as follows:
 \begin{align} \label{eq:sc_M-ary_asymptotics}
		\Omega\!\left( \frac{1}{ \min_{m\neq \bar{m}} [-\ln F\!\left(\rho_{m},\rho_{\bar{m}} \right)] } \right) 
		\leq 
		n^*(\mathcal{S},\varepsilon) 
		\leq O \!\left( \frac{\ln M }{\min_{m\neq \bar{m}} [-\ln F\!\left(\rho_{m},\rho_{\bar{m}} \right)] } \right).
	\end{align}
	The above bounds also hold when replacing $F$ with $F_{\mathrm{H}}$.
\end{corollary}

\section{Applications and further avenues of research}

\label{sec:apps}

We briefly discuss several broad areas in quantum science where the concept of sample complexity in quantum hypothesis testing is relevant. For those who come from diverse fields, our goal here is to illustrate how sample complexity can appear in widely different and perhaps unexpected areas -- thus it should not be confined to limited areas in information theory or computer science alone. We try to distill the intuitive essence of these connections that have appeared in various literature and to encourage the reader to explore these in further detail and to identify new relationships. 

In Section~\ref{sec:optimalquantum}, we explore possible uses of sample complexity in the proofs of optimality of certain quantum algorithms. In Section~\ref{sec:quantumlearning}, we devise ways of understanding sample complexity in the learning framework for quantum states. In Section~\ref{sec:quantumfoundation}, we note the relevance of sample complexity in the foundations of quantum information and quantum mechanics. Sample complexity in these contexts can serve as an improved technical tool, introduce a modified framework for an old problem, or provide new interpretations that originate from seemingly disparate areas. In Section~\ref{sec:priv-QHT}, we summarize recent works that have utilized the sample complexity bounds derived in this work in order to characterize the cost of privacy in the task of quantum hypothesis testing.

Let us emphasise here again that our results for  sample complexity, as presented in Section~\ref{sec:samp-comp-results}, hold for the discrimination not only of discrete-variable (qubits and qudits) quantum systems, but also continuous-variable quantum systems (qumodes) as well as hybrid discrete-continuous variable quantum systems, unless otherwise stated. These sample complexity bounds also hold for mixed states and multiple hypotheses $M$. In the applications discussed, while many previous works have mostly considered $M=2$ and pure qubit-based systems, our findings directly extend the applicability of those results to $M>2$, mixed states, continuous-variable quantum systems and also to hybrid discrete-continuous variable settings.

\subsection{Optimality of quantum algorithms} \label{sec:optimalquantum}

\subsubsection{Quantum simulation of linear ordinary and partial differential equations}
Any linear system of ordinary differential equations (ODEs) and linear partial differential equations (PDEs) can be represented in the following way:
\begin{align} \label{eq:differentialequation}
    \frac{d \mathbf{u}(t)}{dt}=-i \mathbf{A}(t) \mathbf{u}(t), \,\, {\hbox{with}}\,\,{\mathbf{u}(0)=\mathbf{u}_0,}
\end{align}
 For a system of $D$ linear ordinary differential equations for $D$ scalar functions $\{u_j(t)\}_{j=1}^D$ and with $\{|j\rangle\}_j$ an orthonormal basis, then $\mathbf{u}(t) \equiv \sum_{j=1}^D u_j(t)|j\rangle$ in~\eqref{eq:differentialequation} is a $D$-dimensional vector and $\mathbf{A}(t)$ is a $D \times D$ matrix. 
 We note that any inhomogeneous terms $\mathbf{f}$ and higher-order time derivatives can always be accommodated through an appropriate dilation, e.g., $\mathbf{u} \rightarrow \mathbf{u} \otimes |0\rangle + \mathbf{f} \otimes |1\rangle$ and $\mathbf{u} \rightarrow \mathbf{u} \otimes |0\rangle + d\mathbf{u}/dt \otimes |1\rangle$, respectively (for example, see~\cite{schr2, analog2023}). When $D$ is large enough, this system can also represent a discretised linear partial differential equation. For example, in Eulerian discretisation (e.g., finite difference methods), where $n$ is the size of the discretisation and $d$ is the spatial dimension of the PDE, so that $D \sim O(n^{d})$. We will discuss the continuous representation of the partial differential equation later. 
 
 A pure quantum state vector is the normalised vector $|u(t)\rangle \equiv \mathbf{u}(t)/\|\mathbf{u}(t)\|$, where normalisation is through the $l_2$ norm $\| \cdot \|$. For Schr\"odinger-like equations, $\mathbf{A}(t)$ is a Hermitian matrix, so that~\eqref{eq:differentialequation} can be directly approached through quantum simulation~\cite{lloyd1996universal,childs2018toward} with $\mathbf{A}(t)$ the corresponding Hamiltonian. This means that the transformation between $|u(0)\rangle$ and $|u(t)\rangle$ is unitary and the corresponding norms are preserved. For example, if~$\mathbf{A}$ is time-independent, we can simulate the final state $|u(t)\rangle=\exp(-i \mathbf{A} t)|u(0)\rangle$ through the unitary operator $\exp(-i \mathbf{A} t)$. However, for general ODEs, $\mathbf{A}(t)$ is \textit{not} necessarily Hermitian. This means that in order to simulate $|u(t)\rangle$ through quantum simulation, we need to dilate the state $\mathbf{u}(t)$ so that evolution in the dilated space is unitary. Various dilation procedures are possible, including Schr\"odingerisation~\cite{schr1, schr2, analog2023}, block-encoding~\cite{gilyen2019quantum, martyn2021grand,an2022theory}, or Stinespring dilation~\cite{busch2016dilation}. While explicit procedures are necessary for actual implementation of these methods, we will see that the usefulness of sample complexity in quantum hypothesis testing lies in providing us with a lower bound on the optimal cost of any quantum algorithm to prepare~$|u(t)\rangle$~\cite{an2022theory}.

 The basic suggestive idea is the following (for example, see~\cite{an2022theory}). Suppose we begin with two different initial conditions $\mathbf{u}_0$ and $\tilde{\mathbf{u}}_0$ for the same differential equation in~\eqref{eq:differentialequation}. When $\mathbf{A}$ is Hermitian like in the example above, then the distance  is preserved; i.e.,  \begin{equation}
     \|\mathbf{u}(t)-\tilde{\mathbf{u}}(t)\|=\|\exp(-i\mathbf{A}t)(\mathbf{u}_0-\tilde{\mathbf{u}}_0)\|=\|\mathbf{u}_0-\tilde{\mathbf{u}}_0\|.
 \end{equation}
 Thus, the distinguishability of the two initial states and the distinguishability of the two final states do not change since unitary evolution preserves $\|\mathbf{u}(t)-\tilde{\mathbf{u}}(t)\|$ for all $t$. However, in more general cases, $\mathbf{A}=\mathbf{A}_1+i\mathbf{A}_2$ is \textit{not} Hermitian (where $\mathbf{A}_1$, $\mathbf{A}_2$ are Hermitian). Under the assumption that $\mathbf{A}_1$ and $\mathbf{A}_2$ commute, define the ratio~$R$ as follows:
 \begin{equation}
 R \equiv \|\mathbf{u}(t)-\tilde{\mathbf{u}}(t)\|/\|\mathbf{u}_0-\tilde{\mathbf{u}}_0\|=\|\exp(\mathbf{A}_2 T)(\mathbf{u}_0-\tilde{\mathbf{u}}_0)\|/\|\mathbf{u}_0-\tilde{\mathbf{u}}_0\|.    
 \end{equation}
 Then if we assume the spectral norm of $\exp(\mathbf{A}_2 t)$ is large -- i.e., $\max_{\mathbf{v} \neq \mathbf{0}} \left( \| \exp( \mathbf{A}_2 t) \mathbf{v} \| / \| \mathbf{v} \| \right)= \left(\max \text{eigenvalue of } \exp(2\mathbf{A}_2 t) \right)^{1/2}$ is large -- then $R$ grows with $t$ when we assume $\mathbf{v}=\mathbf{u}_0-\tilde{\mathbf{u}}_0$. This means that an initially hard to distinguish pair $(\mathbf{u}_0, \tilde{\mathbf{u}}_0)$ can become easily distinguishable when $t$ is large enough or when $\mathbf{A}_2$ has large enough eigenvalues. The latter case occurs when the differential equation is very dissipative, like the heat equation with high diffusion coefficients. In this case, a quantum simulation algorithm for~\eqref{eq:differentialequation} with non-Hermitian $\mathbf{A}$ can act as a state discriminator. Since the optimal sample complexity~$n^*$ in quantum hypothesis testing provides the optimal cost in performing state discrimination, this sample complexity~$n^*$ can be used to derive a lower bound on the cost for the successful quantum simulation of~\eqref{eq:differentialequation} with some maximum error probability. Note that here we have ignored many details and subtleties (like differences between distances between classical vectors and those between quantum states) to focus only on conveying the intuitive essence of the idea. The arguments above can be refined with the improved $n^*$ lower bounds in this paper and for time-dependent $\mathbf{A}(t)$. Although better bounds than those provided by quantum hypothesis testing are possible~\cite{an2022theory}, these require stronger assumptions than are usually considered in state discrimination -- assuming access to not only the unitary operator creating the state, which then effectively allows amplitude amplification.

 We can extend this reasoning also to linear PDEs in their continuous representation, without discretisation. In this case, Eq.~\eqref{eq:differentialequation} can be used to represent a $(d+1)$-dimensional PDE using the continuous-variable representation $\mathbf{u}(t) \equiv \int_{-\infty}^{\infty} u(t, x)|x\rangle dx$, where we use $x=(x_1,\ldots,x_d)$ to denote the $d$-dimensional spatial degrees of freedom in the PDE. Here~$|x_j\rangle$ is a position eigenstate with corresponding position operator $\hat{x}_j$, where $\hat{x}_j|x_j\rangle=x_j|x_j\rangle$. In this case $\mathbf{A}$ is no longer a finite matrix like for ODEs, but it is rather an infinite-dimensional operator acting on $\mathbf{u}(t)$. Here $\mathbf{A}(t)$ involves both  position and momentum operators $\hat{x}_j, \hat{p}_j$, for $j\in \{1, \ldots, d\}$, obeying the commutation relations $[\hat{x}_j, \hat{p}_j]=iI$. To derive the form of $\mathbf{A}(t)$ from the original PDE, it can be shown, for example in~\cite{analog2023}, that one only needs to make the replacement $x u(t,x)  \rightarrow \hat{x} \mathbf{u}(t)$ and $\partial^k u(t,x)/\partial^k x \rightarrow (i\hat{p})^k \mathbf{u}(t)$. In fact, Eq.~\eqref{eq:differentialequation} can also represent a system of linear PDEs, by using a hybrid continuous-variable and discrete variable representation~\cite{analog2023}.
 
 It is important to note that the sample complexities, denoted by $n^*$, presented in this paper are also applicable to continuous variables as well as hybrid systems. Thus the basic argument for deriving lower bounds on the optimal cost in quantum simulation of PDEs and system of PDEs using the sample complexity in quantum hypothesis testing follows. However, given subtleties with infinite-dimensional systems, these ideas need to be refined and further explored. 

 So far we have only considered the embedding of scalar degrees of freedom into vectors~$\mathbf{u}(t)$, both finite and infinite dimensional. However, there are also differential equations for matrices $\Xi(t)$, for instance
 \begin{align} \label{eq:matrixode}
     \frac{d \Xi(t)}{dt}=\mathcal{L}(\Xi(t)), \qquad \Xi(0)=\Xi_0,
 \end{align}
 where $\mathcal{L}$ is a linear superoperator. We can consider the embedding of this matrix into an unnormalised density matrix, and quantum simulation for density matrices can be subsequently used. Similar arguments now based on quantum hypothesis testing for mixed states could potentially then be applied in this case -- in certain regimes -- to identify lower bounds on the optimal cost in the quantum simulation of~\eqref{eq:matrixode}. 

\textbf{Implications and future directions:} Previously, the sample complexity of symmetric binary quantum hypothesis testing was known only for pure states having a  uniform prior probability. This has been used as a major ingredient in providing lower bounds on the cost of successful quantum simulation of linear ODEs in~\eqref{eq:differentialequation}; see \cite[Lemma~7]{liu2021efficient}.  As a future direction, one can investigate applying the general sample complexity bounds derived in our work to establish lower bounds on the cost of successful quantum simulation of differential equations in terms of the matrices in~\eqref{eq:matrixode}.

\textbf{Linear algebra:} Another class of related applications is to consider quantum algorithms for linear systems of equations, for example~\cite{harrow2009quantum}, which exploits quantum phase estimation. Alternative methods of solving these problems through a dynamical process like~\eqref{eq:differentialequation} is also possible. This is done by mapping the discrete problem onto iterative algorithms and then taking the continuous-time limit~\cite{quantumlinearalgebra2023}. Thus, lower bounds for optimal algorithms for quantum algorithms for linear algebra can also be approached through $n^*$ and can be further explored.

\subsubsection{Quantum simulation of nonlinear ordinary and partial differential equations}

Previously we saw that the optimal sample complexity in quantum hypothesis testing is only relevant for linear differential equations when the dynamics served to increase the distance between two initially closely-separated states, thus making them easier to distinguish. This dynamics cannot include, for instance, purely unitary transformations. In the presence of \textit{nonlinear} dynamics, on the other hand, initially closely-separated states can be driven apart very quickly. The rate of separation is related to the Lyapunov coefficient of the dynamical system and depends on the degree of nonlinearity of the differential equation. For example, see~\cite{vulpiani2009chaos}.

Suppose that two initial conditions $\mathbf{u}_0$ and $\tilde{\mathbf{u}}_0$ satisfy $\|\mathbf{u}_0-\tilde{\mathbf{u}}_0\|^2 \sim \eta_{|\lambda|}(0)$ and they evolve in time $t$ under the same nonlinear dynamics characterised by some parameter $|\lambda|$ (we can let $|\lambda|=0$ denote purely linear and unitary dynamics). We assume $\eta_{|\lambda|}(0)$ to be a constant independent of $|\lambda|$. However, as time grows, $\|\mathbf{u}(t)-\tilde{\mathbf{u}}(t)\|^2 \sim \eta_{|\lambda|}(t)$, so that $\eta_{|\lambda|}(t)$ does depend on $|\lambda|$ for $t >0$. For many nonlinear dynamics of interest, $\eta_{|\lambda|}(t)$ increases with increasing $t$ and $|\lambda|$, thus making the states easier to distinguish. Thus the nonlinear dynamics can serve as a state discriminator. We can now approach this in a similar way to previous arguments in the linear non-unitary dynamics scenario. Assuming $\|\mathbf{u}(t)\| \sim  \|\tilde{\mathbf{u}}(t)\|$ and we embed $\mathbf{u}(t)$ and $ \tilde{\mathbf{u}}(t)$ into the amplitudes of pure quantum states with density matrices $\rho$ and $ \tilde{\rho}$, respectively, then $\|\mathbf{u}(t)-\tilde{\mathbf{u}}(t)\|^2 \sim d^2_H(\rho, \tilde{\rho}) \|\mathbf{u}(t)\|^2$. Ignoring the normalisation constant, we therefore see that the optimal sample complexity $n^* \sim 1/d_H^2 \sim 1/\eta_{|\lambda|}(t)$ in distinguishing quantum states $|u(t)\rangle$ and $|\tilde{u}(t)\rangle$ up to some maximum failure probability $\varepsilon$ provides a lower bound on the quantum simulation of nonlinear dynamics.   It can also place an upper bound on the amount of nonlinearity in a differential equation if we require the quantum simulation to be efficient for that nonlinear dynamics. In an illustrative sense, imposing efficiency with respect to $t$ for the optimal quantum algorithm requires $n^* \sim 1/\eta_{|\lambda|}(t) \lesssim \operatorname{poly}(t)$, where $\eta_{|\lambda|}(t)$ can  in principle be derived or approximated from the known nonlinear dynamics. This therefore puts an upper bound on $|\lambda|$. Thus for nonlinear ODEs with high enough nonlinearity, $\eta \sim \exp(-t)$ is possible, implying $n^* \sim \exp(t)$, which means any quantum algorithm is inefficient for high enough nonlinearities. For example, in~\cite{liu2021efficient}, a system of two differential equations with quadratic nonlinearity was studied, and such an upper bound for a characterisation of nonlinearity was identified.

Given more precise bounds $n^*$ in the current paper, previous arguments can be refined and explored further. The analysis can also be extended to  continuous-variable settings, as well as embedding of the nonlinear dynamics into density matrices, instead of pure states, by exploiting our optimal sample complexity bounds in these scenarios. Thus we can also extend to applications beyond~\eqref{eq:matrixode} to nonlinear differential matrix equations
\begin{align} \label{eq:nonlinearmatrixode}
     \frac{d \Xi(t)}{dt}=\mathcal{N}(\Xi(t)), \qquad \Xi(0)=\Xi_0,
 \end{align}
where $\mathcal{N}$ can be a nonlinear superoperator. Nonlinear differential equations for matrices include important classes like the Riccati differential equation, which appear in many areas ranging from optimal control~\cite{abou2012matrix}, estimation problems~\cite{nikoukhah1990kalman}, and network theory~\cite{anderson1999riccati}. It is also interesting to explore the role of the maximum error probability $\varepsilon$ and the interplay with the characterisation of nonlinearity.

\textbf{Implications and future directions:} Similar to the quantum simulation of linear ODEs and PDEs, as a future direction, one can investigate applying the sample complexity bounds derived for general states (including mixed states) to establish optimality bounds on the quantum simulation of nonlinear differential matrix equations in~\eqref{eq:nonlinearmatrixode}.
\smallskip

\textbf{Nonlinear quantum mechanics}: We can also find applications of sample complexity bounds in quantum hypothesis testing to determining the regime of validity for quantum simulation via effective nonlinear quantum mechanics~\cite{abrams1998nonlinear, childs2016optimal}. For example, one can find an upper bound on the time for which certain nonlinear quantum mechanical models remain valid. For instance, it is known that in the nonlinear Gross--Pitaevskii model with strength~$g$, it is possible to distinguish two states separated by distance $\eta$ in time $t \sim (1/g) \ln(1/\eta)$. We know from the optimal sample complexity that $n^* \sim 1/\eta$, where $n^*$ can also be interpreted as the number of particles whose effective dynamics obeys the Gross--Pitaevskii equation. This means that $t \lesssim (1/g) \ln n^*$~\cite{childs2016optimal}. Further exploration with more models and using more precise $n^*$ bounds and including the maximum failure probability could give further insight on both emergent nonlinearity in quantum mechanics and the computational power of different physical models. 

We emphasize that in the above applications, we are not using a quantum state discrimination protocol in order to perform quantum simulation for differential equations. Rather, the observation is that if we do have good enough quantum algorithms for simulating differential equations with two close initial conditions, then this algorithm is also sufficient for good enough state discrimination between two of the final states when the running time for the differential equation is long enough. Thus, the query complexity in the quantum algorithm for the differential equation provides an upper bound to the cost in a state discrimination problem to distinguish those same final states, or equivalently the sample complexity provides a lower bound on quantum differential equation solvers.

\subsubsection{Unstructured search and related applications}

In a typical search problem for an unstructured dataset, the task is to identify an unknown value $r$ that is an integer in the set $\{1, \ldots, d\}$ (for example, for a given function $f$, one seeks a value of $r$ for which $f(r)=1$). Grover's algorithm is a well-established quantum algorithm for this task~\cite{grover1996fast}, and it accomplishes the desired goal by preparing an approximation of the unknown state~$|r\rangle$ when we have access to an oracle that can validate whether or not we have the correct state~$|r\rangle$. In a continuous-time version of Grover's algorithm, the oracle access one assumes is a unitary generated by the Hamiltonian $H=|r\rangle\! \langle r|+|s\rangle\! \langle s|$ acting on a $d$-dimensional complex Hilbert space, where $|s\rangle=(1/\sqrt{d})\sum_{j=1}^d |j\rangle$. Evolving $|s\rangle$ by $\exp(-iH t)$, it can be easily shown that the probability of the final state being $|r\rangle$ is equal to one when $t \sim \sqrt{d}$~\cite{farhi1998analog}.

Suppose that we modify this search problem into a binary problem, where we ask whether or not this $r$ value even exists (e.g., whether a solution to $f(\cdot)=1$ exists for a given $f$). There are then two different corresponding Hamiltonians $H_0=|r\rangle\! \langle r|+|s\rangle\! \langle s|$ and $H_1=|s\rangle\! \langle s|$. Solving the binary version of the search problem means we want to distinguish between the oracles $U_0=\exp(-iH_0 t)$ and $U_1=\exp(-iH_1 t)$, and we want to distinguish between them by using as few samples of the initial states (and thus also access to the unitaries) as possible. The analysis in~\cite{farhi1998analog} reduces the optimality question to a state discrimination problem between states with distance $\eta$. For example, we can use a Hadamard test (see~\cite{childs2016optimal} and similar ideas in~\cite{abrams1998nonlinear}), which is a circuit involving either controlled-$U_0$ or controlled-$U_1$. The overlap of the two possible final states for the Hadamard test circuit can approach a constant for an appropriate choice for $t$ like $t \sim \sqrt{d}$. The optimal sample complexity in quantum hypothesis testing $n^* \sim \ln(1/\varepsilon)/\eta$ and therefore produces a lower bound on sample complexity for the basic discrimination problem with a maximum failure probability~$\varepsilon$. See~\cite{kimmel2017hamiltonian} for a similar discussion. Extensions to the continuous-variable scenario and to mixed quantum states could also be valuable.

\textbf{Implications and future directions:} The modification of the unstructured search problem into the binary problem above can also be mapped onto other problems, where the optimal sample complexity in quantum hypothesis testing provides a lower bound on the optimal complexity for those algorithms. These other problems include quantum fingerprinting~\cite{buhrman2001quantum}, orthogonality testing~\cite{kimmel2017hamiltonian}, and also density matrix exponentiation~\cite{kimmel2017hamiltonian, lloyd2014quantum}.

\subsubsection{Hidden subgroup problem}

The hidden subgroup problem remains one of the most prominent classes of problems for which quantum algorithms have been developed, including Shor's factoring algorithm~\cite{shor1994algorithms}, Shor's quantum algorithm for discrete logarithms~\cite{shor1994algorithms}, and also the earliest quantum algorithms like Simon's algorithm~\cite{simon1997power} and the Deutsch--Jozsa algorithm~\cite{deutsch1992rapid}. Here we are given a group $G$, and we want to find a generating set for a subgroup $H \subset G$, where $H$ is a so-called `hidden subgroup'. For a finite set $S$, a function $F: G \rightarrow S$ is said to `hide' the subgroup $H$ if $F(g_1)=F(g_2)$ for all $g_1, g_2 \in G$ if and only if $g_1 H=g_2 H$. Then if $F$ is given via an oracle, the task is to determine $H$ through minimal number of queries to the oracle. This problem can actually be rephrased as a multiple quantum hypothesis testing problem~\cite{harrow2012many}.  Here we are given a coset state defined by $\rho_H=(1/ |G|)\sum_{g \in G}|g H\rangle\! \langle gH|$ where $|gH\rangle=(1/\sqrt{|H|})\sum_{h \in H}|gh \rangle$. Let the number of subgroups be $M$. Then the goal of identifying $H \subset G$ involves using a minimal number of samples to distinguish between $M$ different coset states. Thus the bound $n^* \lesssim \ln(M)$ in quantum hypothesis testing for mixed states can be used to upper bound the query complexity $\sim \ln(M)$ for solving the hidden subgroup problem. Using the more precise bounds for $n^*$ obtained in this paper, the bounds for different hidden subgroup problems can then be refined. 

\textbf{Implications and future directions:} The precise sample complexity bounds for hypothesis testing obtained in this paper are potential key tools in providing sharper optimality bounds on  quantum algorithms designed to solve different hidden subgroup problems, beyond those already established in \cite{harrow2012many}.

\subsection{Quantum learning and classification} \label{sec:quantumlearning}

The goal of quantum classification is to classify an unknown quantum state $\sigma$ into one of $M$ classes, where we are not given any prior classical information about $\sigma$. That is, we would like to identify conditions under which different states cannot be discriminated (i.e., belong to the same class), whereas the goal of state discrimination is to identify how to discriminate states. The task of quantum classification requires the construction of a quantum classifier, which can come in different forms: (a) deterministic or (b) probabilistic. A deterministic classifier gives a deterministic outcome for the predicted class whereas a probabilistic classifier gives a stochastic output for the predicted class. For our current purpose, we only discuss the following probabilistic quantum classifier, for illustrative purposes. 

\begin{definition}[Probabilistic quantum classifier] \label{def:quantumclassifier} Suppose that for any input quantum state $\sigma^{\otimes n}$ where $\sigma$ is selected from a given (can be unknown) distribution $\mathcal{D}$, it is guaranteed that $\sigma$ already belongs to one of $M$ distinct classes, labelled $c_{\sigma} \in \{1, \ldots, M\}$ (true label of $\sigma$). We can define a probabilistic quantum classifier for $\sigma$ with respect to a given $n$ by a set $\{\Lambda^{(n)}_k\}_{k=1}^M$ with the following properties. (i) The set forms a POVM, so that $\Lambda^{(n)}_{k}\geq 0$ for all k and $\sum_{k=1}^M \Lambda^{(n)}_{k}=I^{\otimes n}$; (ii) $\Tr[\Lambda^{(n)}_{k} \sigma^{\otimes n}]$ corresponds to the probability that the predicted label of $\sigma$ is $k$.
\end{definition} 
We remark that more generally,  $c_{\sigma}$ just needs to belong to a set with cardinality $M$, but we choose the set $\{1, \ldots, M\}$ throughout for simplicity. 

In the above definition we have not defined what it means to have a successful classifier, which we will discuss later. For example, a perfect classifier with respect to a given $n$ is such that, for every quantum state $\sigma$ selected from the distribution $\mathcal{D}$, the conditions $\text{Tr}[\Lambda_{k=c_{\sigma}} \sigma^{\otimes n}]=1$ and $\text{Tr}[\Lambda_{k \neq c_{\sigma}} \sigma^{\otimes n}]=0$ hold. However, this is usually not possible. For whatever figures of merit for success, the construction of such a POVM requires some information about the classes. In the context of supervised learning problems, we are given training data, which consists of a set of states and their corresponding known (true) labels. 

\begin{definition}[Training dataset]\label{def:trainingdataset}
     The training dataset with $N$ states for the classification problem with $M$ classes is the set $\mathcal{T}_N=\{(s_i, c_i)\}_{i=1}^N$, where $s_i$ is a state and $c_i$ is its known (true)  corresponding label where $c_i \in \{1, \ldots, M\}$.  
There are three main categories of  training data $\mathcal{T}$ that we can consider for our quantum classification problem. 
\begin{enumerate}
    \item $s_i$ is the full classical description of the corresponding quantum state $\Sigma_i$. We then call $\mathcal{T}_N$ a \textit{classical training data set}.

     \item $s_i$ is the quantum state $\Sigma_i$ itself and partial classical information is given. Here we call $\mathcal{T}_N$ a \textit{partially quantum training data set}.
     
    \item $s_i$ is the quantum state $\Sigma_i$ itself with no other information provided about $\Sigma_i$. Here we call $\mathcal{T}_N$ a \textit{fully quantum training data set}.
\end{enumerate}
\end{definition}
We note that, unless otherwise stated, $\Sigma_i$ can be either a finite-dimensional, infinite-dimensional (continuous-variable) quantum state, or even a hybrid (finite-dimensional and infinite dimensional) state. In our formulation, $c_i$ is  a classical label, and we do not here consider the more general case where $c_i$ can be a quantum state itself. 

Then our $M$-ary supervised quantum classification problem involves two steps. Given an unknown quantum state $\sigma$ selected from some distribution $\mathcal{D}$, the task is to assign to $\sigma$ one of $M$ possible distinct labels. The first step -- training (or learning) step -- is the construction of a (probabilistic) quantum classifier -- given access to a training data set $\mathcal{T}_N$. This usually involves an optimisation procedure after being given a figure of merit (loss function) to optimise. The second step -- classifying step -- is the computation of the predicted class of~$\sigma$ when given access to $\sigma$. These two steps -- training and classifying -- are distinct and therefore have different costs associated with them. We can state these costs in an informal way below. 

\begin{definition}[Training cost]
The cost in constructing (or learning) a quantum classifier is the cost required to build this classifier given access to $\mathcal{T}_N$, subject to a bound in precision for a given figure of merit. 
\end{definition}

\begin{definition}[Classifying cost]
\label{def:classifyingcost}  This is the minimal number of copies $n$ of the unknown quantum state $\sigma$ required to make a prediction of its class subject to a given bound in precision for a given figure of merit. 
\end{definition}

For example, for our probabilistic quantum classifier $\{\Lambda_k^{(n)}\}_{k=1}^M$, a minimum of $n$ copies is needed. 

At the moment we have not yet discussed the criteria to identify appropriate figures of merit. There are different figures depending on one's applications, requirements, and constraints. The criterion closest in spirit to error probability in quantum hypothesis testing is the notion of training error defined below. See~\cite{banchi2021generalization} for a  discussion on the training error below for $n=1$ and the relationship to quantum hypothesis testing. There are also other important figures of merit in learning apart from training error. This includes \textit{test error} and also \textit{generalisation error}. We will not go into these concepts here, but interested readers can refer to~\cite{hastie2009elements}. 

\begin{definition}[Training error]
\label{def:trainingerror}
     Suppose we are given the fully quantum training dataset $\mathcal{T}_N$ defined in Definition~\ref{def:trainingdataset} with $\Sigma_i=\rho_i^{\otimes n}$. Then the training error $\varepsilon$ of a probabilistic quantum classifier in Definition~\ref{def:quantumclassifier} is the probability that this classifier makes the wrong prediction, i.e.,
    \begin{align}
        \varepsilon_{\Lambda} \coloneqq \frac{1}{N}\sum_{(c_i, \rho_i^{\otimes n}) \in \mathcal{T}_N} \left(1-\operatorname{Tr}[\Lambda^{(n)}_{c_i}\rho^{\otimes n}_i]\right)=1-\frac{1}{N}\sum_{(c_i, \rho^{\otimes n}_i) \in \mathcal{T}_N} \operatorname{Tr}[\Lambda^{(n)}_{c_i}\rho^{\otimes n}_i].
\end{align}
\end{definition}

    If there are $M$ possible classes, then $c_j$ can only take $M$ possible different values. If we are given sufficient training data, whereby $N>M$, then there must exist an equivalence class of states $\mathcal{S}_j$ containing $\rho_i, \rho_j$ with $i \neq j$ where $c_i=c_j$. We can therefore relabel any $c_i$ with $i>M$ by $c_j$ with $j \leq M$ since there are only $M$ distinct labels. The size of this equivalence class we can label as $N_j$, where $j \in \{1,\ldots, M\}$ and $\sum_{j=1}^M N_j=N$. Furthermore, suppose we consider the partially and fully quantum training dataset such that all the quantum states in the same equivalence class are in fact identical states. This means $\rho_i=\rho_j$ for all $i,j$ where $c_i=c_j$. Thus, one quantum state defines its own class and $N_j$ corresponds to the number of copies of $\rho_j$ one has in $\mathcal{T}_N$ for $j\in \{1, \ldots, M\}$. Such a training set $\mathcal{T}_N$ is then equivalent to $\{N_j \text{ copies of } (c_j,  \rho_j^{\otimes n})\}_{j=1}^M$ with $\sum_{j=1}^M N_j=N$. If we are only allowed to select one state $\rho_j$ at a time and we are given $N$ chances to make the selection, then this is also equivalent to $\{p_j, \rho^{\otimes n}_j\}_{j=1}^M$ where $p_j=N_j/N$ is the probability that the state $\rho_j$ is selected. Note that in this scenario $c_j=j$ is a redundant label once $\rho_j$ is given. Clearly, this can also be considered the input of an $M$-ary quantum hypothesis testing scenario if we are allowed to take $n$ samples of the states $\rho_j$. 
    
    If the POVM $\{\Lambda^{(n)}_{c_j}\}_{j=1}^M$ is used for the quantum state discrimination problem, the corresponding expected error probability is 
\begin{align} \label{eq:quantumclasserror}
    p_{e, \Lambda}=1-\sum_{c_j=j =1}^M p_j \Tr [\Lambda^{(n)}_{c_j} \rho^{\otimes n}_j] =\varepsilon_{\Lambda} ,
\end{align}
which coincides with the training error probability in Definition~\ref{def:trainingerror}. This means that minimising $\varepsilon_{\Lambda}$ over $\{\Lambda^{(n)}_{c_j}\}_{j=1}^M$ (i.e., optimising the training error) is equivalent to minimising $p_{e, \Lambda}$ above, which corresponds to the optimal error probability in state discrimination. For example, when $M=2$ and $n=1$, the optimal training error coincides with the result given by the Helstrom--Holevo theorem in~\eqref{eq:Helstrom2} (also see~\cite{banchi2021generalization}). 

We note that if $\mathcal{T}_N$ were instead the classical training dataset, then each equivalence class $|\mathcal{S}_j|=1$ because each $s_i$ is in the same class, being identical and can be tested to be identical without using more resources. Thus $p_j=1/M$. 

\subsubsection{Sample complexity of symmetric hypothesis testing in learning}

To see the relevance of sample complexity for symmetric hypothesis testing in the learning context, we can consider the following. The optimal sample complexity corresponds to the optimal number of copies $n$ of $\sigma$ required to make a prediction to some maximum error probability $\varepsilon_{\Lambda} \leq \varepsilon$. Thus we can define optimal sample complexity in learning as follows
\begin{align} \label{eq:quantumlearningcomplexity}
    N^*(\varepsilon)\coloneqq \inf_{\{\Lambda_{c_1}^{(n)}, \ldots, \Lambda_{c_M}^{(n)}\}} \{ n \in \mathbb{N} : \varepsilon_{\Lambda}^{(n)}=1-\sum_{c_j=j=1}^M p_j \Tr[\Lambda_{c_j}^{(n)}\rho_j^{\otimes n}] \leq \varepsilon \},
\end{align}
by analogy with Definition~\ref{def:M-ary}.

We can interpret the constraint $\varepsilon_{\Lambda} \leq \varepsilon$ in~\eqref{eq:quantumlearningcomplexity} in different ways. In the context of interpreting $\varepsilon_{\Lambda}$ as training error, $N^*$ here is the optimal classifying cost in Definition~\ref{def:classifyingcost} when we fix an upper bound $\varepsilon$ on the training error. This upper bound can be motivated in different ways. Often for training data, a regularisation is required to prevent over-fitting. This happens when the training error might be very small or near perfect, but it could result in the classifier performing badly on new data not in the original training dataset. To prevent this, some non-zero training error is required. Thus there exists some $\varepsilon_{\text{reg}}$ where $\varepsilon_{\text{reg}} <\varepsilon_{\Lambda}$. This means that any upper bound should be chosen to satisfy $\varepsilon \geq \varepsilon_{\text{reg}}$. It is also possible in some scenarios to constrain the maximum allowed training error directly, so that $\varepsilon$ can be set to be a constant. 

In the case of binary classification where $M=2$, if $\mathcal{T}_N$ defined above coincides with the true distribution of states from which $\sigma$ is selected, then the $\varepsilon_{\Lambda} \leq \varepsilon$ condition actually corresponds to the classifier being called $\varepsilon$-\textit{approximately correct} in the context of  probably approximately correct (PAC) learning~\cite{Aar07, CH16a, Roc18, arunachalam2018optimal}. However, unlike PAC learning, our sample complexity is defined with respect to a given distribution of states, and so it is not identical to the sample complexity defined in PAC learning. 

In the case of arbitrary $M$ classes, the corresponding optimal classifying cost is $N^*(\varepsilon) \leq O(\ln(M))$, where we use the result in Theorem~\ref{theorem:sc_M-ary} and we ignore dependencies on all parameters except $M$. For example, we see that the optimal protocol from quantum hypothesis testing is exponentially more efficient in $M$ compared to the algorithm `classification via state discrimination' in~\cite{gambs2008quantum}, where $\sigma$ is pairwise compared to each $\rho_j$ in $\mathcal{T}_N$. Of course, this $O(\ln M)$ scaling is also present for classical states, so this improvement is not due to the effect of exploiting quantum correlations in collective measurements. This $O(\ln(M))$ scaling is recovered, for example, for pure states in a different algorithm having access to Helstrom measurements~\cite{gambs2008quantum}.

\textbf{Implications and future directions:}
In this work, we do not consider the training cost in constructing the classifier, but we evaluate the classifying cost when the classifier is given. We defer the study of analysing the cost in identifying the classifier to future work. It is important to note that the cost of selecting the best classifier  would vary depending on whether the dataset is classical, partially quantum, or fully quantum.

In this section, we showed that the optimal sample complexity of learning, as defined in~\eqref{eq:quantumlearningcomplexity}, is directly related to the sample complexity of $M$-ary hypothesis testing. 
Now, the connection to $M$-ary hypothesis testing leads to the upper bound $O\!\left( \ln (M)\right)$ on the optimal sample complexity of learning, 
for $M \geq 2$. This indeed provides insight into how  the number of classes in the learning problem impacts the learning process.

\subsubsection{Sample complexity of asymmetric hypothesis testing in learning}

Now we can consider the asymmetric hypothesis testing scenario. In symmetric hypothesis testing, the error probabilities due to a false positive (type~I error) result and a false negative (type~II error) result are weighted in the same way. This corresponds to the situation of binary quantum classification where, in the expected error probability in~\eqref{eq:quantumclasserror}, we have $p_j=1/2$ for $j \in \{1,2\}$. However, in asymmetric hypothesis testing, these two types of errors are treated differently. 

This asymmetric setting for hypothesis testing has been previously connected to robustness of quantum classifiers~\cite{weber2021optimal}. Intuitively, if the optimal hypothesis test performs poorly in distinguishing between the original state $\sigma$ and a perturbed state $\sigma'$, then a quantum classifier will likely categorise $\sigma, \sigma'$ into the same class and is thus robust against perturbations $\sigma \rightarrow \sigma'$. For example, it was shown in~\cite{weber2021optimal} that asymmetric quantum hypothesis testing provides the robust region around $\sigma$ when the optimal type~II error probability is greater than 1/2. This formalism involves constraining the type~I error probability while optimising the type~II error probability, which is the usual setup in asymmetric hypothesis testing.  

In the definition of optimal sample complexity in Definition~\ref{def:binary_asymmetric}, however, we must place constant upper bounds on both type~I and type~II errors and 
only optimise 
over the number of copies of the input quantum states put into the classifier. This corresponds to a more complex learning scenario, which is related but not identical to the asymmetric loss function introduced in~\cite{bach2006considering} (this corresponds to the scenario where $p_1 \neq p_2$). In our case, the sample complexity corresponds to the minimal classifying cost for the probabilistic quantum classifier we defined where we fix upper bounds to the error probability terms $p_1 \Tr[\Lambda^{(n)}_2\rho_1^{\otimes n}]$ and $p_2 \Tr[\Lambda^{(n)}_1 \rho^{\otimes n}_2]$ \textit{separately}. 

\subsubsection{Concluding remarks on sample complexity in learning}

Going beyond supervised learning problems to unsupervised problems -- where we are not provided with training data -- these can also be considered in the context of state discrimination problems. For an example, see~\cite{sentis2019unsupervised}. The focus is on finding the optimal single-shot protocol, which does not fit our sample complexity paradigm. 

It is an interesting general question to consider how learning protocols will change if we focus instead on optimising sample complexity. In the context of supervised problems, this is relevant to small sample learning, when we only have a limited number of training data to learning from. This is important when the available quantum data production is scarce.

\subsection{Foundational quantum information and quantum mechanics}

\label{sec:quantumfoundation}

Optimal sample complexity of quantum hypothesis testing could also be relevant to foundational quantum information and quantum mechanics. We have already touched upon works that study the use of nonlinear quantum mechanical systems to solve unstructured search problems. The stronger the nonlinearity, the more efficient the corresponding search algorithm is. However, given that we know that such search problems and their generalisations can be NP-complete and even in \#P, it is not likely that this can be solved efficiently even by a quantum computer, in polynomial time. Although quantum mechanics is believed to be fundamentally linear, it is not known if there could be a more fundamental theory that is nonlinear. Thus, bounds from optimal sample complexity could potentially put fundamental bounds on nonlinear quantum mechanics at the foundational level~\cite{abrams1998nonlinear}. 

\textbf{Implications and future directions:}
There are many areas where the minimum error probability in quantum hypothesis testing is relevant to the foundations of quantum information theory and quantum mechanics (see~\cite{bae2015quantum} for many examples). Applications range from understanding no-go theorems in the interpretation of quantum states~\cite{pusey2012reality} and  related to bounding types of  \texorpdfstring{$\psi$}{Psi}-epistemic theories~\cite{leifer2013maximally}, `reproducing' standard quantum mechanics from the larger class of general probabilistic theories~\cite{arai2023derivation}, the operational meaning of min-entropy~\cite{konig2009operational}, security proofs in quantum key distribution~\cite{leverrier2009unconditional}, the construction of dimension witnesses for quantum states~\cite{brunner2013dimension}, relationship to no-signalling~\cite{bae2015quantum} and quantum cloning~\cite{bae2006asymptotic}.

It is then interesting to ask whether the optimal sample complexity setting can give rise to intriguing new questions. For example, it is known that in some instances optimal asymptotic cloning is equivalent to optimal state discrimination~\cite{bae2006asymptotic, bruss2000phase}. Asymptotic cloning refers to the limit where the number of clones tend to infinity. Equivalence is only possible in the asymptotic limit because optimal state discrimination is in general imperfect, which leads to an imperfect cloning procedure. Suppose we fix a degree of imperfection allowed in a cloning process. Then focusing instead on the protocol for optimal sample complexity could provide an alternative imperfect cloning procedure. 

\subsection{Private quantum hypothesis testing}

\label{sec:priv-QHT}

In this section, we briefly remark how the results of the present paper have concretely influenced subsequent works on quantum local differential privacy \cite{nuradha2024contractionQLDP,cheng2024samplePrivate}. Before going into more detail, let us motivate this topic with some background. With the development of quantum technologies, it is vital to ensure the privacy of generated quantum states. Recently, statistical privacy frameworks for quantum data have been introduced, generalizing classical statistical privacy frameworks studied in~\cite{DMNS06,DR14,KM14,nuradha2022pufferfishJ}. These privacy frameworks provide \textit{provable privacy guarantees} to the fully quantum and hybrid quantum-classical setting~\cite{QDP_computation17, aaronson2019gentle,hirche2023quantum, NGW_QPP_24}. Quantum local differential privacy (QLDP) is one such framework, which ensures that two distinct quantum states passed through a private quantum channel are difficult to distinguish by a
measurement~\cite{hirche2023quantum} (see also~\cite{NGW_QPP_24}). Formally, we say that a quantum channel $\mathcal{A}$ satisfies $\varepsilon$-QLDP for $\varepsilon >0$ if the following condition holds: 
\begin{equation}
    \sup_{\rho,\sigma \in \mathcal{D}} E_{e^\varepsilon}\!\left( \mathcal{A}(\rho) \Vert \mathcal{A}(\sigma) \right) =0, 
\end{equation}
where the supremization is over all quantum states with the same dimension and \\ $E_\gamma(\rho \Vert \sigma) \coloneqq \Tr\!\left[(\rho-\gamma \sigma)_+ \right]$ is the hockey-stick divergence, defined for $\gamma \geq 1$, with $\left(  A\right)  _{+}\coloneqq \sum_{i:a_{i}\geq0}a_{i}|i\rangle\!\langle i| $
for a Hermitian operator $A$ with spectral decomposition $A = \sum_{i}a_{i}|i\rangle\!\langle i|$.

The sample complexity of hypothesis selection under privacy constraints imposed by classical local differential privacy has been studied extensively in the classical literature \cite{Structu_Optimal_Tests19,Private_H_Selec,Locally_private_HS20,Contraction_local_new24}. Related to the quantum setting, the sample complexity of quantum hypothesis testing under privacy constraints imposed by quantum local differential privacy have been studied quite recently in~\cite{nuradha2024contractionQLDP,cheng2024samplePrivate}. These  works utilize the non-private sample complexities derived in this work to quantify the cost of privatizing data samples. For instance, consider the binary symmetric hypothesis testing of two orthogonal states $\rho$ and $\sigma$. To this end, we need only one sample of the unknown state to declare whether it is $\rho$ or $\sigma$ (Remark~\ref{rem:trivial-conditions}). Now, let us consider the scenario where we are given access to the privatized states of the orthogonal input states, denoted as $\mathcal{A}(\rho)$ and $\mathcal{A}(\sigma)$, where $\mathcal{A}$ is a quantum channel that satisfies $\varepsilon$-QLDP. For this setting, even if we are allowed to use the best private channel $\mathcal{A}$ satisfying the stated privacy level, the above works showed that the smallest number of samples to achieve a fixed error in the hypothesis testing of the two states scaled as $\Theta\! \left(\left((e^\varepsilon +1)/(e^\varepsilon-1)\right)^2\right)$ for all $\varepsilon > 0$ \cite[Corollary~2]{nuradha2024contractionQLDP}. This essentially quantifies the cost that we have to pay in the hypothesis testing to ensure privacy in this setting.

 Similarly, our results provide a foundation for analyzing the impact on resources in information-constrained settings, as done with privacy for a fundamental statistical task of quantum hypothesis testing. In that respect, our work also opens up a new research direction, namely, to  characterize the sample complexity of information-constrained quantum hypothesis testing.

\section{Summary and future directions}

In this paper, we defined the sample complexity of quantum hypothesis testing for the symmetric binary setting, the asymmetric binary setting, and for multiple hypotheses. We gave a characterization of the sample complexity for the two aforementioned binary settings. We found in the symmetric binary setting that it depends logarithmically on the priors and inverse error probability and inversely on the negative logarithm of the fidelity. In the asymmetric binary setting, it depends logarithmically on the inverse type~II error probability and inversely on the quantum relative entropy, provided that the type~II error probability is sufficiently small. Due to non-commutativity in the quantum case, there is freedom in characterizing the sample complexity of symmetric binary quantum hypothesis testing, and we found that other measures like the Holevo fidelity or the broader family of $z$-fidelities also characterize it. Furthermore, collective measurement strategies or simpler product measurements followed by classical post-processing achieve the same sample complexity, indicating that strategies with radically different technological requirements achieve the same sample complexity, in contrast to previous findings in the information-theoretic setting. It is also interesting to extend quantum sample complexity bounds to the task of the discrimination of quantum channels. 

We also summarized 
many applications of sample complexity of quantum hypothesis testing, including topics as diverse as quantum algorithms for simulation and unstructured search, quantum learning and classification,  foundations of quantum mechanics, and private quantum hypothesis testing. Sample complexity plays an important role in establishing the fundamental limitations of the first two tasks. For simulation and search, the idea is that these algorithms  produce distinguishable states as output when run on different input states, so that they serve as state discriminators and are ultimately subject to the limitations on such state discriminators set by quantum hypothesis testing. Throughout Section~\ref{sec:apps}, we illustrated a number of questions to motivate future work in this direction.

\bigskip 
\textbf{Note:} While finalizing our paper, we noticed the independent and concurrent arXiv post~\cite{sample_complexity_classical2024}, which considers sample complexity of classical hypothesis testing in the symmetric and asymmetric binary settings. We note here that all of our results apply to the classical case by substituting commuting density operators that encode probability distributions along their diagonals.

\begin{acknowledgments}
We thank Cl{\'e}ment~Canonne for several helpful discussions. We are grateful to Marco Fanizza for notifying us of an error in a previous version of Corollary~\ref{corollary:binary-asymmetric-informal}. 
ND acknowledges support from the Engineering and Physical Sciences Research Council [Grant Ref: EP/Y028732/1].
HCC acknowledges support from National Science and Technology Council under grants 112-2636-E-002-009, 113-2119-M-007-006, 113-2119-M-001-006, NSTC 113-2124-M-002-003, and from  Ministry of Education under grants NTU-112V1904-4, NTU-CC-113L891605, and NTU-113L900702. NL acknowledges funding from the Science and Technology Program of Shanghai, China (21JC1402900), the NSFC grants No. 12341104 and No. 12471411, the Shanghai Jiao Tong University 2030 Initiative, and the Fundamental Research Funds for the Central Universities. TN  acknowledges support from the NSF under grant no.~2329662. RS acknowledges funding from the Cambridge Commonwealth,
European and International Trust and from the European Research Council (ERC Grant AlgoQIP, Agreement No. 851716). MMW acknowledges support from the NSF under grants 2329662, 2315398, 1907615 and is especially grateful for an international travel supplement to 1907615. This supplement supported a trip to visit the quantum information group at Cambridge University in October 2022, during which this research project was initiated.
\end{acknowledgments}

\bibliography{Ref}

\appendix

\section{Properties of and relations between quantum divergences}

\label{sec:proofs-divergences}

\begin{lemma} \label{lemma:Fact}
	Let $A$, $B$, and $C$ be arbitrary positive semi-definite operators, and let $\rho$ and $\sigma$ be states. The following hold.
	\begin{enumerate}[(i)]
		\item\label{fact:additivity} Multiplicativity and Additivity:
		For every $n\in\mathbb{N}$,
		\begin{align}
			F\!\left(\rho^{\otimes n}, \sigma^{\otimes n}\right) 
			&= \left[ F(\rho,\sigma) \right]^n,
			\\
			F_\mathrm{H}\!\left(\rho^{\otimes n}, \sigma^{\otimes n}\right) 
			&= \left[ F_\mathrm{H}(\rho,\sigma) \right]^n,
			\\
			Q_\textnormal{min}\!\left(\rho^{\otimes n}, \sigma^{\otimes n}\right) 
			&= \left[ Q_\textnormal{min}(\rho,\sigma) \right]^n,
			\\
			\label{eq:Bures_additivity}
			1 - \frac12\left[ d_\mathrm{B}\!\left(\rho^{\otimes n}, \sigma^{\otimes n}\right) \right]^2
			&= \left( 1 - \frac12\left[ d_\mathrm{B}(\rho,\sigma) \right]^2 \right)^n,
			\\
			\label{eq:Hellinger_additivity}			
			1 - \frac12\left[ d_\mathrm{H}\!\left(\rho^{\otimes n}, \sigma^{\otimes n}\right) \right]^2
			&= \left( 1 - \frac12\left[ d_\mathrm{H}(\rho,\sigma) \right]^2 \right)^n,
			\\
			C\!\left(\rho^{\otimes n}\| \sigma^{\otimes n}\right)
			&= n C(\rho\|\sigma).
		\end{align}
	
		\item\label{fact:relation_fidelity}
		Relations between distances~\cite{Ara90, LT76, ANS+08}:
		\begin{align} 
  \label{eq:relation_fidelity}
			F_\mathrm{H}(\rho,\sigma)
			&\leq F(\rho,\sigma)
			\leq Q_{\min}(\rho\|\sigma)
			\leq \sqrt{F_\mathrm{H}}(\rho,\sigma),
			\\
			\label{eq:relation_Bures_Hellinger}
			d_\mathrm{B}(\rho,\sigma)
			&\leq
			d_\mathrm{H}(\rho,\sigma)
            \leq
            \sqrt{2} d_\mathrm{B}(\rho,\sigma),
			\\
			\label{eq:relation_Bures_Chernoff}
			 \left[ d_\mathrm{B}(\rho,\sigma) \right]^2
			&\leq - \ln F(\rho,\sigma)
			\leq -\ln F_\mathrm{H}(\rho,\sigma)			
			\leq 2 C(\rho\|\sigma)
			\leq -2\ln F(\rho,\sigma).
		\end{align}

		\item\label{fact:Chernoff} Quantum Chernoff bound 		\cite{ACM+07}:
		\begin{align} \label{eq:Chernoff}
			\frac{1}{2}\left( \Tr\!\left[A+B\right]-\left\|A-B\right\|_1 \right)
			\leq Q_\textnormal{min}(A\|B).
		\end{align}
		
		\item\label{fact:FG} Generalized Fuchs--van de Graaf inequality~\cite{Holevo1972fid,FG99, Aud14}:
		\begin{align} \label{eq:FG99}
			\Tr\!\left[A+B\right] - 2 \Tr\!\left[ \sqrt{A}\sqrt{B} \right]
			\leq \left\|A-B\right\|_1
			\leq \sqrt{ \left( \Tr[A+B]\right)^2 - 4\left\|\sqrt{A}\sqrt{B} \right\|_1^2 }.
		\end{align}
		
		\item\label{fact:PGM_MIN} A relation for the pretty-good test~\cite{Cheng2022}:
		\begin{align} \label{eq:Cheng2022}
			\Tr\!\left[ A \left(A+B\right)^{-1/2} B \left( A+B \right)^{-1/2} \right] \leq
			\frac12\Tr\!\left[A+B\right] - \frac12\left\|A-B\right\|_1.
		\end{align}
		
		\item\label{fact:subadditivity} Subadditivty~\cite{AM14}:
		\begin{align}
			\left\|\sqrt{A}\sqrt{B+C}\right\|_1 
			\leq \left\|\sqrt{A}\sqrt{B}\right\|_1 + \left\|\sqrt{A}\sqrt{C}\right\|_1.
		\end{align}

	\end{enumerate}
\end{lemma}
\begin{proof}
	\ref{fact:additivity}: The multiplicativity of $F$, $F_\mathrm{H}$, and $Q_{\min}$ follow by inspection.
	Eqs.~\eqref{eq:Bures_additivity} and~\eqref{eq:Hellinger_additivity} follow from the identities $1 - \frac12 \left[ d_\mathrm{B}(\rho,\sigma) \right]^2 = \sqrt{F}(\rho,\sigma)$ and $1 - \frac12 \left[ d_\mathrm{H}(\rho,\sigma) \right]^2 = \sqrt{F_\mathrm{H}}(\rho,\sigma)$.
	
	\ref{fact:relation_fidelity}:
	The first inequality of~\eqref{eq:relation_fidelity} follows from the Araki--Lieb--Thirring inequality~\cite{Ara90, LT76}, and the second from a~\cite[Theorem 6]{ANS+08} with $A = \rho$, $B = \sigma$, and $s = 1/2$:
	\begin{align} \label{eq:Audenaert_Thm6}
		\left\|\sqrt{A}\sqrt{B}\right\|_1
		\leq \left( \Tr\!\left[ A^s B^{1-s} \right] \right)^{1/2} \left( \Tr\!\left[A\right]\right)^{(1-s)/2} \left(\Tr\!\left[B\right]\right)^{s/2}, \quad \forall s\in[0,1].
	\end{align}
	The last inequality of~\eqref{eq:relation_fidelity} follows from definitions.
	We note that the relation $F(\rho,\sigma)
	\leq \sqrt{F_\mathrm{H}}(\rho,\sigma)$
	 was also shown in~\cite[Eq.~(38)]{IRS17}.
	 
	The first inequality of~\eqref{eq:relation_Bures_Hellinger} follows from~\eqref{eq:relation_fidelity}, and the second follows from
	\begin{align}
		\left[ d_\mathrm{H}(\rho,\sigma) \right]^2
		&= 2 \left( 1 - \sqrt{F_{\mathrm{H}}}(\rho,\sigma) \right)
		\\
		&\leq 2 ( 1 - F(\rho,\sigma))
		\\
		&= 2 \left( 1 - \sqrt{F}(\rho,\sigma) \right) \left( 1 + \sqrt{F}(\rho,\sigma) \right)
		\\
		&\leq 4 \left( 1 - \sqrt{F}(\rho,\sigma) \right) 
		\\
		&= 2\left[ d_\mathrm{B}(\rho,\sigma) \right]^2.
	\end{align}
    The first inequality of~\eqref{eq:relation_Bures_Chernoff}
    follows from $[d_{\mathrm{B}}(\rho\,\sigma)]^2 = 2 ( 1 - \sqrt{F}(\rho,\sigma) ) \geq -2 \ln \sqrt{F}(\rho,\sigma)$ because of $\ln x \leq x - 1$ for all $x>0$.
    The rest of inequalities in~\eqref{eq:relation_Bures_Chernoff}
    follow from~\eqref{eq:relation_fidelity}.

	\ref{fact:Chernoff}: The inequality was proved in~\cite[Theorem 1]{ACM+07}.
	
	\ref{fact:FG}: 
	The first inequality follows from~\eqref{eq:Chernoff} by choosing $s=1/2$ as a feasible solution in the definition of $Q_\text{min}$ (see also~\cite{PS70} and~\cite[Eq.~(2.6)]{ZW14}).
	The second inequality follows from~\cite[Lemma 2.4]{AM14} (see also~\cite[Theorem 5]{Aud14},~\cite[Supplementary Lemma 3]{CKW14} and~\cite[Eq.~(156)]{WBH+20}).
	
	\ref{fact:PGM_MIN}: The inequality was proved in~\cite[Lemma 1]{Cheng2022} (see also~\cite[Lemma 3]{SGC22a}).
	
	\ref{fact:subadditivity}: The inequality was proved in~\cite[Lemma 4.9]{AM14}.
	
\end{proof}

Recalling $\beta_{\varepsilon}(\rho\Vert \sigma)$ defined in~\eqref{eq:beta-err-asymm}, we have the following bounds relating it to quantum R\'enyi relative entropies.

\begin{lemma}[{\cite[Eq.~(75)]{MO14},~\cite[Proposition~7.71]{KW20},~\cite[Theorem 1]{CG24a}}] \label{lemma:strong_converse}
	Let $\rho$ and $\sigma$ be states.
	For all $\alpha \in (1,\infty)$ and $\varepsilon \in [0,1)$, the following inequality holds
	\begin{align}
		-\ln \beta_{\varepsilon}(\rho\Vert \sigma)
		\leq \widetilde{D}_{\alpha}(\rho\Vert \sigma)
		+\frac{\alpha}{\alpha-1} \ln\!\left( \frac{1}{1-\varepsilon} \right).
	\end{align}
\end{lemma}

\begin{lemma}[{\cite[Theorem 5]{ANS+08},~\cite[Proposition~3.2]{AMV12},~\cite[Proposition~7.72]{KW20},~\cite[Section IV.1]{Cheng2022}}] \label{lemma:Hoeffding}
	Let $\rho$ and $\sigma$ be states.
	For all $\alpha \in (0,1)$ and $\varepsilon \in (0,1]$, the following inequality holds
	\begin{align}
		-\ln \beta_{\varepsilon}(\rho\Vert \sigma)
		\geq {D}_{\alpha}(\rho\Vert \sigma)
		+\frac{\alpha}{\alpha-1} \ln\!\left( \frac{1}{\varepsilon} \right).
	\end{align}
\end{lemma}

\section{Efficient representation of optimal measurements and efficient calculation of optimal error probabilities in quantum hypothesis testing}

\label{app:efficient-algo-sdp}

In this appendix, we begin by showing that the optimal measurement in symmetric hypothesis testing, i.e., the optimiser  in 
\begin{align}
	\label{eq:perrsdp} p_e(p, \rho,q,\sigma,n)
 & = \min_{\Lambda^{(n)}} \left\{ p \Tr\!\left[ (I^{\otimes n}-\Lambda^{(n)})\rho^{\otimes n} \right] + q \Tr[\Lambda^{(n)} \sigma^{\otimes n}] : 0\leq \Lambda^{(n)} \leq I^{\otimes n}  \right\}\\&= p   -\max_{\Lambda^{(n)}}\left\{ \Tr[\Lambda^{(n)}\left(p\rho^{\otimes n}- q\sigma^{\otimes n}\right)] : 0\leq \Lambda^{(n)} \leq I^{\otimes n} \right\}
 \\
	&= \frac12 \left( 1 - \left\| p \rho^{\otimes n} - q \sigma^{\otimes n} \right\|_1 \right),
\end{align}
can be efficiently represented by $\operatorname{poly}(n)$ variables that can be computed in $\operatorname{poly}(n)$ time (see Lemma~\ref{lem:EfficHolevo} below). For that we essentially follow argument of~\cite{FawziIEEE22} and~\cite[Lemma~11]{BDS+24}. After doing so, we then remark that the same reasoning can be applied to asymmetric binary and multiple quantum hypothesis testing (see Remark~\ref{rem:asym-bin-mult-poly-time}).

Let us first introduce some notation. We denote the symmetric group by $\Sn$ and denote for every permutation $\pi\in\Sn$ the unitary $P_{\cH}(\pi)$ on $\cH^{\otimes n}$ which permutes the $n$ copies of~$\cH^{\otimes n}$ according to $\pi$. For any operator $X$  on $\cH^{\otimes n}$ we denote the group average by
        \begin{equation}
        \label{eq:PermAver}
            \overline{X} \coloneqq \frac{1} { \abs{\Sn}} \sum_{\pi \in \Sn} P_\cH(\pi) X P_\cH(\pi)^{\dagger}\, ,
        \end{equation}
        and the set of all permutation invariant operators by
        \begin{equation}
            \End(\cH^{\otimes n})\coloneqq \left\{A \in \mathcal{B}(\cH^{\otimes n}) | P_{\cH}(\pi) A P_{\cH}(\pi)^{\dagger} = A, \quad \forall \pi \in \Sn\right\} \,.
        \end{equation}
        where we denoted the set of bounded operators on $\cH^{\otimes n}$ by $\mathcal{B}(\cH^{\otimes n}).$ From standard representation theory of the symmetric group, we know that
        \begin{align*}
          \dim\End(\cH^{\otimes n}) \le (n+1)^{d^2}. 
        \end{align*}
        See~\cite[Section 4.1]{FawziIEEE22} and~\cite[Section 2.1]{litjens_semidefinite_2017}.
        The main idea is to note that the optimiser $\Lambda^{(n)}_\star$ of~\eqref{eq:perrsdp} is permutation invariant itself, i.e.
\begin{align}
  \Lambda^{(n)}_\star = \left\{p\rho^{\otimes n}- q\sigma^{\otimes n}\right\}  \in \End(\cH^{\otimes n}),
\end{align}
which shows that it can be described by $\poly(n)$ many classical variables. To find these variables efficiently, i.e. in $\poly(n)$ time, we use the techniques of~\cite{FawziIEEE22,BDS+24} to show that the SDP in~\eqref{eq:perrsdp} is equivalent to one with $\poly(n)$ many variables and constraints whose coffecients can be found from~\eqref{eq:perrsdp} in $\poly(n)$ time. This is the content of the following lemma:

\begin{lemma}
\label{lem:EfficHolevo}
The SDP 
 \begin{maxi}|l|
            {\Lambda^{(n)}\in\mathcal{B}(\cH^{\otimes n})}{\Tr[\Lambda^{(n)}\left(p\rho^{\otimes n}- q\sigma^{\otimes n}\right)]}
            {}{}
            \addConstraint{ 0\leq \Lambda^{(n)} \leq I^{\otimes n}}
            \label{eq:sdp_original}
        \end{maxi}
can be reduced to an equivalent SDP with $\dim\End(\cH^{\otimes n}) \le (n+1)^{d^2}$ many variables and $(n+1)^d$ many constraints. Moreover, the coefficients of the reduced SDP can be efficiently computed in $\poly(n)$ time from the original SDP~\eqref{eq:sdp_original}.
\end{lemma}

\begin{proof}
Note that, for any feasible $\Lambda^{(n)}$ in~\eqref{eq:sdp_original}, also $\overline{\Lambda^{(n)}}$ is feasible, as the group average~\eqref{eq:PermAver} is a positive map.  Furthermore, as by definition we have $p\rho^{\otimes n}- q\sigma^{\otimes n} \in  \End(\cH^{\otimes n}),$ we see \begin{align}\Tr[\Lambda^{(n)}\left(p\rho^{\otimes n}- q\sigma^{\otimes n}\right)] = \Tr[\overline{\Lambda^{(n)}}\left(p\rho^{\otimes n}- q\sigma^{\otimes n}\right)]\end{align} and hence the SDP~\eqref{eq:sdp_original} can equivalently be expressed by simply maximising over the smaller space $\End(\cH^{\otimes n}),$ i.e. by
 \begin{maxi}|l|
            {\Lambda^{(n)}\in\End(\cH^{\otimes n})}{\Tr[\Lambda^{(n)}\left(p\rho^{\otimes n}- q\sigma^{\otimes n}\right)]}
            {}{}
            \addConstraint{ 0\leq \Lambda^{(n)} \leq I^{\otimes n}}.
            \label{eq:sdp_original2}
        \end{maxi}
As pointed out above, this reduces number of variables over which we optimize over to $\dim\End(\cH^{\otimes n}) \le (n+1)^{d^2} = \poly(n).$ However, as in the current form the SDP~\eqref{eq:sdp_original2} is phrased in terms of exponentially large operators, this does not yet show computability in $\poly(n)$ time. To close that remaining gap we use the approach of~\cite[p. 7353]{FawziIEEE22} or~\cite[Lemma 11]{BDS+24}. There, the authors explicitly define an orthogonal basis (with respect to the Hilbert--Schmidt scalar product) denoted by $\left(C^{\cH}_r\right)_{r=1}^{\dim\End(\cH^{\otimes n})}$. For that particular basis and for all $A\in
\mathcal{B}(\cH)$, the authors provide an explicit formula for the coefficients $\left(\gamma_r\right)_{r=1}^{\dim\End(\cH^{\otimes n})}$ of the expansion
\begin{align}
    A^{\otimes n} = \sum_{r=1}^{\dim\End(\cH^{\otimes n})} \gamma_r C^{\cH}_r
\end{align}
  In particular it is shown that $\left(\gamma_r\right)_{r=1}^{\dim\End(\cH^{\otimes n})}$ can be computed in $\poly(n)$ time.
Hence, by linearity also the coefficients $\left(\kappa_r\right)_{r=1}^{\dim\End(\cH^{\otimes n})}$ of the operator $p\rho^{\otimes n} -q\sigma^{\otimes n}$ in this basis, i.e., satisfying
\begin{align}
    p\rho^{\otimes n} -q\sigma^{\otimes n} = \sum_{r=1}^{\dim\End(\cH^{\otimes n})} \kappa_r C^{\cH}_r
\end{align}
can be computed in $\poly(n)$ time. Furthermore, as explained in~\cite{FawziIEEE22,BDS+24}, the square of the norms of the basis operators, i.e. $\Tr[(C^{\cH}_r)^\dagger C^{\cH}_r]$, can be computed efficiently in $O(d^2)$, meaning constant in $n$, i.e. $O(1)$, time.

As $\left(C^{\cH}_r\right)_{r=1}^{\dim\End(\cH^{\otimes n})}$ is a basis, we can also expand the variable $\Lambda^{(n)}\in \End(\cH^{\otimes n})$ in the SDP~\eqref{eq:sdp_original2} as
\begin{align}
    \Lambda^{(n)} = \sum_{r=1}^{\dim\End(\cH^{\otimes n})} x_r C^{\cH}_r.
\end{align} 
As shown in~\cite{FawziIEEE22,BDS+24}, the constraint \begin{align}0 \le \sum_{r=1}^{\dim\End(\cH^{\otimes n})} x_r C^{\cH}_r \le I^{\otimes n}\end{align} in~\eqref{eq:sdp_original2} can be checked in $\poly(n)$ time as well: To see this, note that there exists a $*$-algebra isomorphism from $\End(\HS^{\otimes n})$ to block-diagonal matrices 
        \begin{equation}\label{eq:sdp_star_isomorphism}
            \phi_{\HS}: \End(\HS^{\otimes n}) \to \bigoplus_{i = 1}^{t_{\HS}} \mathbb{C}^{m_i \times m_i}\,,
        \end{equation}
        where
        \begin{equation}\label{eq:sdp_star_isomorphism_coefficients}
            t_{\HS} \leq (n + 1)^{d} \qquad \mathrm{and} \qquad \sum_{i = 1}^{t_\cH} m_i^2 = \dim(\End(\HS^{\otimes n})) \leq (n + 1)^{d^2}\,.
        \end{equation}
        For $A \in \End(\HS^{\otimes n})$ we have $\phi_{\HS}(A) \in \bigoplus_{i = 1}^{t_{\HS}} \mathbb{C}^{m_i \times m_i}$ and write
        $\llbracket \phi_{\HS}(A) \rrbracket_i$ for the $i$-th block of $\phi_{\HS}(A)$. By~\cite[Lemma 3.3]{FawziIEEE22} we have 
        \begin{equation}\label{eq:end_positivity}
        A \geq 0 \Leftrightarrow \phi_{\HS}(A) \geq 0 \Leftrightarrow \llbracket \phi_{\HS}(A) \rrbracket_i \geq 0 \quad \forall i \in \{1, \ldots, t_\cH\} \,.
        \end{equation}
        Hence,
        \begin{align}
        \label{eq:PositivityConstraint}
            0 \le \sum_{r=1}^{\dim\End(\cH^{\otimes n})} x_r C^{\cH}_r \le I^{\otimes n}  \ \Longleftrightarrow \ 0 \le \sum_{r=1}^{\dim\End(\cH^{\otimes n})}x_r\llbracket  \phi_{\HS}(C^{\cH}_r) \rrbracket_i \le I_{\mathbb{C}^{m_i}} \;\forall i \in \{1, \ldots, t_\cH\},
        \end{align}
        where we denoted the identity on $\mathbb{C}^{m_i}$ by $I_{\mathbb{C}^{m_i}}.$
        Furthermore, as shown in~\cite[pp. 7352–735]{FawziIEEE22}, for every basis element $C^{\cH}_r$ the corresponding block diagonal matrix $\phi_{\HS}(C^{\cH}_r)$ can be computed in $\poly(n)$ time. Overall, this shows that the constraint~\eqref{eq:PositivityConstraint} can be checked in $\poly(n)$ time.

Hence, the above shows that the SDP~\eqref{eq:sdp_original} can equivalently be written as 
 \begin{maxi}|l|
            {\left(x_r\right)_{r=1}^{\dim\End(\cH^{\otimes n})}\ \in \mathbb{C}^{\dim\End(\cH^{\otimes n})}}{\sum_{r=1}^{\dim\End(\cH^{\otimes n})}x_r(\kappa_r)^*\Tr[(C^{\cH}_r)^\dagger C^{\cH}_r]}
            {}{}
            \addConstraint{ 0 \le \sum_{r=1}^{\dim\End(\cH^{\otimes n})}x_r\llbracket  \phi_{\HS}(C^{\cH}_r) \rrbracket_i \le I_{\mathbb{C}^{m_i}} }\addConstraint{\forall i \in \{ 1, \ldots, t_\cH\}}
            \label{eq:sdp_original3}
        \end{maxi}
where all coefficients can be found in $\poly(n)$ from the original SDP~\eqref{eq:sdp_original}. As the SDP~\eqref{eq:sdp_original3} has $\poly(n)$ variables and $\poly(n)$ constraints, it can be solved with an additive error $\varepsilon$ in $O(\poly(n)\log(1/\varepsilon))$ time; see, e.g.,~\cite{jiang_faster_2020}.
\end{proof}

\begin{remark}
\label{rem:asym-bin-mult-poly-time}
By exact same line of argument as in the proof Lemma~\ref{lem:EfficHolevo},  we can also reduce the SDP
\begin{mini}|l|
     {\Lambda^{(n)}\in\mathcal{B}(\cH^{\otimes n})}{\Tr[\Lambda^{(n)}\sigma^{\otimes n}]}{}{}
            \addConstraint{0\le\Lambda^{(n)}\le I^{\otimes n}}
            \addConstraint{\Tr[\Lambda^{(n)}\rho^{\otimes n}]\ge 1- \varepsilon}
            \label{eq:sdp_originalAsym},
        \end{mini}        
        which was defined in~\eqref{eq:beta-err-asymm} in the context of asymmetric binary hypothesis testing, and
\begin{mini}|l|
            {\Lambda^{(n)}_1,...,\Lambda_m^{(n)}\in\mathcal{B}(\cH^{\otimes n})}{\sum_{m=1}^M p_m\Tr[(I^{\otimes n}-\Lambda^{(n)})\rho_m^{\otimes n}]}
            {}{}
            \addConstraint{ \Lambda_m^{(n)}\ge 0}\addConstraint{\sum_{m=1}^M\Lambda_m^{(n)}=I^{\otimes n}}{}
            \label{eq:sdp_originalM},
\end{mini}
        which was defined in~\eqref{eq:error_M-ary} in the context of $M$-ary hypothesis testing to equivalent SDPs with $\poly(n)$ classical variables and $\poly(n)$ constraints. Moreover, the coefficients of the reduced SDPs can again be efficiently computed in $\poly(n)$ time from the original SDPs. 

        This shows how also the optimal measurements of asymmetric binary hypothesis testing as well as the optimal measurements in $M$-ary hypothesis testing can be written in terms of $\poly(n)$ many classical variables that can be found efficiently in $\poly(n)$ time.
\end{remark}

\section{Proof of Equations~\eqref{eq:asymm-beta-rewrite-both-errs}--\eqref{eq:asymm-beta-rewrite-2}}

\label{app:proof-of-equiv-exps-asymm-err}

To prove Equations~\eqref{eq:asymm-beta-rewrite-both-errs}--\eqref{eq:asymm-beta-rewrite-2}, we use that the optimal error in asymmetric binary quantum hypothesis testing is attained. In finite dimensions this follows simply by compactness. In infinite dimensions the same proof goes through by noticing that we still have the right notion of compactness, i.e. weak-$^*$ compactness, due to the Banach--Alaoglu theorem. We first formally state and prove this fact in the following lemma and then continue to give the proof of Equations~\eqref{eq:asymm-beta-rewrite-both-errs}--\eqref{eq:asymm-beta-rewrite-2}.

\begin{lemma} \label{lem:BanackAlaoglu} Let $\cH$ be a separable Hilbert space, and let $\rho$ and $ \sigma$ be  quantum states on $\cH$. Fix $\eps\ge 0$. Then the infimum in the optimal error in asymmetric binary quantum hypothesis testing is attained, i.e. 
    \begin{align*}
        \beta_\eps(\rho\|\sigma) \coloneqq\inf_{\substack{0\le \Lambda\le I\\\Tr((I-\Lambda)\rho)\le \eps }}\Tr(\Lambda\sigma) =\min_{\substack{0\le \Lambda\le I\\\Tr((I-\Lambda)\rho)\le \eps }}\Tr(\Lambda\sigma).
    \end{align*}
\end{lemma}

\begin{proof}
By the definition of $\beta_\eps(\rho\|\sigma)$ there exists a sequence of bounded operators $\left(\Lambda_k\right)_{k\in\N}$ such that $0\le \Lambda_k\le I$ and $\Tr((I-\Lambda_k)\rho)\le \eps$ for all $k\in\N$ and 
\begin{align}
\label{eq:InfSequence}
    \lim_{k\to \infty} \Tr(\Lambda_k\sigma) =  \beta_\eps(\rho\|\sigma).
\end{align}
By the Banach–Alaoglu theorem there exists a subsequence $\left(\Lambda_{k_l}\right)_{l\in\N}$ that converges weakly-$^*$ to some bounded operator $\Lambda$; i.e., using that the dual of the trace class operators is the space of bounded operators we have for all trace class operators $T$ that
\begin{align}
    \lim_{l\to\infty}\Tr(\Lambda_{k_l}T) = \Tr(\Lambda T).
\end{align}
From that we immediately see that $0\le \Lambda\le I$ and $\Tr((I-\Lambda)\rho)\le \eps$ and furthermore by~\eqref{eq:InfSequence} also
\begin{align*}
    \Tr(\Lambda\sigma)= \lim_{l\to \infty} \Tr(\Lambda_{k_l}\sigma) =  \beta_\eps(\rho\|\sigma),
\end{align*}
which finishes the proof.
\end{proof}

\begin{proof}[Proof of Equations~\eqref{eq:asymm-beta-rewrite-both-errs}--\eqref{eq:asymm-beta-rewrite-2}]
Let us define
\begin{equation}
\beta_{\varepsilon}(\rho^{\otimes n}\Vert\sigma^{\otimes n}):=\inf
_{\Lambda^{(n)}}\left\{
\begin{array}
[c]{c}
\operatorname{Tr}[\Lambda^{(n)}\sigma^{\otimes n}]:\operatorname{Tr}\!\left[
\left(  I^{\otimes n}-\Lambda^{(n)}\right)  \rho^{\otimes n}\right]
\leq\varepsilon,\\
0\leq\Lambda^{(n)}\leq I^{\otimes n}
\end{array}
\right\}  .\label{eq:beta-quantity}
\end{equation}
Our goal here is to prove that
\begin{equation}
\inf\left\{  n\in\mathbb{N}:\beta_{\varepsilon}(\rho^{\otimes n}\Vert
\sigma^{\otimes n})\leq\delta\right\}  =\inf\left\{
\begin{array}
[c]{c}
n\in\mathbb{N}:\operatorname{Tr}[\Lambda^{(n)}\sigma^{\otimes n}]\leq\delta,\\
\operatorname{Tr}\!\left[  \left(  I^{\otimes n}-\Lambda^{(n)}\right)
\rho^{\otimes n}\right]  \leq\varepsilon,\\
0\leq\Lambda^{(n)}\leq I^{\otimes n}
\end{array}
\right\}  .
\end{equation}
Let $n$ be a feasible choice for the optimization on the left-hand side,
meaning that for this choice of $n$, an optimal choice $\Lambda^{(n)}$ (which exists due to Lemma~\ref{lem:BanackAlaoglu}) for the
optimization $\beta_{\varepsilon}(\rho^{\otimes n}\Vert\sigma^{\otimes n})$
satisfies $\operatorname{Tr}[\Lambda^{(n)}\sigma^{\otimes n}]=\beta
_{\varepsilon}(\rho^{\otimes n}\Vert\sigma^{\otimes n})\leq\delta$. Then, by
the definition in~\eqref{eq:beta-quantity} and the constraint
$\operatorname{Tr}[\Lambda^{(n)}\sigma^{\otimes n}]\leq\delta$, this means
that $n$ and $\Lambda^{(n)}$ are feasible for the optimization on the
right-hand side. So we conclude that
\begin{equation}
n\geq\inf\left\{
\begin{array}
[c]{c}
n\in\mathbb{N}:\operatorname{Tr}[\Lambda^{(n)}\sigma^{\otimes n}]\leq\delta,\\
\operatorname{Tr}\!\left[  \left(  I^{\otimes n}-\Lambda^{(n)}\right)
\rho^{\otimes n}\right]  \leq\varepsilon,\\
0\leq\Lambda^{(n)}\leq I^{\otimes n}
\end{array}
\right\}  .
\end{equation}
Since this argument holds for every feasible choice of $n$ on the left-hand
side, we conclude that
\begin{equation}
\inf\left\{  n\in\mathbb{N}:\beta_{\varepsilon}(\rho^{\otimes n}\Vert
\sigma^{\otimes n})\leq\delta\right\}  \geq\inf\left\{
\begin{array}
[c]{c}
n\in\mathbb{N}:\operatorname{Tr}[\Lambda^{(n)}\sigma^{\otimes n}]\leq\delta,\\
\operatorname{Tr}\!\left[  \left(  I^{\otimes n}-\Lambda^{(n)}\right)
\rho^{\otimes n}\right]  \leq\varepsilon,\\
0\leq\Lambda^{(n)}\leq I^{\otimes n}
\end{array}
\right\}  .
\end{equation}

Now let $n$ and $\Lambda^{(n)}$ be feasible choices for the optimization on
the right-hand side. Then it follows that $\Lambda^{(n)}$ is feasible for
$\beta_{\varepsilon}(\rho^{\otimes n}\Vert\sigma^{\otimes n})$.\ Since
$\operatorname{Tr}[\Lambda^{(n)}\sigma^{\otimes n}]\leq\delta$, we conclude
that
\begin{equation}
\delta\geq\operatorname{Tr}[\Lambda^{(n)}\sigma^{\otimes n}]\geq
\beta_{\varepsilon}(\rho^{\otimes n}\Vert\sigma^{\otimes n}).
\end{equation}
Thus, the choice of $n$ is feasible for the optimization on the left-hand
side. We then conclude that
\begin{equation}
\inf\left\{  n\in\mathbb{N}:\beta_{\varepsilon}(\rho^{\otimes n}\Vert
\sigma^{\otimes n})\leq\delta\right\}  \leq n,
\end{equation}
and since this inequality holds for every feasible choice of $n$ on the
right-hand side, we conclude the opposite inequality:
\begin{equation}
\inf\left\{  n\in\mathbb{N}:\beta_{\varepsilon}(\rho^{\otimes n}\Vert
\sigma^{\otimes n})\leq\delta\right\}  \leq\inf\left\{
\begin{array}
[c]{c}
n\in\mathbb{N}:\operatorname{Tr}[\Lambda^{(n)}\sigma^{\otimes n}]\leq\delta,\\
\operatorname{Tr}\!\left[  \left(  I^{\otimes n}-\Lambda^{(n)}\right)
\rho^{\otimes n}\right]  \leq\varepsilon,\\
0\leq\Lambda^{(n)}\leq I^{\otimes n}
\end{array}
\right\}  .
\end{equation}

By similar reasoning, we conclude that
\begin{equation}
\inf\left\{  n\in\mathbb{N}:\beta_{\delta}(\sigma^{\otimes n}\Vert
\rho^{\otimes n})\leq\varepsilon\right\}  =\inf\left\{
\begin{array}
[c]{c}
n\in\mathbb{N}:\operatorname{Tr}[\Lambda^{(n)}\sigma^{\otimes n}]\leq\delta,\\
\operatorname{Tr}\!\left[  \left(  I^{\otimes n}-\Lambda^{(n)}\right)
\rho^{\otimes n}\right]  \leq\varepsilon,\\
0\leq\Lambda^{(n)}\leq I^{\otimes n}
\end{array}
\right\}  .
\end{equation}
\end{proof}

\section{Proof of Remark~\ref{rem:trivial-conditions}}

\label{app:proof-remark-triv-cond}

    We begin by establishing~\eqref{eq:binary_symmetric1} under the assumption that $\varepsilon \in [1/2,1]$. 
	The trivial strategy of guessing the state uniformly at random achieves an error probability of $1/2$ (i.e., with $n=1$ and $\Lambda^{(1)} = I/2$ in~\eqref{eq:Helstrom1}). 
	Thus, if $\varepsilon \in [1/2,1]$, then the trivial strategy accomplishes the task with just a single sample. Alternatively, by considering the expression in~\eqref{eq:def:sample_complexity_symmetric3} and noting that the norm is always non-negative, if $\varepsilon \in [1/2,1]$, then the left-hand side of the inequality is $\leq 0 $ while the right-hand side is $\geq 0$, so that the constraint is trivially met and the smallest possible feasible $n \in \mathbb{N}$ is $n=1$.

    If $\rho\sigma = 0$, then we set the test $\Lambda^{(1)}$  to be the projection onto the support of $\rho$.
    In this case, $p_e = 0 $ and hence $n=1$ suffices.

	Now we establish~\eqref{eq:binary_symmetric1} under the assumption that there exists $s\in[0,1]$ such that $\varepsilon \geq p^s q^{1-s}$. 
	With a single sample, we can achieve the error probability $\frac12 \left( 1 - \left\| p \rho - q \sigma \right\|_1 \right)$, and by invoking Lemma~\ref{lemma:Fact}-\ref{fact:Chernoff} with $A = p\rho$ and $B = q\sigma$, we conclude that
	\begin{align}
		\frac12 \left( 1 - \left\| p \rho - q \sigma \right\|_1 \right)
		&\leq p^s q^{1-s} \Tr\!\left[ \rho^s \sigma^{1-s} \right]
		\\
		&\leq p^s q^{1-s}
		\\
		&\leq \varepsilon.
	\end{align}
	In the above, we used the H\"{o}lder inequality $|\Tr[AB] | \leq \|A\|_r \|B\|_t$, which holds for $r,t\geq 1$, such that $r^{-1} + t^{-1} = 1$, to conclude that
 \begin{equation}
     \label{eq:renyi-q-lte-1}
 \Tr[\rho^s \sigma^{1-s}]\leq 1.
 \end{equation}
	That is, set $r=s^{-1}, t = (1-s)^{-1}$, $A=\rho$, and $B=\sigma$.
	The last inequality follows by assumption.
	So we conclude that only a single sample is needed in this case also.
	
    We finally prove~\eqref{eq:binary_symmetric2}, which is the statement that if $\rho = \sigma$ and $\min\{p,q\}> \varepsilon \in [0,1/2)$, then $n^*(p,\rho,q,\sigma,\varepsilon) = +\infty$.
	This follows because it is impossible to distinguish the states and the desired inequality $\frac12 \left( 1 - \left\| p \rho^{\otimes n} - q \sigma^{\otimes n} \right\|_1 \right) \leq \varepsilon$ cannot be satisfied for any possible value of~$n$.
	Indeed, in this case,
	\begin{align}
		\frac12 \left( 1 - \left\| p \rho^{\otimes n} - q \sigma^{\otimes n} \right\|_1 \right)
		&= \frac12 \left(1 - |p-q| \left\|\rho^{\otimes n} \right\|_1 \right)
		\\
		&= \frac12 \left( 1 - |p-(1-p)| \right)
		\\
		&= \frac12 ( 1 - |1-2p|).
	\end{align}
	When $p\leq 1/2$, we have that $\frac12 ( 1 - |1-2p|) = \frac12 (1-(1-2p)) = p$, while if $p>1/2$, we have that
	$\frac12 ( 1 - |1-2p|) = \frac12(1-(2p-1)) = 1-p = q$.
	Thus,
	\begin{align}
		\frac12 (1-|1-2p|) = \min\{p,q\},
	\end{align}
	and so
	\begin{align}
		\frac12 \left( 1 - \left\| p \rho^{\otimes n} - q \sigma^{\otimes n} \right\|_1 \right) = \min\{p,q\}
	\end{align}
	in this case. Thus,~\eqref{eq:binary_symmetric2} follows.

\section{Proof of Theorem~\ref{theorem:binary_symmetric_pure}}

\label{app:proof-binary_symmetric_pure}

    First recall the identity in~\eqref{eq:vector-trace-norm-fid-identity}.
 Applying this
to our case of interest, and setting $|\varphi\rangle = \sqrt{p}|\psi\rangle^{\otimes n}$ and $|\zeta\rangle = \sqrt{q}|\phi\rangle^{\otimes n}$, we find that
\begin{align}
1-2\varepsilon & \leq\left\Vert p\psi^{\otimes n}-q\phi^{\otimes n}\right\Vert
_{1}\\
& =\sqrt{\left(  p+q\right)  ^{2}-4pq\left\vert \langle\psi|\phi
\rangle\right\vert ^{2n}}\\
& =\sqrt{1-4pq\left\vert \langle\psi|\phi\rangle\right\vert ^{2n}}\\
& =\sqrt{1-4pq {F}(\psi,\phi)^{n}},
\end{align}
which we can rewrite as
\begin{align}
\left(  1-2\varepsilon\right)  ^{2}  & \leq1-4pq F(\psi,\phi)^{n} \notag \\
\Leftrightarrow\qquad F(\psi,\phi)^{n}  & \leq\frac{1-\left(  1-2\varepsilon
\right)  ^{2}}{4pq}=\frac{4\varepsilon\left(  1-\varepsilon\right)  }{4pq}\\
\Leftrightarrow\qquad n\ln F(\psi,\phi)  & \leq\ln\!\left(  \frac{\varepsilon
\left(  1-\varepsilon\right)  }{pq}\right) \\
\Leftrightarrow\qquad -n\ln F(\psi,\phi)  & \geq -\ln\!\left(  \frac{\varepsilon
\left(  1-\varepsilon\right)  }{pq}\right)\\
\Leftrightarrow\qquad n  & \geq \frac{\ln\!\left(  \frac{ pq }{\varepsilon
\left(  1-\varepsilon\right)}\right)}{-\ln F(\psi,\phi)}.
\end{align}
This finally implies that
\begin{equation}
n=\left\lceil \frac{\ln\!\left(  \frac{pq}{\varepsilon\left(  1-\varepsilon
\right)  }\right)  }{-\ln F(\psi,\phi)}\right\rceil ,
\end{equation}
concluding the proof.

\section{Proof of Theorem~\ref{theorem:binary_symmetric}}

\label{app:binary_symmetric}

	We begin by establishing the first inequality in~\eqref{eq:binary_symmetric3}. That is,
	we first show that
	\begin{align}
	\label{eq:binary_symmetric3-1}
	n^*(p,\rho,q,\sigma, \varepsilon)		
	\geq 
	\frac{\ln(pq/\varepsilon) }{ -\ln F(\rho,\sigma) }.
	\end{align}
	Let us rewrite the constraint in the definition of sample complexity as in~\eqref{eq:def:sample_complexity_symmetric2}:
	\begin{align}
		\varepsilon 
		\geq \frac12 \left( 1 - \left\| p \rho^{\otimes n} - q \sigma^{\otimes n} \right\|_1 \right).
	\end{align}
	Applying the right-most inequality in Lemma~\ref{lemma:Fact}-\ref{fact:FG} with $A=p\rho$ and $B = q\sigma$ leads to $\left\| p \rho^{\otimes n} - q \sigma^{\otimes n} \right\|_1^2 \leq 1 - 4pq F(\rho^{\otimes n}, \sigma^{\otimes n})$, which in turn leads to the following lower bound:
	\begin{align}
		\varepsilon 
		&\geq \frac12 \left( 1 - \left\| p \rho^{\otimes n} - q \sigma^{\otimes n} \right\|_1 \right)
		\\
		&\geq \frac{2pq F(\rho^{\otimes n}, \sigma^{\otimes n}) }{ 1 + \left\| p \rho^{\otimes n} - q \sigma^{\otimes n} \right\|_1  }
		\\
		&\geq pq F(\rho^{\otimes n}, \sigma^{\otimes n})
		\\
		\label{eq:lower_bound_symmetric_binary}
		&= pq \left[ F(\rho,\sigma) \right]^n
	\end{align}
	This arrives at the desired~\eqref{eq:binary_symmetric3-1}. Here we also used the fact that $\left\| p \rho^{\otimes n} - q \sigma^{\otimes n} \right\|_1 \leq \left\| p \rho^{\otimes n} \right\|_1 + \left\| q \sigma^{\otimes n} \right\|_1 = p+q =1$.

	 Next we prove the following inequality stated in~\eqref{eq:binary_symmetric3}:
	\begin{align}
		n^*(p,\rho,q,\sigma, \varepsilon) \geq \frac{pq-\varepsilon(1-\varepsilon)}{pq \left[d_{\mathrm{B}}(\rho,\sigma)\right]^2  }.
  \label{eq:proof-2nd-ineq-dB}
	\end{align}
	Let us rewrite the constraint in the definition of sample complexity  as in~\eqref{eq:def:sample_complexity_symmetric3}:
	\begin{align}
		1-2\varepsilon \leq  \left\| p \rho^{\otimes n} - q \sigma^{\otimes n} \right\|_1.
	\end{align}
	Applying Lemma~\ref{lemma:Fact}-\ref{fact:FG} with $A=p\rho$ and $B = q\sigma$, consider that
	\begin{align}
		\|p\rho-q\sigma\|_1^2 
		&\leq 1 - 4pq F(\rho,\sigma)
		\\
		&= 1 - 4pq + 4pq - 4pq F(\rho,\sigma)
		\\
		&= 1 - 4pq + 4pq (1 - F(\rho,\sigma))
		\\
		&= 1 - 4pq + 4pq \left( 1 + \sqrt{F}(\rho,\sigma) \right) \left( 1 - \sqrt{F}(\rho,\sigma) \right)
		\\
		&\leq 1 -  4pq + 4pq \left[ d_{\mathrm{B}}(\rho,\sigma) \right]^2.
	\end{align} 
	Applying this, we find that
	\begin{align}
		(1-2\varepsilon)^2
		&\leq \left\| p \rho^{\otimes n} - q \sigma^{\otimes n} \right\|_1^2
		\\
		&\leq 1 - 4pq + 4pq \left[ d_{\mathrm{B}}\!\left(\rho^{\otimes n},\sigma^{\otimes n} \right) \right]^2
		\\
		&= 1 - 4pq + 8pq \frac{\left[ d_{\mathrm{B}}\!\left(\rho^{\otimes n},\sigma^{\otimes n} \right) \right]^2}{2}
		\\
		&= 1 - 4pq + 8pq \left( 1- \left( 1 - \frac{\left[ d_{\mathrm{B}}\!\left(\rho^{\otimes n},\sigma^{\otimes n} \right) \right]^2}{2} \right) \right)
		\\		
		&= 1 - 4pq + 8pq \left( 1- \left( 1 - \frac{\left[ d_{\mathrm{B}}\!\left(\rho,\sigma \right) \right]^2}{2} \right)^n \right)
		\\
		&\leq 1 - 4pq + 8pqn \left( \frac{\left[ d_{\mathrm{B}}\!\left(\rho,\sigma \right) \right]^2}{2} \right)
		\\
		&= 1 - 4pq +4pqn \left[ d_{\mathrm{B}}\!\left(\rho,\sigma \right) \right]^2.
	\end{align}
	This implies that
	\begin{align}
		n &\geq \frac{(1-2\varepsilon)^2 - (1-4pq)}{ 4pq \left[ d_{\mathrm{B}}\!\left(\rho,\sigma \right) \right]^2 }
		\\
		&= \frac{4(pq-\varepsilon(1-\varepsilon))}{ 4pq \left[ d_{\mathrm{B}}\!\left(\rho,\sigma \right) \right]^2 },
	\end{align}
	thus establishing~\eqref{eq:proof-2nd-ineq-dB}.
	
	Let us now establish the rightmost inequality in~\eqref{eq:binary_symmetric3}.
 Set
\begin{equation}
n\coloneqq \left\lceil \inf_{s\in\left[  0,1\right]  }\frac{\ln\!\left(  \frac
{p^{s}q^{1-s}}{\varepsilon}\right)  }{-\ln\operatorname{Tr}[\rho^{s}
\sigma^{1-s}]}\right\rceil ,\label{eq:choice-samp-comp-upp-bnd}
\end{equation}
and let $s^{\ast}\in\left[  0,1\right]  $ be the optimal value of $s$ above.
Noting that $-\ln \Tr\!\left[ \rho^s \sigma^{1-s} \right]\geq 0$ due to~\eqref{eq:renyi-q-lte-1},
then
\begin{align}
n  & \geq\frac{\ln\!\left(  \frac{p^{s^{\ast}}q^{1-s^{\ast}}}{\varepsilon
}\right)  }{-\ln\operatorname{Tr}[\rho^{s^{\ast}}\sigma^{1-s^{\ast}}]}\\
\Leftrightarrow\quad-n\ln\operatorname{Tr}[\rho^{s^{\ast}}\sigma^{1-s^{\ast}}]  &
\geq\ln\!\left(  p^{s^{\ast}}q^{1-s^{\ast}}\right)  -\ln\varepsilon\\
\Leftrightarrow\quad-\ln\operatorname{Tr}\!\left[\left(  \rho^{\otimes n}\right)
^{s^{\ast}}\left(  \sigma^{\otimes n}\right)  ^{1-s^{\ast}}\right]-\ln\!\left(
p^{s^{\ast}}q^{1-s^{\ast}}\right)    & \geq-\ln\varepsilon\\
\Leftrightarrow\quad\operatorname{Tr}\!\left[\left(  p\rho^{\otimes n}\right)
^{s^{\ast}}\left(  q\sigma^{\otimes n}\right)  ^{1-s^{\ast}}\right]  &
\leq\varepsilon.
\end{align}
Now applying Lemma~\ref{lemma:Fact}-\ref{fact:Chernoff} with $A = p \rho^{\otimes n}$ and $B = q\sigma^{\otimes n}$, we conclude that
\begin{align}
\varepsilon & \geq\operatorname{Tr}\!\left[\left(  p\rho^{\otimes n}\right)
^{s^{\ast}}\left(  q\sigma^{\otimes n}\right)  ^{1-s^{\ast}}\right]\\
& \geq\frac{1}{2}\left(  1-\left\Vert p\rho^{\otimes n}-q\sigma^{\otimes
n}\right\Vert _{1}\right)  .
\label{eq:binary_symmetric_achievability_final}
\end{align}
The choice of $n$ in~\eqref{eq:choice-samp-comp-upp-bnd} thus satisfies the
constraint~\eqref{eq:def:sample_complexity_symmetric2} in Definition~\ref{def:binary_symmetric}, and since the optimal sample complexity
cannot exceed this choice, this concludes our proof of the rightmost inequality in~\eqref{eq:binary_symmetric3}.

\section{Fuchs--Caves measurement strategy for achieving the asymptotic sample complexity of symmetric binary quantum hypothesis testing}

\label{app:fuchs-caves}

In this appendix, we show how the upper bound on sample complexity in~\eqref{eq:binary-informal-other2}
can be achieved by repeatedly applying the Fuchs--Caves measurement~\cite{Fuchs1995} along with
classical post-processing. The development here generalizes that in
\cite[Lemma~2]{wilde_et_al:LIPIcs.TQC.2013.157}. For simplicity, here we focus
on the case in which $\sigma$ is a positive definite state. Let us first
recall the operator geometric mean of two positive definite operators $A$
and $B$:
\begin{equation}
A\#B\coloneqq B^{1/2}\left(  B^{-1/2}AB^{-1/2}\right)  ^{1/2}B^{1/2},
\end{equation}
which reduces to the usual geometric mean in the case that $A$ and $B$ are
scalars. It is defined as $\lim_{\varepsilon \to 0^+} (A+\varepsilon I)\#(B+\varepsilon I)$ in the more general case that $A$ and $B$ are positive semi-definite.   The Fuchs--Caves observable for distinguishing $\rho^{\otimes n}$
from $\sigma^{\otimes n}$, with respective priors $p$ and $q$, is then
$p\rho^{\otimes n}\#\left(  q\sigma^{\otimes n}\right)  ^{-1}$. Observe that
\begin{equation}
p\rho^{\otimes n}\#\left(  q\sigma^{\otimes n}\right)  ^{-1}=\sqrt{\frac{p}
{q}}\left(  \rho\#\sigma^{-1}\right)  ^{\otimes n}
,\label{eq:tensor-power-geo-mean}
\end{equation}
which implies that a diagonalizing basis for the Fuchs--Caves observable can
be realized as a tensor product of the diagonalizing basis for $\rho
\#\sigma^{-1}$. That is, let us suppose that $\rho\#\sigma^{-1}$ has the
following spectral decomposition:
\begin{equation}
\rho\#\sigma^{-1}=\sum_{y}\lambda_{y}|y\rangle\!\langle
y|.\label{eq:FC-eigen-decomp}
\end{equation}
It was observed in~\cite{Fuchs1995} that the eigenvalues of $\rho\#\sigma^{-1}$ take the
following form:
\begin{equation}
\lambda_{y}=\left(  \frac{\langle y|\rho|y\rangle}{\langle y|\sigma|y\rangle
}\right)  ^{1/2}.\label{eq:FC-eigenvals}
\end{equation}
The so-called Fuchs-Caves measurement for distinguishing between the states $\rho$ and $\sigma$ is a projective measurement defined by the set of rank-one projections $\{ 
|y\rangle\!\langle y|\}$. Performing this measurement leads to the outcome probability distributions $p:=\{p(y)\}$, with $p(y) = \langle y | \rho |y \rangle$, and  $q:= \{q(y)\}$, with $q(y) = \langle y | \sigma |y \rangle$.
It is known that the fidelity $F(\rho, \sigma)$, which is a measure of distinguishability between the two states $\rho$ and $\sigma$, is exactly equal to the classical fidelity $F(p,q)$ of the distributions  $p$ and $q$. In other words,
the quantum fidelity is achieved by the Fuchs--Caves
measurement:
\begin{equation}
F(\rho,\sigma)=\left(  \sum_{y}\sqrt{\langle y|\rho|y\rangle\!\langle
y|\sigma|y\rangle}\right)  ^{2}.
\end{equation}
Then applying~\eqref{eq:tensor-power-geo-mean} and~\eqref{eq:FC-eigen-decomp}, we find that
\begin{equation}
p\rho^{\otimes n}\#\left(  q\sigma^{\otimes n}\right)  ^{-1}=\sqrt{\frac{p}
{q}}\sum_{y^{n}}\lambda_{y^{n}}|y^{n}\rangle\!\langle y^{n}|,
\end{equation}
where
\begin{align}
\lambda_{y^{n}}  & \coloneqq \lambda_{y_{1}}\times\lambda_{y_{2}}\times\cdots
\times\lambda_{y_{n}},\\
|y^{n}\rangle & \coloneqq |y_{1}\rangle\otimes|y_{2}\rangle\otimes\cdots\otimes
|y_{n}\rangle.
\end{align}
Furthermore, observe that
\begin{equation}
\left[  F(\rho,\sigma)\right]  ^{n}=F(\rho^{\otimes n},\sigma^{\otimes
n})=\left(  \sum_{y^{n}}\sqrt{\langle y^{n}|\rho^{\otimes n}|y^{n}
\rangle\ \langle y^{n}|\sigma^{\otimes n}|y^{n}\rangle}\right)  ^{2}.
\end{equation}

The idea of the Fuchs--Caves measurement in this case is to perform the
measurement $\left\{  |y\rangle\!\langle y|\right\}  _{y}$ on each individual
system, leading to a measurement outcome sequence $y^{n}$, and then decide
$\rho$ if $\lambda_{y^{n}}\sqrt{\frac{p}{q}}\geq1$ and $\sigma$ if
$\lambda_{y^{n}}\sqrt{\frac{p}{q}}<1$. In this sense, the test being performed
is a \textquotedblleft quantum likelihood ratio test,\textquotedblright\ as
discussed in~\cite{Fuchs1995}. So indeed the strategy consists of a product measurement
followed by classical post-processing. The error probability of this
measurement strategy is then given by:
\begin{equation}
p\operatorname{Tr}[\Lambda_{\sigma}^{(n)}\rho^{\otimes n}]+q\operatorname{Tr}
[\Lambda_{\rho}^{(n)}\sigma^{\otimes n}],
\end{equation}
where
\begin{align}
\Lambda_{\rho}^{(n)}  & \coloneqq \sum_{y^{n}:\lambda_{y^{n}}\sqrt{\frac{p}{q}}\geq
1}|y^{n}\rangle\!\langle y^{n}|,\\
\Lambda_{\sigma}^{(n)}  & \coloneqq \sum_{y^{n}:\lambda_{y^{n}}\sqrt{\frac{p}{q}}
<1}|y^{n}\rangle\!\langle y^{n}|.
\end{align}
Now generalizing the error analysis in the proof of~\cite[Lemma~2]
{wilde_et_al:LIPIcs.TQC.2013.157} and observing from~\eqref{eq:tensor-power-geo-mean},~\eqref{eq:FC-eigen-decomp}, and~\eqref{eq:FC-eigenvals} that
\begin{equation}
\lambda_{y^{n}}\sqrt{\frac{p}{q}}\geq1\qquad\Leftrightarrow\qquad\left(
p\prod\limits_{i=1}^{n}\langle y_{i}|\rho|y_{i}\rangle\right)  ^{1/2}
\geq\left(  q\prod\limits_{i=1}^{n}\langle y_{i}|\sigma|y_{i}\rangle\right)
^{1/2},
\end{equation}
we find the following upper bound on the error probability of this
Fuchs--Caves strategy:
\begin{align}
& p\operatorname{Tr}[\Lambda_{\sigma}^{(n)}\rho^{\otimes n}
]+q\operatorname{Tr}[\Lambda_{\rho}^{(n)}\sigma^{\otimes n}]\nonumber\\
& =p\sum_{y^{n}:\lambda_{y^{n}}\sqrt{\frac{p}{q}}<1}\langle y^{n}
|\rho^{\otimes n}|y^{n}\rangle+q\sum_{y^{n}:\lambda_{y^{n}}\sqrt{\frac{p}{q}
}\geq1}\langle y^{n}|\sigma^{\otimes n}|y^{n}\rangle
\label{eq:FC-err-analysis-1}\\
& =p\sum_{y^{n}:\lambda_{y^{n}}\sqrt{\frac{p}{q}}<1}\prod\limits_{i=1}
^{n}\langle y_{i}|\rho|y_{i}\rangle+q\sum_{y^{n}:\lambda_{y^{n}}\sqrt{\frac
{p}{q}}\geq1}\prod\limits_{i=1}^{n}\langle y_{i}|\sigma|y_{i}\rangle\\
& =\sum_{y^{n}:\lambda_{y^{n}}\sqrt{\frac{p}{q}}<1}\left(  p\prod
\limits_{i=1}^{n}\langle y_{i}|\rho|y_{i}\rangle\right)  ^{1/2}\left(
p\prod\limits_{i=1}^{n}\langle y_{i}|\rho|y_{i}\rangle\right)  ^{1/2}
\nonumber\\
& \qquad+\sum_{y^{n}:\lambda_{y^{n}}\sqrt{\frac{p}{q}}\geq1}\left(
q\prod\limits_{i=1}^{n}\langle y_{i}|\sigma|y_{i}\rangle\right)  ^{1/2}\left(
q\prod\limits_{i=1}^{n}\langle y_{i}|\sigma|y_{i}\rangle\right)  ^{1/2}\\
& \leq\sum_{y^{n}:\lambda_{y^{n}}\sqrt{\frac{p}{q}}<1}\left(  p\prod
\limits_{i=1}^{n}\langle y_{i}|\rho|y_{i}\rangle\right)  ^{1/2}\left(
q\prod\limits_{i=1}^{n}\langle y_{i}|\sigma|y_{i}\rangle\right)
^{1/2}\nonumber\\
& \qquad+\sum_{y^{n}:\lambda_{y^{n}}\sqrt{\frac{p}{q}}\geq1}\left(
p\prod\limits_{i=1}^{n}\langle y_{i}|\rho|y_{i}\rangle\right)  ^{1/2}\left(
q\prod\limits_{i=1}^{n}\langle y_{i}|\sigma|y_{i}\rangle\right)  ^{1/2}\\
& =\sum_{y^{n}}\left(  p\prod\limits_{i=1}^{n}\langle y_{i}|\rho|y_{i}
\rangle\right)  ^{1/2}\left(  q\prod\limits_{i=1}^{n}\langle y_{i}
|\sigma|y_{i}\rangle\right)  ^{1/2}\\
& =\sqrt{pq}\sum_{y^{n}}\sqrt{\langle y^{n}|\rho^{\otimes n}|y^{n}
\rangle\ \langle y^{n}|\sigma^{\otimes n}|y^{n}\rangle}\\
& =\sqrt{pq}F(\rho,\sigma)^{n/2}.\label{eq:FC-err-analysis-last}
\end{align}
Turning this bound around for the purposes of sample complexity, by setting
\begin{equation}
n\coloneqq \left\lceil \frac{\ln\!\left(  \frac{\sqrt{pq}}{\varepsilon}\right)
}{-\frac{1}{2}\ln F(\rho,\sigma)}\right\rceil ,
\end{equation}
where $\varepsilon$ is the desired threshold in~\eqref{eq:def:sample_complexity_symmetric1}, we find that
\begin{align}
n  & \geq\frac{\ln\!\left(  \frac{\sqrt{pq}}{\varepsilon}\right)  }{-\frac
{1}{2}\ln F(\rho,\sigma)}\\
\Leftrightarrow\qquad n\ln\!\left(  \frac{1}{\sqrt{F}(\rho,\sigma)}\right)
& \geq\ln\!\left(  \frac{\sqrt{pq}}{\varepsilon}\right)  \\
\Leftrightarrow\qquad\frac{1}{\left[  \sqrt{F}(\rho,\sigma)\right]  ^{n}}  &
\geq\frac{\sqrt{pq}}{\varepsilon}\\
\Leftrightarrow\qquad\varepsilon & \geq\sqrt{pq}\left[  \sqrt{F}(\rho
,\sigma)\right]  ^{n},
\end{align}
which, when combined with~\eqref{eq:FC-err-analysis-1}--\eqref{eq:FC-err-analysis-last}, establishes
that the Fuchs--Caves measurement strategy achieves the sample complexity
upper bound in~\eqref{eq:binary-informal-other2}.

\section{Proof of Theorem~\ref{thm:binary_asymmetric}}

\label{app:proof_samp_comp_binary_asymmetric}

Let us begin by proving the lower bound in~\eqref{eq:binary_asymmetric_samp_comp}. Here we assume the existence of a scheme such that the constraint
$\beta_{\varepsilon}(\rho^{\otimes n}\Vert\sigma^{\otimes n})\leq\delta$
holds. Then we apply 
Lemma~\ref{lemma:strong_converse}
to conclude that the following holds for all $\alpha \in (1,\gamma]$:
\begin{align}
\ln\!\left(  \frac{1}{\delta}\right)   &  \leq-\ln\beta_{\varepsilon}
(\rho^{\otimes n}\Vert\sigma^{\otimes n})\\
&  \leq\widetilde{D}_{\alpha}(\rho^{\otimes n}\Vert\sigma^{\otimes n}
)+\frac{\alpha}{\alpha-1}\ln\!\left(  \frac{1}{1-\varepsilon}\right)  \\
&  =n\widetilde{D}_{\alpha}(\rho\Vert\sigma)+\frac{\alpha}{\alpha-1}
\ln\!\left(  \frac{1}{1-\varepsilon}\right)  .
\end{align}
Rewriting this gives the lower bound:
\begin{equation}
n^{\ast}(\rho,\sigma,\varepsilon,\delta)\geq\sup_{\alpha \in (1,\gamma]}\left(  \frac
{\ln\!\left(  \frac{\left(  1-\varepsilon\right)  ^{\alpha^{\prime}}}{\delta
}\right)  }{\widetilde{D}_{\alpha}(\rho\Vert\sigma)}\right)  .
\end{equation}
Now applying the same reasoning, but instead applying the expression in~\eqref{eq:asymm-beta-rewrite-2}, we conclude the lower bound
\begin{equation}
n^{\ast}(\rho,\sigma,\varepsilon,\delta)\geq\sup_{\alpha \in (1,\gamma]}\left(  \frac
{\ln\!\left(  \frac{\left(  1-\delta\right)  ^{\alpha^{\prime}}}{\varepsilon
}\right)  }{\widetilde{D}_{\alpha}(\sigma\Vert\rho)}\right)  ,
\end{equation}
thus finishing the proof of the first inequality in~\eqref{eq:binary_asymmetric_samp_comp}.

Now we prove the upper bound in~\eqref{eq:binary_asymmetric_samp_comp}.
For every $n\in\mathbb{N}$, all $\alpha \in (0,1)$, and all $\varepsilon\in(0,1]$, we apply 
Lemma~\ref{lemma:Hoeffding}
and additivity of the Petz--R\'enyi relative entropy to obtain
\begin{align}
-\ln\beta_{\varepsilon}(\rho^{\otimes n}\Vert\sigma^{\otimes n})
&\geq
D_{\alpha}(\rho^{\otimes n}\Vert\sigma^{\otimes n})+\alpha^{\prime
}\ln\!\left(  \frac{1}{\varepsilon}\right)
\label{eq:petz-renyi-to-hypo-test-1}
\\
&= n D_{\alpha}(\rho\Vert\sigma)+\alpha^{\prime
}\ln\!\left(  \frac{1}{\varepsilon}\right).
\label{eq:petz-renyi-to-hypo-test-2}
\end{align}
Note here that $D_{\alpha}(\rho\Vert\sigma)$ is strictly positive and finite for every $\alpha \in (0,1)$, provided that $\rho\not\perp \sigma$. Now setting
\begin{equation}
    n \coloneqq \left \lceil   \frac{\ln\!\left(  \frac{\varepsilon^{\alpha^{\prime}}
}{\delta}\right)  }{D_{\alpha}(\rho\Vert\sigma)}  \right\rceil,
\label{eq:choice-n-ach-asym-bin}
\end{equation}
we conclude that
\begin{equation}
    n \geq    \frac{\ln\!\left(  \frac{\varepsilon^{\alpha^{\prime}}
}{\delta}\right)  }{D_{\alpha}(\rho\Vert\sigma)}  ,
\end{equation}
which can be rewritten as
\begin{align}
n D_{\alpha}(\rho\Vert\sigma)+\alpha^{\prime
}\ln\!\left(  \frac{1}{\varepsilon}\right)
\geq \ln\! \left( \frac{1}{\delta} \right).
\label{eq:n-constraint-apply}
\end{align}
By applying~\eqref{eq:petz-renyi-to-hypo-test-1}--\eqref{eq:petz-renyi-to-hypo-test-2} and~\eqref{eq:n-constraint-apply}, we then conclude that
\begin{equation}
\beta_{\varepsilon}(\rho^{\otimes n}\Vert\sigma^{\otimes n})\leq\delta,
\end{equation}
which means that the choice of $n$ in~\eqref{eq:choice-n-ach-asym-bin} leads to the desired constraint in~\eqref{eq:asymm-beta-rewrite} being
satisfied. Since the sample complexity is the minimum value of $n\in
\mathbb{N}$ such that~\eqref{eq:asymm-beta-rewrite} holds, we thus conclude
the following upper bound on sample complexity:
\begin{equation}
n^{\ast}(\rho,\sigma,\varepsilon,\delta)\leq\left\lceil   \frac{\ln\!\left(  \frac{\varepsilon^{\alpha^{\prime}}
}{\delta}\right)  }{D_{\alpha}(\rho\Vert\sigma)}  \right\rceil .
\end{equation}
Since the above bound holds for every $\alpha \in (0,1)$, we conclude
\begin{equation}
n^{\ast}(\rho,\sigma,\varepsilon,\delta)
\leq \inf_{\alpha\in\left(
0,1\right)  } \left\lceil   \frac{\ln\!\left(  \frac{\varepsilon^{\alpha^{\prime}}
}{\delta}\right)  }{D_{\alpha}(\rho\Vert\sigma)}  \right\rceil
=\left\lceil \inf_{\alpha\in\left(
0,1\right)  }\left(  \frac{\ln\!\left(  \frac{\varepsilon^{\alpha^{\prime}}
}{\delta}\right)  }{D_{\alpha}(\rho\Vert\sigma)}\right)  \right\rceil .
\end{equation}
Now applying the same reasoning, but instead considering the expression in~\eqref{eq:asymm-beta-rewrite-2}, we conclude the following upper bound:
\begin{equation}
n^{\ast}(\rho,\sigma,\varepsilon,\delta)\leq\left\lceil \inf_{\alpha\in\left(
0,1\right)  }\left(  \frac{\ln\!\left(  \frac{\delta^{\alpha^{\prime}}
}{\varepsilon}\right)  }{D_{\alpha}(\sigma\Vert\rho)}\right)  \right\rceil ,
\end{equation}
thus finishing the proof of the second inequality in~\eqref{eq:binary_asymmetric_samp_comp}.

\section{Proof of Corollary~\ref{corollary:binary-asymmetric-informal}}

\label{app:proof-binary-asymmetric-informal}

Before proving Corollary~\ref{corollary:binary-asymmetric-informal}, we state Lemma~\ref{lem:rel-ent-to-renyis} below. Several aspects of this lemma were already proven as \cite[Lemma~13]{BDS+24} and extended
to the infinite-dimensional case in \cite[Lemma~4.13]{bergh2023infinite} (see also \cite{TCR2009,tomamichel2013framework} for earlier statements along these lines). Here, we restate
the lemma, since our convention with all relative entropies in our paper
involves the natural logarithm instead of the binary logarithm. Furthermore,
since it is rather cumbersome to check and translate all the details of the
proof of \cite[Lemma~13]{BDS+24} using our convention, we have opted to provide
a complete proof 
below, which follows that of \cite[Lemma~13]{BDS+24} quite
closely. It should be noted also that the proof below holds in the more
general case that $\sigma$ is a positive semi-definite operator, which can be
useful in other applications. 
Generalizing this proof to the infinite-dimensional case follows straightforwardly from the techniques of \cite[Lemma~4.13]{bergh2023infinite}.


\begin{lemma}
\label{lem:rel-ent-to-renyis}
Let $\rho$ be a quantum state and $\sigma$ a positive semi-definite operator.
For $\gamma\in(0,1]$, define
\begin{equation}
c_{\gamma}(\rho\Vert\sigma)\coloneqq\frac{1}{\gamma}\ln\!\left(  e^{\gamma
D_{1+\gamma}(\rho\Vert\sigma)}+e^{-\gamma D_{1-\gamma}(\rho\Vert\sigma
)}+1\right)  .
\end{equation}
The following bounds hold for all $\gamma\in(0,1]$:
\begin{align}
D(\rho\Vert\sigma)  &  \leq c_{\gamma}(\rho\Vert\sigma
),\label{eq:rel-ent-to-c-g-bound-up}\\
D(\rho\Vert\sigma)  &  \leq\gamma\left[  c_{\gamma}(\rho\Vert\sigma)\right]
^{2}. \label{eq:rel-ent-to-c-g-bound-up-squared}
\end{align}
In the case that $D(\rho\Vert\sigma)\geq0$, for all $\gamma\in(0,1]$ and
$\zeta\in(0,\frac{\gamma}{2}]$, the following bound holds:
\begin{equation}
D_{1+\zeta}(\rho\Vert\sigma)\leq D(\rho\Vert\sigma)+\zeta\left[  c_{\gamma
}(\rho\Vert\sigma)\right]  ^{2}. \label{eq:alpha-g1-bound-rel-ent}
\end{equation}
For all $\gamma\in(0,1]$ and $\zeta\in(0,\frac{\ln3}{2c_{\gamma}(\rho
\Vert\sigma)}]$, the following bound holds:
\begin{equation}
D_{1+\zeta}(\rho\Vert\sigma)\leq D(\rho\Vert\sigma)+\zeta K\left[  c_{\gamma
}(\rho\Vert\sigma)\right]  ^{2},
\label{eq:rel-ent-u-bnd-Petz-geq-1-noconstr}
\end{equation}
where
\begin{equation}
K\coloneqq\cosh\!\left(  \frac{\ln3}{2}\right)  \approx 1.155 . \label{eq:def-K-constant}
\end{equation}
If there exists $\gamma^{\prime}\in(0,1]$ such that $D_{1+\gamma^{\prime}
}(\rho\Vert\sigma)<+\infty$, then for all $\gamma\in(0,\gamma^{\prime}]$ and
$\zeta\in(0,\frac{\ln3}{2c_{\gamma}(\rho\Vert\sigma)}]$, the following bound
holds:
\begin{equation}
D(\rho\Vert\sigma)\leq D_{1-\zeta}(\rho\Vert\sigma)+\zeta K\left[  c_{\gamma
}(\rho\Vert\sigma)\right]  ^{2}, \label{eq:rel-ent-to-alpha-below-one}
\end{equation}

\end{lemma}

\begin{proof}
%
First, suppose that $D(\rho\Vert\sigma)=+\infty$. Then it follows that
$c_{\gamma}(\rho\Vert\sigma)=+\infty$, so that~\eqref{eq:rel-ent-to-c-g-bound-up},~\eqref{eq:rel-ent-to-c-g-bound-up-squared}, and~\eqref{eq:alpha-g1-bound-rel-ent} hold trivially in this case.

So let us now suppose that there exists $\gamma \in (0,1]$ such that $D_{1+\gamma}(\rho\Vert\sigma)<+\infty $ (and thus $D(\rho\Vert\sigma)<+\infty$) and prove~\eqref{eq:rel-ent-to-c-g-bound-up} and~\eqref{eq:rel-ent-to-c-g-bound-up-squared}. Observe that
\begin{equation}
D(\rho\Vert\sigma)\leq D_{1+\gamma}(\rho\Vert\sigma)=\frac{1}{\gamma}
\ln\!\left(  e^{\gamma D_{1+\gamma}(\rho\Vert\sigma)}\right)  \leq c_{\gamma
}(\rho\Vert\sigma), \label{eq:D-to-g-cg}
\end{equation}
thus establishing~\eqref{eq:rel-ent-to-c-g-bound-up}. Also, under this
assumption, we necessarily have that $\operatorname{supp}(\rho)\subseteq
\operatorname{supp}(\sigma)$. So then we can restrict the Hilbert space to the
support of $\sigma$, where $\sigma$ is an invertible (positive
definite)\ operator. Now define
\begin{align}
X  &  \coloneqq \rho\otimes\left(  \sigma^{-1}\right)  ^{T},\\
|\varphi\rangle &  \coloneqq \left(  \sqrt{\rho}\otimes I\right)  \sum
_{i}|i\rangle|i\rangle,
\end{align}
so that $|\varphi\rangle$ is a normalized state vector, and observe that
\begin{equation}
D(\rho\Vert\sigma)=\langle\varphi|\ln X|\varphi\rangle,
\label{eq:rel-ent-rewrite}
\end{equation}
and for all $\zeta\in\lbrack-1,0)\cup(0,1]$,
\begin{equation}
D_{1+\zeta}(\rho\Vert\sigma)=\frac{1}{\zeta}\ln\langle\varphi|X^{\zeta
}|\varphi\rangle. \label{eq:Petz-renyi-rewrite}
\end{equation}
Additionally, observe that for all $\gamma\in(0,1]$,
\begin{equation}
\label{eq:c-gamma-X-powers}
c_{\gamma}(\rho\Vert\sigma)=\frac{1}{\gamma}\ln\langle\varphi|\left(
X^{\gamma}+X^{-\gamma}+I\right)  |\varphi\rangle,
\end{equation}
Since the function $x^{\gamma}+x^{-\gamma}+1$ has a minimum at $x=1$, it
follows that $x^{\gamma}+x^{-\gamma}+1\geq3$ for all $x>0$. So
the
following inequality holds for all $\gamma\in(0,1]$:
\begin{equation}
\gamma c_{\gamma}(\rho\Vert\sigma)\geq \ln 3 \geq 1. \label{eq:g-cg-to-1}
\end{equation}
Thus, from~\eqref{eq:D-to-g-cg} and~\eqref{eq:g-cg-to-1}, we conclude that
$\gamma D(\rho\Vert\sigma)\leq\gamma c_{\gamma}(\rho\Vert\sigma)\leq\left[
\gamma c_{\gamma}(\rho\Vert\sigma)\right]  ^{2}$, which implies~\eqref{eq:rel-ent-to-c-g-bound-up-squared}.

We now prove~\eqref{eq:alpha-g1-bound-rel-ent}. For all $t>0$ and $\zeta
\in\mathbb{R}$, the first term in the Taylor expansion of $t^{\zeta}$ around
$\zeta=0$ is $\zeta\ln t$. Then we can write $t^{\zeta}=1+\zeta\ln
t+r_{\zeta}(t)$, where
\begin{equation}
r_{\zeta}(t)\coloneqq t^{\zeta}-\zeta\ln t-1. \label{eq:def-r-function}
\end{equation}
Given that $1+x\leq e^{x}$ for all $x\in\mathbb{R}$, we find that
\begin{align}
r_{\zeta}(t)  &  =t^{\zeta}-\zeta\ln t-1\label{eq:r-less-than-s-1}\\
&  =t^{\zeta}+1-\zeta\ln t-2\\
&  \leq t^{\zeta}+e^{-\zeta\ln t}-2\\
&  =e^{\zeta\ln t}+e^{-\zeta\ln t}-2\\
&  =2\left[  \cosh(\zeta\ln t)-1\right]  \eqqcolon s_{\zeta}(t).
\label{eq:r-less-than-s-last}
\end{align}
Observe that
\begin{align}
s_{-\zeta}(t)  &  =s_{\zeta}(t),\label{eq:s-even-func}\\
s_{\zeta}(t)  &  =s_{\gamma\zeta}(t^{1/\gamma})\quad\forall\gamma\in
\mathbb{R}.
\end{align}
One can also verify that $s_{\zeta}(t)$ is monotonically increasing in $t$ for
$t\geq1$ and concave in $t$ if $\zeta\leq1/2$ and $t\geq3$. For all $t\geq0$,
either $t$ or $1/t$ is $\geq1$, so that we can use the monotonicity of
$s_{\zeta}(t)$ to conclude that
\begin{equation}
s_{\zeta}(t)=s_{\zeta}\!\left(  \frac{1}{t}\right)  \leq s_{\zeta}\!\left(
t+\frac{1}{t}\right)  .
\end{equation}
Using this, we conclude for all $t\geq0$ and $\gamma\in(0,1]$ that
\begin{align}
s_{\zeta}(t)  &  =s_{\zeta/\gamma}(t^{\gamma})\\
&  \leq s_{\zeta/\gamma}(t^{\gamma}+t^{-\gamma})\\
&  \leq s_{\zeta/\gamma}(t^{\gamma}+t^{-\gamma}+1).
\end{align}
As mentioned above, the function $t^{\gamma}+t^{-\gamma}+1$ has a minimum at
$t=1$, and so it follows that $t^{\gamma}+t^{-\gamma}+1\geq3$ for all $t>0$. Then we
can use the inequality above, i.e., $s_{\zeta}(t)\leq s_{\zeta/\gamma}(t^{\gamma}+t^{-\gamma}+1)$, and concavity of $s_{\zeta/\gamma}$ for $\zeta\leq\frac{\gamma}{2}$ to
conclude that
\begin{align}
\langle\varphi|s_{\zeta}(X)|\varphi\rangle &  \leq\langle\varphi
|s_{\zeta/\gamma}(X^{\gamma}+X^{-\gamma}+I)|\varphi\rangle\\
&  \leq s_{\zeta/\gamma}(\langle\varphi|\left(  X^{\gamma}+X^{-\gamma
}+I\right)  |\varphi\rangle)\label{eq:concavity-applied-rel-ent-pf}\\
&  =s_{\zeta/\gamma}(e^{\gamma c_{\gamma}(\rho\Vert\sigma)}).
\end{align}
The second inequality follows from applying Jensen's inequality. The equality follows from~\eqref{eq:c-gamma-X-powers}. We finally
use Taylor's theorem in the Lagrange form of the remainder to conclude that
the following holds for all $t\geq0$ and some $\eta\in(0,\zeta)$:
\begin{align}
s_{\zeta}(t)  &  =s_{0}(t)+\left.  \frac{d}{d\zeta}s_{\zeta}(t)\right\vert
_{\zeta=0}\zeta+\frac{1}{2}\left.  \frac{d^{2}}{d\zeta^{2}}s_{\zeta
}(t)\right\vert _{\zeta=\eta}\zeta^{2}\\
&  =\zeta^{2}\left(  \ln t\right)  ^{2}\cosh(\eta\ln t)\\
&  \leq\zeta^{2}\left(  \ln t\right)  ^{2}\cosh(\zeta\ln t).
\label{eq:cosh-mono-inc}
\end{align}
For the second equality, we used the facts that
\begin{align}
s_{0}(t)  &  =\left.  \frac{d}{d\zeta}s_{\zeta}(t)\right\vert _{\zeta=0}=0,\\
\left.  \frac{d^{2}}{d\zeta^{2}}s_{\zeta}(t)\right\vert _{\zeta=\eta}  &
=2\left(  \ln t\right)  ^{2}\cosh(\eta\ln t).
\end{align}
For the inequality in~\eqref{eq:cosh-mono-inc}, we used the fact that
$\eta \mapsto \cosh(\eta\ln t)$ is monotone increasing on $\eta\in(0,\zeta]$. Then we
conclude that
\begin{align}
\langle\varphi|s_{\zeta}(X)|\varphi\rangle &  \leq s_{\zeta/\gamma}(e^{\gamma
c_{\gamma}(\rho\Vert\sigma)})\label{eq:s-quant-final-bnd-1}\\
&  \leq\left(  \frac{\zeta}{\gamma}\right)  ^{2}\left(  \ln e^{\gamma
c_{\gamma}(\rho\Vert\sigma)}\right)  ^{2}\cosh\!\left(  \frac{\zeta}{\gamma
}\ln e^{\gamma c_{\gamma}(\rho\Vert\sigma)}\right) \\
&  =\zeta^{2}\left[  c_{\gamma}(\rho\Vert\sigma)\right]  ^{2}\cosh\!\left(
\zeta c_{\gamma}(\rho\Vert\sigma)\right)  . \label{eq:s-quant-final-bnd-last}
\end{align}
Then for $\zeta\in(0,\frac{\gamma}{2}]$ and when $D(\rho\Vert\sigma)\geq0$,
consider that
\begin{align}
D_{1+\zeta}(\rho\Vert\sigma)  &  =\frac{1}{\zeta}\ln\langle\varphi|X^{\zeta
}|\varphi\rangle\\
&  =\frac{1}{\zeta}\ln\!\left[  1+\zeta\langle\varphi|\ln X|\varphi
\rangle+\langle\varphi|r_{\zeta}(X)|\varphi\rangle\right] \\
&  =\frac{1}{\zeta}\ln\!\left[  1+\zeta D(\rho\Vert\sigma)+\langle
\varphi|r_{\zeta}(X)|\varphi\rangle\right] \\
&  \leq\frac{1}{\zeta}\ln\!\left[  1+\zeta D(\rho\Vert\sigma)+\langle
\varphi|s_{\zeta}(X)|\varphi\rangle\right] \\
&  =\frac{1}{\zeta}\ln\!\left[  1+\zeta D(\rho\Vert\sigma)\right]  +\frac
{1}{\zeta}\ln\!\left[  1+\frac{\langle\varphi|s_{\zeta}(X)|\varphi\rangle
}{1+\zeta D(\rho\Vert\sigma)}\right] \\
&  \leq D(\rho\Vert\sigma)+\frac{1}{\zeta}\ln\!\left[  1+\frac{\langle
\varphi|s_{\zeta}(X)|\varphi\rangle}{1+\zeta D(\rho\Vert\sigma)}\right] \\
&  \leq D(\rho\Vert\sigma)+\frac{1}{\zeta}\ln\!\left[  1+\langle\varphi
|s_{\zeta}(X)|\varphi\rangle\right] \\
&  \leq D(\rho\Vert\sigma)+\frac{1}{\zeta}\ln\!\left[  1+\zeta^{2}\left[
c_{\gamma}(\rho\Vert\sigma)\right]  ^{2}\cosh\!\left(  \zeta c_{\gamma}
(\rho\Vert\sigma)\right)  \right] \\
&  \leq D(\rho\Vert\sigma)+\zeta\left[  c_{\gamma}(\rho\Vert\sigma)\right]
^{2}.
\end{align}
The first equality follows from~\eqref{eq:Petz-renyi-rewrite}, the second
equality from~\eqref{eq:def-r-function}, and the third equality from~\eqref{eq:rel-ent-rewrite}. The first inequality follows from~\eqref{eq:r-less-than-s-1}--\eqref{eq:r-less-than-s-last}. The second
inequality follows because $\ln(1+x)\leq x$ for $x\geq 0$. The third
inequality follows because $\zeta D(\rho\Vert\sigma)\geq0$. The penultimate
inequality follows from~\eqref{eq:s-quant-final-bnd-1}--\eqref{eq:s-quant-final-bnd-last}. The final
inequality follows because the function $k\rightarrow k^{2}-\ln(1+k^{2}
\cosh(k))$ is monotonically increasing for $k\geq0$, so that it is non-negative.

We now prove \eqref{eq:rel-ent-u-bnd-Petz-geq-1-noconstr}. Suppose that $\zeta\in(0,\frac{\ln3}{2c_{\gamma}(\rho\Vert\sigma)}]$,
while making no constraint on the sign of $D(\rho\Vert\sigma)$. Observe from~\eqref{eq:g-cg-to-1} that
\begin{equation}
\frac{\ln3}{2c_{\gamma}(\rho\Vert\sigma)}\leq\frac{\gamma c_{\gamma}(\rho
\Vert\sigma)}{2c_{\gamma}(\rho\Vert\sigma)}=\frac{\gamma}{2},
\end{equation}
ensuring that the needed concavity of $s_{\zeta/\gamma}(X)$ in~\eqref{eq:concavity-applied-rel-ent-pf} holds. Starting from the first few
steps above, consider that
\begin{align}
D_{1+\zeta}(\rho\Vert\sigma)  &  =\frac{1}{\zeta}\ln\!\left[  1+\zeta
D(\rho\Vert\sigma)+\langle\varphi|r_{\zeta}(X)|\varphi\rangle\right] \\
&  \leq\frac{1}{\zeta}\ln\!\left[  1+\zeta D(\rho\Vert\sigma)+\langle
\varphi|s_{\zeta}(X)|\varphi\rangle\right] \\
&  \leq D(\rho\Vert\sigma)+\frac{1}{\zeta}\langle\varphi|s_{\zeta}
(X)|\varphi\rangle\\
&  \leq D(\rho\Vert\sigma)+\zeta\left[  c_{\gamma}(\rho\Vert\sigma)\right]
^{2}\cosh\!\left(  \zeta c_{\gamma}(\rho\Vert\sigma)\right) \\
&  \leq D(\rho\Vert\sigma)+\zeta\left[  c_{\gamma}(\rho\Vert\sigma)\right]
^{2}\cosh\!\left(  \frac{\ln3}{2}\right) \\
&  =D(\rho\Vert\sigma)+\zeta K\left[  c_{\gamma}(\rho\Vert\sigma)\right]
^{2}.
\end{align}
The first inequality follows from~\eqref{eq:r-less-than-s-1}--\eqref{eq:r-less-than-s-last}.
The second
inequality follows because $\ln(1+x)\leq x$ for $x\geq-1$, the latter
condition ensured because $\zeta D(\rho\Vert\sigma)+\langle\varphi|r_{\zeta
}(X)|\varphi\rangle\geq-1$. The third inequality follows from
~\eqref{eq:s-quant-final-bnd-1}--\eqref{eq:s-quant-final-bnd-last}. The
penultimate inequality follows from the assumption that $\zeta\in(0,\frac
{\ln3}{2c_{\gamma}(\rho\Vert\sigma)}]$ and the fact that $\cosh$ is monotone
increasing on this interval. The final equality follows from the definition of
$K$ in~\eqref{eq:def-K-constant}.

Finally, we prove \eqref{eq:rel-ent-to-alpha-below-one}. Suppose that $\zeta\in(0,\frac{\ln3}{2c_{\gamma}(\rho\Vert\sigma)}]$.
Similar to the above, consider that
\begin{align}
D_{1-\zeta}(\rho\Vert\sigma)  &  =-\frac{1}{\zeta}\ln\!\left[  1-\zeta
D(\rho\Vert\sigma)+\langle\varphi|r_{-\zeta}(X)|\varphi\rangle\right] \\
&  \geq-\frac{1}{\zeta}\ln\!\left[  1-\zeta D(\rho\Vert\sigma)+\langle
\varphi|s_{-\zeta}(X)|\varphi\rangle\right] \\
&  =-\frac{1}{\zeta}\ln\!\left[  1-\zeta D(\rho\Vert\sigma)+\langle
\varphi|s_{\zeta}(X)|\varphi\rangle\right] \\
&  \geq D(\rho\Vert\sigma)-\frac{1}{\zeta}\langle\varphi|s_{\zeta}
(X)|\varphi\rangle\\
&  \geq D(\rho\Vert\sigma)-\zeta\left[  c_{\gamma}(\rho\Vert\sigma)\right]
^{2}\cosh\!\left(  \zeta c_{\gamma}(\rho\Vert\sigma)\right) \\
&  \geq D(\rho\Vert\sigma)-\zeta\left[  c_{\gamma}(\rho\Vert\sigma)\right]
^{2}\cosh\!\left(  \frac{\ln3}{2}\right) \\
&  =D(\rho\Vert\sigma)-\zeta K\left[  c_{\gamma}(\rho\Vert\sigma)\right]
^{2}.
\end{align}
The first inequality follows from~\eqref{eq:r-less-than-s-1}--\eqref{eq:r-less-than-s-last}.
The second
equality follows from~\eqref{eq:s-even-func}. The second inequality follows
because $\ln(1+x)\leq x$ for $x\geq-1$. The final steps follow from similar
reasons as above.
\end{proof}

\bigskip

With Lemma~\ref{lem:rel-ent-to-renyis} in hand, we now move on to the proof of Corollary~\ref{corollary:binary-asymmetric-informal}.

\begin{proof}[Proof of Corollary~\ref{corollary:binary-asymmetric-informal}]
Let $\gamma\in(0,1]$ be such that $D_{1+\gamma}(\rho\Vert\sigma)<+\infty$.
Define
\begin{equation}
c_{\gamma}(\rho\Vert\sigma)\coloneqq\frac{1}{\gamma}\ln\!\left(  e^{\gamma
D_{1+\gamma}(\rho\Vert\sigma)}+e^{-\gamma D_{1-\gamma}(\rho\Vert\sigma
)}+1\right)  .
\end{equation}
For all $\xi^{\prime}$ satisfying $0<\xi^{\prime}<\gamma/2$, the following
inequality holds as a consequence of~\eqref{eq:alpha-g1-bound-rel-ent}:
\begin{equation}
D_{1+\xi^{\prime}}(\rho\Vert\sigma)\leq D(\rho\Vert\sigma)+\xi^{\prime}\left[
c_{\gamma}(\rho\Vert\sigma)\right]  ^{2}.
\end{equation}
Setting $\xi\coloneqq\xi^{\prime}\left[  c_{\gamma}(\rho\Vert\sigma)\right]
^{2}$, we can rewrite this as follows: for all $\xi$ satisfying
\begin{equation}
0<\xi<\frac{\gamma\left[  c_{\gamma}(\rho\Vert\sigma)\right]  ^{2}}{2},
\label{eq:xi-constraint-1}
\end{equation}
the following inequality holds:
\begin{equation}
D_{1+\xi^{\prime}}(\rho\Vert\sigma)\leq D(\rho\Vert\sigma)+\xi.
\end{equation}
Now applying Theorem~\ref{thm:binary_asymmetric}, while noting that the sandwiched R\'enyi relative
entropies therein can be replaced with the Petz--R\'enyi relative entropies (due
to $\widetilde{D}_{\alpha}\leq D_{\alpha}$ holding for all $\alpha > 0$),\ we conclude that
\begin{align}
n^{\ast}(\rho,\sigma,\varepsilon,\delta)  &  \geq\frac{\ln\!\left(
\frac{\left(  1-\varepsilon\right)  ^{\frac{1+\xi^{\prime}}{\xi^{\prime}}}
}{\delta}\right)  }{D_{1+\xi^{\prime}}(\rho\Vert\sigma)}\\
&  =\frac{\ln\!\left(  \frac{1}{\delta}\right)  -\left(  \frac{1}{\xi^{\prime
}}+1\right)  \ln\!\left(  \frac{1}{1-\varepsilon}\right)  }{D_{1+\xi^{\prime}
}(\rho\Vert\sigma)}\\
&  \geq\frac{\ln\!\left(  \frac{1}{\delta}\right)  -\left(  \frac{1}
{\xi^{\prime}}+1\right)  \ln\!\left(  \frac{1}{1-\varepsilon}\right)  }
{D(\rho\Vert\sigma)+\xi}\\
&  =\frac{\ln\!\left(  \frac{1}{\delta}\right)  -\left(  \frac{\left[
c_{\gamma}(\rho\Vert\sigma)\right]  ^{2}}{\xi}+1\right)  \ln\!\left(  \frac
{1}{1-\varepsilon}\right)  }{D(\rho\Vert\sigma)+\xi}.
\end{align}
Let us now pick $\xi=\frac{1}{4}D(\rho\Vert\sigma)$. Then the inequality in~\eqref{eq:xi-constraint-1} follows as a consequence of~\eqref{eq:rel-ent-to-c-g-bound-up-squared}. Substituting above, we find that
\begin{equation}
n^{\ast}(\rho,\sigma,\varepsilon,\delta)\geq\frac{\ln\!\left(  \frac{1}
{\delta}\right)  -\left(  \frac{4\left[  c_{\gamma}(\rho\Vert\sigma)\right]
^{2}}{D(\rho\Vert\sigma)}+1\right)  \ln\!\left(  \frac{1}{1-\varepsilon
}\right)  }{\frac{5}{4}D(\rho\Vert\sigma)}.
\end{equation}
Thus, if
\begin{align}
\ln\!\left(  \frac{1}{\delta}\right)   &  \geq2\left(  \frac{4\left[
c_{\gamma}(\rho\Vert\sigma)\right]  ^{2}}{D(\rho\Vert\sigma)}+1\right)
\ln\!\left(  \frac{1}{1-\varepsilon}\right) \\
\Longleftrightarrow\qquad\delta &  \leq\exp\!\left[  -2\left(  \frac{4\left[
c_{\gamma}(\rho\Vert\sigma)\right]  ^{2}}{D(\rho\Vert\sigma)}+1\right)
\ln\!\left(  \frac{1}{1-\varepsilon}\right)  \right]  ,
\label{eq:small-delta-1}
\end{align}
then
\begin{equation}
n^{\ast}(\rho,\sigma,\varepsilon,\delta)\geq\frac{2}{5}\left(  \frac
{\ln\!\left(  \frac{1}{\delta}\right)  }{D(\rho\Vert\sigma)}\right)  ,
\end{equation}
proving the lower bound in \eqref{eq:binary-asymmetric-informal1}.

For all $\xi^{\prime}$ satisfying $0<\xi^{\prime}<\frac{\ln3}{2c_{\gamma}
    (\rho\Vert\sigma)}$, the following inequality holds, as a consequence of~\eqref{eq:rel-ent-to-alpha-below-one}:
\begin{equation}
D(\rho\Vert\sigma)\leq D_{1-\xi^{\prime}}(\rho\Vert\sigma)+\xi^{\prime
}K\left[  c_{\gamma}(\rho\Vert\sigma)\right]  ^{2},
\end{equation}
where
\begin{equation}
K\coloneqq\cosh\!\left(  \frac{\ln3}{2}\right)  .
\end{equation}
By setting $\xi\coloneqq\xi^{\prime}K\left[  c_{\gamma}(\rho\Vert
\sigma)\right]  ^{2}$, we can rewrite this as follows: For all $\xi$
satisfying
\begin{equation}
0<\xi<\frac{K\ln3}{2}c_{\gamma}(\rho\Vert\sigma),\label{eq:xi-condition-K}
\end{equation}
the following inequality holds:
\begin{equation}
D(\rho\Vert\sigma)\leq D_{1-\xi^{\prime}}(\rho\Vert\sigma)+\xi.
\end{equation}
Proceeding similarly as before, and applying Theorem~\ref{thm:binary_asymmetric}, consider that
\begin{align}
n^{\ast}(\rho,\sigma,\varepsilon,\delta) &  \leq\left\lceil \frac{\ln\!\left(
\frac{\varepsilon^{-\left(  1-\xi^{\prime}\right)  /\xi^{\prime}}}{\delta
}\right)  }{D_{1-\xi^{\prime}}(\rho\Vert\sigma)}\right\rceil \\
&  =\left\lceil \frac{\ln\!\left(  \frac{1}{\delta}\right)  +\left(  \frac
{1}{\xi^{\prime}}-1\right)  \ln\!\left(  \frac{1}{\varepsilon}\right)
}{D_{1-\xi^{\prime}}(\rho\Vert\sigma)}\right\rceil \\
&  \leq\left\lceil \frac{\ln\!\left(  \frac{1}{\delta}\right)  +\left(
\frac{1}{\xi^{\prime}}-1\right)  \ln\!\left(  \frac{1}{\varepsilon}\right)
}{D(\rho\Vert\sigma)-\xi}\right\rceil \\
&  =\left\lceil \frac{\ln\!\left(  \frac{1}{\delta}\right)  +\left(
\frac{K\left[  c_{\gamma}(\rho\Vert\sigma)\right]  ^{2}}{\xi}-1\right)
\ln\!\left(  \frac{1}{\varepsilon}\right)  }{D(\rho\Vert\sigma)-\xi
}\right\rceil .
\end{align}
Observing that $\frac{K\ln3}{2}\approx0.634$, now pick
\begin{equation}
\xi=\frac{1}{2}D(\rho\Vert\sigma),
\end{equation}
so that the inequality in~\eqref{eq:xi-condition-K} is satisfied, as a
consequence of~\eqref{eq:rel-ent-to-c-g-bound-up}. Then
\begin{equation}
n^{\ast}(\rho,\sigma,\varepsilon,\delta)\leq\left\lceil \frac{\ln\!\left(
\frac{1}{\delta}\right)  +\left(  \frac{2K\left[  c_{\gamma}(\rho\Vert
\sigma)\right]  ^{2}}{D(\rho\Vert\sigma)}-1\right)  \ln\!\left(  \frac
{1}{\varepsilon}\right)  }{\frac{1}{2}D(\rho\Vert\sigma)}\right\rceil .
\end{equation}
Thus, if
\begin{align}
\ln\!\left(  \frac{1}{\delta}\right)   &  \geq\left(  \frac{2K\left[
c_{\gamma}(\rho\Vert\sigma)\right]  ^{2}}{D(\rho\Vert\sigma)}-1\right)
\ln\!\left(  \frac{1}{\varepsilon}\right)  \\
\Longleftrightarrow\qquad\delta &  \leq\exp\!\left[  -\left(  \frac{2K\left[
c_{\gamma}(\rho\Vert\sigma)\right]  ^{2}}{D(\rho\Vert\sigma)}-1\right)
\ln\!\left(  \frac{1}{\varepsilon}\right)  \right]  ,\label{eq:small-delta-2}
\end{align}
then
\begin{equation}
n^{\ast}(\rho,\sigma,\varepsilon,\delta)\leq\left\lceil 4\left(  \frac
{\ln\!\left(  \frac{1}{\delta}\right)  }{D(\rho\Vert\sigma)}\right)
\right\rceil ,
\end{equation}
proving the upper bound in \eqref{eq:binary-asymmetric-informal1}.
\end{proof}

\bigskip

 \section{Proof of Theorem~\ref{theorem:sc_M-ary}}

\label{app:proof-sc_M-ary}
 
	The achievability part (i.e., the upper bound in~\eqref{eq:sc_M-ary}) can be deduced from~\cite[Theorem 7.2]{AM14}.
	Here, we provide an alternative simple proof for that. In addition, we provide a bound with improved constants, by using~\cite[Eq.~(8)]{AM14}.
	
	We use the following pretty-good measurement~\cite{Bel75, HW94} with respect to the ensemble $\mathcal{S}$; i.e., for each $m\in\{1,\ldots, M\}$,
	\begin{align}
		\Lambda_m &\coloneqq  \left( \sum_{\bar{m}=1}^M A_{\bar{m}} \right)^{-\frac12} A_m \left( \sum_{\bar{m}=1}^M A_{\bar{m}} \right)^{-\frac12},
		\\
		A_m &\coloneqq  p_m \rho_m.
	\end{align}
	Then,
	\begin{align}
		p_e(\mathcal{S})  &\leq \sum_{m=1}^M \Tr\!\left[ A_m \left( I - \Lambda_m \right) \right] 
		\\
		&= \sum_{m=1}^M \Tr\!\left[ A_m \left( A_m +\sum_{\bar{m}\neq m} A_{\bar{m}} \right)^{-\frac12} \left( \sum_{\bar{m}\neq m} A_{\bar{m}} \right) \left( A_m  + \sum_{\bar{m}\neq m} A_{\bar{m}} \right)^{-\frac12} \right]
		\\
		&\overset{\textnormal{(a)}}{\leq}
		\sum_{m=1}^M \left\| \sqrt{A_m} \sqrt{ \sum_{\bar{m}\neq m} A_{\bar{m}} } \right\|_1
		\\
		&\overset{\textnormal{(b)}}{\leq} \sum_{m=1}^M \sum_{\bar{m}\neq m} \left\| \sqrt{A_m} \sqrt{A_{\bar{m}}} \right\|_1
		\\
		&= \sum_{m=1}^M \sum_{\bar{m}\neq m} \sqrt{p_m}\sqrt{p_{\bar{m}}} \sqrt{F}\left( \rho_m, \rho_{\bar{m}} \right)
		\\
        \label{eq:M-ary-last-upper-bound}
		&\leq M(M-1) \max_{\bar{m}\neq m}  \sqrt{p_m}\sqrt{p_{\bar{m}}} \sqrt{F}\left( \rho_m, \rho_{\bar{m}}  \right),
	\end{align}
	where (a) follows from Lemma~\ref{lemma:Fact}-\ref{fact:PGM_MIN}, the first inequality of Lemma~\ref{lemma:Fact}-\ref{fact:FG}, and the first inequality of Lemma~\ref{lemma:Fact}-\ref{fact:relation_fidelity};
	(b) follows from the subadditivity property stated in Fact~\ref{fact:subadditivity}. 
 
	Using the multiplicativity of the fidelity~$F$, i.e., Lemma~\ref{lemma:Fact}-\ref{fact:additivity}, and following the same reasoning used in~\eqref{eq:choice-samp-comp-upp-bnd}--\eqref{eq:binary_symmetric_achievability_final}, 
	we obtain the following upper bound on the sample complexity:
	\begin{equation} \label{eq:with_no_1_2}
		n^*(\mathcal{S},\varepsilon) \leq 
  \left\lceil \max_{m\neq \bar{m}}  \frac{2\ln\!\left( \frac{ M(M-1) \sqrt{p_m} \sqrt{p_{\bar{m}}}  }{ \varepsilon } \right) }{ - \ln F\!\left(\rho_{m},\rho_{\bar{m}} \right) } \right\rceil . 
	\end{equation}

Next, we provide an improved bound compared to~\eqref{eq:with_no_1_2} by adapting~\cite[Eq.~(8)]{AM14}. Using this inequality, we find that
\begin{align}
    p_e(\mathcal{S})  & \leq \frac{1}{2} \sum_{m=1}^M \sum_{\bar{m}\neq m} \sqrt{p_m}\sqrt{p_{\bar{m}}} \sqrt{F}\left( \rho_m, \rho_{\bar{m}}\right) \\ 
    &\leq \frac{1}{2} M(M-1) \max_{\bar{m}\neq m}  \sqrt{p_m}\sqrt{p_{\bar{m}}} \sqrt{F}\left( \rho_m, \rho_{\bar{m}}  \right).
\end{align}
Similarly, by using the multiplicativity of the fidelity~$F$, 
and following the same reasoning used in~\eqref{eq:choice-samp-comp-upp-bnd}--\eqref{eq:binary_symmetric_achievability_final}, 
	we obtain the following upper bound on the sample complexity:
	\begin{align}
		n^*(\mathcal{S},\varepsilon) \leq 
  \left\lceil \max_{m\neq \bar{m}}  \frac{2\ln\!\left( \frac{ M(M-1) \sqrt{p_m} \sqrt{p_{\bar{m}}}  }{ 2\varepsilon } \right) }{ - \ln F\!\left(\rho_{m},\rho_{\bar{m}} \right) } \right\rceil ,
	\end{align}
 thus establishing the second inequality in~\eqref{eq:sc_M-ary}.

	The converse (i.e., the lower bound in~\eqref{eq:sc_M-ary}) follows from classical reasoning; namely, any instance of pairwise simple hypothesis testing is easier than the $M$-ary hypothesis testing.
	Let $\{ \Lambda_1^\star, \ldots, \Lambda_M^\star \}$ be an optimum measurement achieving the minimum error probability, i.e.
	\begin{align} \label{eq:M-ary_converse1}
		p_e(\mathcal{S})  &=  \sum_{m=1}^M \Tr\!\left[ A_m \left( I - \Lambda_m^\star \right) \right].
	\end{align}
	Then, for every $m\neq \bar{m}$, 
	\begin{align}
	p_e(\mathcal{S})
	&\geq 
	\Tr\!\left[ A_m \left( I - \Lambda_m^\star \right) \right]
		+ \Tr\!\left[ A_{\bar m} \left( I - \Lambda_{\bar m}^\star \right) \right]
	\\
	&\geq \inf_{ \Lambda_m,\Lambda_{\bar m}\geq 0   } \left\{\Tr\!\left[ A_m \left( I - \Lambda_m \right) \right]
		+ \Tr\!\left[ A_{\bar m} \left( I - \Lambda_{\bar m} \right) \right]: \Lambda_m+\Lambda_{\bar m} = I - \sum_{\tilde{m}\neq m, \bar{m}} \Lambda_{\tilde{m}}\right\}
	\\
	&\geq \inf_{ \Lambda_m,\Lambda_{\bar m}\geq 0  } \left\{\Tr\!\left[ A_m \left( I - \Lambda_m \right) \right]
	+ \Tr\!\left[ A_{\bar m} \left( I - \Lambda_{\bar m} \right) \right]: \Lambda_m+\Lambda_{\bar m} = I\right\}
	\\
	&= (p_m+p_{\bar m}) \inf_{ \substack{\Lambda_m,\Lambda_{\bar m}\geq 0\\ \Lambda_m+\Lambda_{\bar m} = I}  } \Tr\!\left[ \frac{p_m}{p_m+p_{\bar m}} \rho_m \left( I - \Lambda_m \right) \right]
	+ \Tr\!\left[ \frac{p_{\bar m}}{p_m+p_{\bar m}} \rho_{\bar m} \left( I - \Lambda_{\bar m} \right) \right]
	\\
	&\geq (p_m+p_{\bar m}) \frac{p_m p_{\bar m}}{(p_m+p_{\bar m})^2} F(\rho_m,\rho_{\bar m}).
	\end{align}
To see the third inequality, let $\Lambda_m,\Lambda_{\bar m}\geq 0$ be arbitrary operators satisfying $\Lambda_m+\Lambda_{\bar m} = I - \sum_{\tilde{m}\neq m, \bar{m}} \Lambda_{\tilde{m}}$. Then choose $X,Y\geq 0$ such that $X+Y =  I - (\Lambda_m+\Lambda_{\bar m})$. Now set $\Lambda_m' \coloneqq \Lambda_m + X$ and $\Lambda'_{\bar{m}} \coloneqq \Lambda_{\bar{m}} + Y$. So then $\Lambda_m', \Lambda'_{\bar{m}}\geq 0$ and $\Lambda_m'+ \Lambda'_{\bar{m}} = I$, and thus the choices $\Lambda_m'$ and $\Lambda'_{\bar{m}}$ are feasible for the optimization in the third line. Furthermore, we have that $\operatorname{Tr}[(I-\Lambda_m)\omega] \geq \operatorname{Tr}[(I-\Lambda'_m)\omega]$ and $\operatorname{Tr}[(I-\Lambda_{\bar{m}})\omega] \geq \operatorname{Tr}[(I-\Lambda'_{\bar{m}})\omega]$ for every $\omega \geq 0$. Thus, the objective function in the third line is never larger than that in the second line, thus justifying the third inequality.
	The last line follows from the converse (i.e.,~the lower bound of the sample complexity) of binary hypothesis testing, as shown in~\eqref{eq:lower_bound_symmetric_binary} in the proof of  Theorem~\ref{theorem:binary_symmetric}.

	Hence, by using the multiplicativity of the fidelity again, we obtain
	\begin{align}
		n^*(\mathcal{S},\varepsilon) \geq \max_{m\neq \bar{m}} \frac{ \ln\!\left( \frac{p_m p_{\bar m}}{ (p_m + p_{\bar m})\varepsilon } \right) }{ -\ln F(\rho_m, \rho_{\bar m }) },
	\end{align}
	concluding the proof.

\begin{remark}
The key ingredient of proving the upper bound in Theorem~\ref{theorem:sc_M-ary} is the subadditivity of the square-root of the fidelity (Lemma~\ref{lemma:Fact}-\ref{fact:subadditivity}).
However, the Holevo fidelity generally does not satisfy such a property.
This is why we could not directly achieve the Holevo fidelity in~\eqref{eq:M-ary-last-upper-bound}.
On the other hand, the result mentioned in Remark~\ref{rem:Ke_Li} directly yields the Holevo fidelity.
Note that one could still achieve the Holevo fidelity by losing a factor of $2$ (Lemma~\ref{lemma:Fact}-\ref{fact:relation_fidelity}), i.e.,
\begin{equation}
\frac{1}{-\ln F(\rho,\sigma) } \leq \frac{2}{ - \ln F_{\mathrm{H}}(\rho,\sigma)}.
\end{equation}
\end{remark}

\section{Proof of Remark~\ref{rem:Ke_Li}} 

\label{app:proof:rem:Ke_Li}

Let $\mathcal{S}^{(n)} \coloneqq  \left\{(p_m, \rho_m^{\otimes n})\right\}_{m=1}^M$.
	The following bound was established as~\cite[Eq.~(36)]{Li16}, as a step in the proof of the multiple quantum Chernoff exponent:
	\begin{align}
 \label{eq:KeLiBound}
		p_e(\mathcal{S}^{(n)}) 
		\leq 5 M (M-1)^3 (n+1)^{2d} \max_{1\leq m \leq M}\{p_m\} \max_{m\neq \bar{m}}  \mathrm{e}^{ - n C(\rho_m\Vert \rho_{\bar{m}}) },
	\end{align}	
	where $d$ is the dimension of the underlying Hilbert space.
Denote
\begin{align}
    c \coloneqq 5 M (M-1)^3 \max_{1\leq m \leq M}\{p_m\} \quad\text{and}\quad B\coloneqq\min_{m \neq \bar{m}} \frac{C(\rho_m\Vert\rho_{\bar{m}})}{2}.
\end{align}
Note that the function $f(x) = C(x+1)^{2d}e^{-Bx}$ is maximised at $x_0= \frac{2d}{B}-1$ with maximal value 
\begin{align}
A&\coloneqq f(x_0) = c\left(\frac{2d}{B}\right)^de^{B-2d} \\
&= 5 M (M-1)^3 \max_{1\leq m \leq M}\{p_m\}\left(\frac{4d}{\min_{m \neq \bar{m}} C(\rho_m\Vert \rho_{\bar{m}})}\right)^d \mathrm{e}^{\min_{m \neq \bar{m}} \frac{C(\rho_m\Vert\rho_{\bar{m}})}{2}-2d}
\end{align}
From~\eqref{eq:KeLiBound} we therefore see
\begin{align}
    p_e(\mathcal{S}^{(n)}) \le c (n+1)^{2d} \mathrm{e}^{-2Bn}\le A \mathrm{e}^{-Bn} 
\end{align}
which hence gives
\begin{align}
    n^*(\mathcal{S},\varepsilon) \le \left\lceil\frac{\ln\left(A/\varepsilon\right)}{\min_{m \neq \bar{m}} \frac{C(\rho_m\Vert\rho_{\bar{m}})}{2}}\right\rceil.
\end{align}
Invoking $C(\rho_m\Vert\rho_{\bar{m}}) \geq - \frac{1}{2} \ln F(\rho_m,\rho_{\bar{m}})$
(Lemma~\ref{lemma:Fact}-\ref{fact:relation_fidelity}), the above bound also implies 
\begin{align}
    n^*(\mathcal{S},\varepsilon)
    \leq  O\!\left( \max_{m \neq \bar{m}}\frac{\ln M}{ -\ln F(\rho_m,\rho_{\bar{m}})} \right).
\end{align}

\end{document}